\documentclass[3p,number,times,fleqn,preprint]{elsarticle}

\usepackage{enumitem}
\usepackage{latexsym}
\usepackage{amsmath,amssymb}
\usepackage{pifont} 
\usepackage{hyperref}
\usepackage{booktabs}
\usepackage{multirow}

\usepackage{xcolor}

\usepackage{tikz}
\usetikzlibrary{shapes,arrows,positioning,calc,automata,patterns,arrows.meta}

\newtheorem{theorem}{Theorem}
\newtheorem{lemma}[theorem]{Lemma}
\newtheorem{corollary}[theorem]{Corollary}

\newtheorem{example}[theorem]{Example}
\newtheorem{remark}[theorem]{Remark}
\newtheorem{definition}[theorem]{Definition}

\renewcommand{\qed}{\hfill\ding{113}}

\newproof{proof}{Proof}

\newcommand{\nb}[1]{}
\newcommand{\nz}[1]{}

\newcommand{\LTL}{\textsl{LTL}}
\newcommand{\DL}{\textsl{DL-Lite}}
\newcommand{\EL}{\mathcal{EL}}
\newcommand{\OWLQL}{\textsl{OWL\,2\,QL}}

\newcommand{\LTLpr}{\textsl{Prior-LTL}}
\newcommand{\LTLkr}{\textsl{Krom-LTL}}
\newcommand{\LTLho}{\textsl{Horn-LTL}}

\newcommand{\Xallop}{^{\smash{\Box\raisebox{1pt}{$\scriptscriptstyle\bigcirc$}}}}
\newcommand{\Xnext}{^{\smash{\raisebox{1pt}{$\scriptscriptstyle\bigcirc$}}}}
\newcommand{\Xbox}{^{\smash{\Box}}}

\newcommand{\NLogSpace}{\textsc{NLogSpace}}
\newcommand{\LogSpace}{\textsc{LogSpace}}
\newcommand{\NP}{\textsc{NP}}
\newcommand{\NL}{\textsc{NL}}
\newcommand{\coNP}{\textsc{coNP}}
\newcommand{\PSpace}{\textsc{PSpace}}
\newcommand{\ExpSpace}{\textsc{ExpSpace}}
\newcommand{\NCo}{{{\ensuremath{\textsc{NC}^1}}}}
\newcommand{\ACz}{{\ensuremath{\textsc{AC}^0}}}
\newcommand{\LogTime}{\textsc{LogTime}}
\newcommand{\ExpTime}{\textsc{ExpTime}}

\newcommand{\PTime}{\textsc{P}}

\newcommand{\FO}{\textup{FO}}
\newcommand{\MSO}{\textup{MSO}}
\newcommand{\MFO}{\textup{MFO}}
\newcommand{\FOE}{\textup{FO}(<,\equiv)}
\newcommand{\FOEZ}{\smash{\textup{FO}^{\mathbb Z}}(<,\equiv)}

\newcommand{\bool}{\textit{bool}}
\newcommand{\horn}{\textit{horn}}

\newcommand{\krom}{\textit{krom}}
\newcommand{\core}{\textit{core}}

\newcommand{\FF}{{\scriptscriptstyle F}}
\newcommand{\PP}{{\scriptscriptstyle P}}

\newcommand{\nm}[1]{\textit{#1}}

\newcommand{\nxt}{{\ensuremath\raisebox{0.25ex}{\text{\scriptsize$\bigcirc$}}}}
\newcommand{\Rnext}{\nxt_{\!\scriptscriptstyle F}}
\newcommand{\Lnext}{\nxt_{\!\scriptscriptstyle P}}

\newcommand{\Rdiamond}{\Diamond_{\!\scriptscriptstyle F}}
\newcommand{\Ldiamond}{\Diamond_{\!\scriptscriptstyle P}}

\newcommand{\Rbox}{\rule{0pt}{1.4ex}\Box_{\!\scriptscriptstyle F}}
\newcommand{\Lbox}{\rule{0pt}{1.4ex}\Box_{\!\scriptscriptstyle P}}

\newcommand{\Si}{\mathbin{\mathcal{S}}}
\newcommand{\U}{\mathbin{\mathcal{U}}}

\newcommand{\Z}{\mathbb{Z}}
\newcommand{\N}{\mathbb{N}}

\newcommand{\Mmf}{\mathfrak{M}}
\newcommand{\A}{\mathcal{A}}
\newcommand{\TO}{\mathcal{O}}
\newcommand{\C}{\mathfrak{C}}
\newcommand{\K}{\mathcal{K}}

\newcommand{\can}{{\C_{\TO,\A}}}

\newcommand{\q}{\boldsymbol q}
\renewcommand{\rq}{\boldsymbol Q}

\newcommand{\len}{\boldsymbol l}

\newcommand{\tem}{\mathsf{tem}}
\newcommand{\ans}{\mathsf{ans}}
\newcommand{\cl}{\mathsf{cl}}
\newcommand{\subq}{\mathsf{sub}_{\q}}
\newcommand{\subo}{\mathsf{sub}_{\TO}}

\newcommand{\nfao}{\mathfrak{A}_{\TO}}

\newcommand{\frag}{{\boldsymbol{c}}}
\newcommand{\op}{{\boldsymbol{o}}}

\newcommand{\OMIQ}{OMIQ}
\newcommand{\OMPIQ}{OMPIQ}
\newcommand{\OMQ}{OMQ}
\newcommand{\OMAQ}{OMAQ}

\newcommand{\avec}[1]{\boldsymbol{#1}}
\newcommand{\SA}{\mathfrak S_\A}
\newcommand{\GOA}{\mathfrak G_{\TO,\A}}

\newcommand{\type}{\tau}

\newcommand{\parity}{\textsc{Parity}}
\newcommand{\plus}{\textsc{plus}}
\newcommand{\TIMES}{\textsc{times}}

\newcommand{\kpar}{k}

\newcommand{\lang}{\mathcal{L}}

\newcommand{\tp}{\tau}

\newcommand{\tpset}{\mathfrak{T}}

\newcommand{\pathf}{\textsf{path}}
\newcommand{\typef}{\textsf{type}}
\newcommand{\suc}{\textsf{suc}}
\newcommand{\entailsf}{\textsf{entails}}
\newcommand{\intervalf}{\textsf{interval}}

\newcommand{\RPR}{\textup{RPR}}

\begin{document}

\begin{frontmatter}



\title{First-Order Rewritability of Ontology-Mediated Queries in\\ Linear Temporal Logic}


\author[bz]{Alessandro Artale}
\ead{artale@inf.unibz.it}

\author[bbk]{Roman Kontchakov}
\ead{roman@dcs.bbk.ac.uk}

\author[dr]{Alisa Kovtunova}
\ead{alisa.kovtunova@tu-dresden.de}

\author[bbk]{Vladislav Ryzhikov}
\ead{vlad@dcs.bbk.ac.uk}

\author[liv]{Frank Wolter}
\ead{wolter@liverpool.ac.uk}

\author[bbk,hse]{Michael Zakharyaschev}
\ead{michael@dcs.bbk.ac.uk}

\address[bz]{KRDB Research Centre, Free University of Bozen-Bolzano, Italy}
\address[bbk]{Department of Computer Science and Information Systems, Birkbeck, University of London, UK}
\address[dr]{Chair for Automata Theory, Technische Universit\"at Dresden, Germany}
\address[liv]{Department of Computer Science, University of Liverpool, UK}
\address[hse]{HSE University, Moscow, Russia}


\begin{abstract}
We investigate ontology-based data access to temporal data.
We consider temporal ontologies given in linear temporal logic \LTL{} interpreted
over discrete time $(\Z,<)$. Queries are given in \LTL{} or $\MFO(<)$, monadic first-order logic with a built-in linear order. Our concern is first-order rewritability of ontology-mediated queries (OMQs) consisting of a temporal ontology and a query.
By taking account of the temporal operators used in the ontology
and distinguishing between ontologies given in full \LTL{} and its
core, Krom and Horn fragments, we identify a hierarchy of
OMQs with atomic queries by proving rewritability into either
$\FO(<)$, first-order logic with the built-in linear order, or
$\FOE$, which extends $\FO(<)$ with the standard arithmetic predicates $x \equiv 0\, (\text{mod}\ n)$, for any fixed $n > 1$, or
$\FO(\RPR)$, which extends $\FO(<)$ with relational primitive recursion.
In terms of circuit complexity, $\FOE$- and $\FO(\RPR)$-rewritability guarantee OMQ answering in uniform \ACz{} and, respectively, \NCo.

We
obtain similar hierarchies for more expressive types of queries: positive \LTL{}-formulas, monotone \mbox{$\MFO(<)$-} and arbitrary $\MFO(<)$-formulas. Our results are directly applicable if the temporal data to be accessed is one-dimensional; moreover, they lay  foundations for investigating ontology-based access using combinations of temporal and description logics over
two-dimensional temporal data.
\end{abstract}



\end{frontmatter}



\section{Introduction}
\label{intro}

Ontology-mediated query answering has recently become one of the most successful applications of description logics (DLs) and semantic technologies.
Its main aim is to facilitate user-friendly access to possibly heterogeneous, distributed and incomplete data.
To this end, an ontology is employed to provide (\emph{i}) a convenient and uniform vocabulary for formulating queries and (\emph{ii}) a conceptual model of the domain for capturing background knowledge and obtaining more complete answers. Thus, instead of querying data directly by means of convoluted database queries, one can use ontology-mediated queries (OMQs, for short) of the form $\q=(\TO,\varphi)$, where $\TO$ is an ontology and $\varphi$ a query formulated in the vocabulary of $\TO$. Under the standard certain answer semantics for OMQs, the answers to~$\q$ over a data instance~$\A$ are exactly those tuples of individual names from $\A$ that satisfy $\varphi$ in every model of $\TO$ and~$\A$. Because of this open-world semantics, answering OMQs can be computationally much harder than evaluating standard database queries. For example, answering an atomic query $A(x)$ using an ontology in the standard description logic $\mathcal{ALC}$ can be \coNP-hard for data complexity---the complexity measure (adopted in this paper) that regards the OMQ as fixed and the data instance as the only input to the OMQ answering problem. For this reason, weaker description logics (DLs) have been developed to enable not only tractable OMQ answering but even a reduction of OMQ answering to evaluation of standard relational database queries directly over the data, which is in \ACz{} for data complexity. In fact, the popular and very successful \DL{} family of DLs was designed so as to ensure rewritability of OMQs with conjunctive queries into first-order logic (FO) queries, and so to SQL. \DL{} underpins the W3C standard ontology language \OWLQL~\cite{DBLP:journals/jar/CalvaneseGLLR07,DBLP:journals/jair/ArtaleCKZ09}. For applications of OMQ answering with \OWLQL{}, the reader is referred to~\cite{DBLP:conf/fois/AntonioliCCGLLPVC14,FishMarkPaper,DBLP:journals/computer/GieseSVWHJLRXOR15,DBLP:journals/semweb/CalvaneseGLLPRRRS11,DBLP:journals/eaai/CalvaneseLMRRR16,DBLP:conf/semweb/Rodriguez-MuroKZ13,DBLP:journals/internet/SequedaM17,DBLP:conf/semweb/HovlandKSWZ17}; for a recent survey, consult~\cite{IJCAI-18}.

\DL{} and \OWLQL{} were designed to represent knowledge about \emph{static} domains and are not suitable when the data and the vocabulary the user is interested in are essentially temporal. To extend OMQ answering to temporal domains, the ontology language needs to be extended by various temporal constructs studied in the context of temporal representation and reasoning~\cite{Gabbayetal94,gkwz,DBLP:books/cu/Demri2016}. In fact, combinations of DLs with
temporal formalisms have been widely investigated since the pioneering work of Schmiedel~\cite{DBLP:conf/aaai/Schmiedel90} and Schild~\cite{DBLP:conf/epia/Schild93} in the early 1990s; we refer the reader to~\cite{gkwz,Baader:2003:EDL:885746.885753,DBLP:reference/fai/ArtaleF05,DBLP:conf/time/LutzWZ08} for surveys and \cite{DBLP:conf/dlog/PagliarecciST13,DBLP:journals/tocl/ArtaleKRZ14,DBLP:conf/kr/Gutierrez-BasultoJ014,DBLP:conf/ijcai/Gutierrez-Basulto15,DBLP:conf/ecai/Gutierrez-Basulto16,DBLP:conf/dlog/BaaderBKOT17} for more recent developments. However, the main reasoning task
targeted in this line of research has been knowledge base satisfiability rather than OMQ answering, with the general aim of probing various combinations of temporal and DL constructs that ensure decidability of satisfiability with acceptable combined complexity (which is the complexity measure that regards both the ontology and data instance as input).

Motivated by the success of \DL{} and the paradigm of FO-rewritability
in OMQ answering over static domains, our \emph{ultimate aim} is the study of FO-rewritability of OMQs with temporal constructs in both ontologies and queries
over temporal databases. To lay the foundations for this project, in this article
we consider the basic scenario of querying timestamped `propositional' data in a synchronous system with a centralised clock. We thus do not yet consider general
temporal relational data but focus on `non-relational' pure temporal data.
We use the standard discrete time model with the (positive and negative)
integers~$\Z$ and the order $<$ as precedence relation. The most basic
and fundamental temporal language for the discrete time model is the linear
temporal logic \LTL{} with the temporal operators
$\Rnext$ (at the next moment of time), $\Rdiamond$ (eventually), $\Rbox$ (always in the future),
$\mathcal{U}$ (until), and their past-time counterparts $\Lnext$ (at the previous moment),
$\Ldiamond$ (some time in the past), $\Lbox$ (always in the past) and $\mathcal{S}$ (since); see~\cite{DBLP:books/daglib/0077033,Gabbayetal94,DBLP:books/cu/Demri2016} and references therein. \LTL{} and its fragments are particularly
natural for our study of FO-rewritability as, by the celebrated Kamp's Theorem, \LTL{} is expressively complete in the sense that anything that can be said in $\MFO(<)$, monadic first-order logic with the precedence relation $<$ over the discrete (in fact, any Dedekind complete) model of time, and with reference to a single time point can also be  expressed in \LTL{}~\cite{phd-kamp,DBLP:journals/corr/Rabinovich14}.

Thus, in this article, we conduct an in-depth study of FO-rewritability and data complexity of OMQs with  ontologies formulated in fragments of (propositional) \LTL{} and queries given as \LTL- or $\MFO(<)$-formulas, assuming that (\emph{i})~ontology axioms hold at all times and (\emph{ii}) data instances are finite sets of facts
of the form $A(a,\ell)$ saying that $A$ is true of the individual $a$
at the time instant $\ell \in \Z$.
To illustrate, suppose that data instances contain facts about the status of a research article submitted to a journal using predicates for the events $\nm{Submission}$, $\nm{Notification}$,
$\nm{Accept}$, $\nm{Reject}$, $\nm{Revise}$ and $\nm{Publication}$, and that temporal domain knowledge about these events is formulated in an ontology~$\TO$ as follows:
\begin{equation}\label{pub1}
\nm{Notification} \leftrightarrow \nm{Reject} \vee \nm{Accept} \vee \nm{Revise}
\end{equation}
states that, at any moment of time, every notification is either a reject, accept or revision notification, and that it can only be one of them:
\begin{equation}\label{pub2}
\nm{Reject}\land \nm{Accept} \to \bot, \qquad
\nm{Revise}\land \nm{Accept} \to \bot, \qquad
\nm{Reject}\land \nm{Revise} \to \bot.
\end{equation}
The ontology $\TO$ says that any event $P$ above, except $\nm{Notification}$
and $\nm{Revise}$, can happen only once for any article:
\begin{equation}\label{pub3}
P \to \neg\Ldiamond P \land \neg\Rdiamond P.
\end{equation}
It contains obvious necessary preconditions for publication and notification:
\begin{equation}\label{pub4}
\nm{Publication} \to \Ldiamond \nm{Accept}, \qquad
\nm{Notification} \to \Ldiamond \nm{Submission},
\end{equation}
and also the post-conditions (eventual consequences) of acceptance, submission and a revision notification (for simplicity, we assume that after a revision notification the authors always eventually receive a notification regarding a revised version):
\begin{equation}\label{pub5}
\nm{Accept} \to \Rdiamond \nm{Publication}, \qquad
\nm{Submission} \to \Rdiamond \nm{Notification}, \qquad
\nm{Revise} \to \Rdiamond \nm{Notification}.
\end{equation}
Finally, the ontology $\TO$ states that acceptance and rejection notifications are final:
\begin{equation}\label{pub6}
	\nm{Accept} \lor  \nm{Reject} \to \neg\Rdiamond \nm{Notification}.
\end{equation}
Consider now the following set $\A$ of timestamped facts:
\begin{equation*}
\nm{Notification}(a,\,\text{Oct2017}), \qquad \nm{Revise}(a,\,\text{Oct2019}), \qquad \nm{Publication}(a,\,\text{Dec2019}).
\end{equation*}
Thus, according to $\A$, the authors received a notification about their article
$a$ in October 2017,
they received a revise notification about $a$ in October 2019, and
article $a$ was published in December 2019.
In the context of this example, it is natural to identify months in consecutive years with moments of time in the discrete time model $(\Z,<)$: for example, October 2017 is 0, October~2019 is 24,  etc.

To illustrate OMQ answering, consider first the (atomic) \LTL-formula~$\varphi_1 = \nm{Revise}$. By \eqref{pub1} and \eqref{pub6}, the pairs $(a,\,\text{Oct2017})$ and $(a,\,\text{Oct2019})$ are the certain answers to the OMQ $\q_1=(\TO,\varphi_1)$; see Fig.~\ref{fig:submission-cycle}. Now, consider the \LTL-formula~$\varphi_2=\Ldiamond \nm{Submission}$, which we understand as a query asking for months when the article had been previously submitted. Despite the fact that $\A$ contains no facts about  predicate $\nm{Submission}$, the axioms~\eqref{pub4} in $\TO$ imply that $\varphi_2$ is true for $a$ at all points of the infinite interval~\mbox{$[\text{Oct2017},+\infty)$} in every model of $\TO$ and~$\A$, and so all pairs $(a,n)$ with $n\in [\text{Oct2017},+\infty)$ might be regarded as certain answers to the OMQ~\mbox{$\q_2=(\TO,\varphi_2)$}. However, as usual in database theory, we are only interested in finitely many answers from the \emph{active domain}, which can be defined as the smallest convex subset of $\Z$ containing all the timestamps from the data instance. Thus, the certain answers to $\q_2$ over~$\A$ are the pairs $(a,n)$ with $n$ in the closed interval $[\text{Oct2017}, \text{Dec2019}]$, which is shown in Fig.~\ref{fig:submission-cycle} by shading.

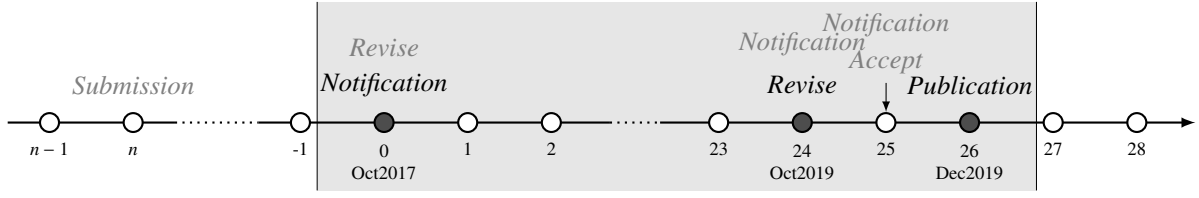
\begin{figure}
\centering%
\begin{tikzpicture}[nd/.style={draw,thick,circle,inner sep=0pt,minimum size=2.5mm,fill=white},xscale=1.1,>=latex]
\fill[gray!20] (4.2,-0.9) rectangle +(8.6,2.5);
\draw[ultra thin] (4.2,-0.9) -- ++(0,2.5);
\draw[ultra thin] (12.8,-0.9) -- ++(0,2.5);
\draw[thick] (0.5,0) -- (2.5,0);
\draw[thick,dotted] (2.5,0) -- (3.5,0);
\draw[thick] (3.5,0) -- (7.7,0);
\draw[thick,dotted] (7.7,0) -- (8.3,0);
\draw[thick,->] (8.3,0) -- (14.7,0);
\node[nd,label=below:{\scriptsize $n-1$}] at (1,0) {};
\node[nd,label=below:{\scriptsize \phantom{$1$}$n$\phantom{$1$}}] at (2,0) {};
\node[nd,label=below:{\scriptsize -1}] at (4,0) {};
\node at (2,0.5) {\textcolor{gray}{$\nm{Submission}$}};
\node[nd,fill=black!70,label=below:{\scriptsize \begin{tabular}{c}0\\Oct2017\end{tabular}}] at (5,0) {};
\node[nd,label=below:{\scriptsize 1}] at (6,0) {};
\node[nd,label=below:{\scriptsize 2}] at (7,0) {};
\node[nd,label=below:{\scriptsize 23}] at (9,0) {};
\node[nd,fill=black!70,label=below:{\scriptsize \begin{tabular}{c}24\\Oct2019\end{tabular}}] at (10,0) {};
\node[nd,label=below:{\scriptsize 25}] (a0) at (11,0) {};
\node at (5,0.5) {$\nm{Notification}$};
\node at (5,1) {\textcolor{gray}{$\nm{Revise}$}};
\node[nd,fill=black!70,label=below:{\scriptsize \begin{tabular}{c}26\\Dec2019\end{tabular}}] at (12,0) {};
\node[nd,label=below:{\scriptsize 27}] at (13,0) {};
\node[nd,label=below:{\scriptsize 28}] at (14,0) {};
\node at (10,0.5) {$\nm{Revise}$};
\node at (10,1.05) {\textcolor{gray}{$\nm{Notification}$}};
\node at (12,0.5) {$\nm{Publication}$};
\node (a) at (11,0.8) {\textcolor{gray}{$\nm{Accept}$}};
\node at (11,1.3) {\textcolor{gray}{$\nm{Notification}$}};
\draw[->,thin] (a) -- (a0);
\end{tikzpicture}
\caption{An example illustrating models of $\A$ and $\TO$ (implied information is shown in grey).}\label{fig:submission-cycle}
\end{figure}

First-order rewritability makes it possible to find certain answers to such OMQs without
ontology reasoning, simply by evaluating their FO-rewritings directly over the data instance (with built-in predicates such as $<$ interpreted over the active domain). For the two OMQs above, FO-rewritings are as follows:
the FO-formula
\begin{equation*}
\rq_1(x,t) ~=~ \nm{Revise}(x,t) \ \ \vee \ \ \bigl(\nm{Notification}(x,t) \wedge \exists s\,\bigl((t<s) \wedge N(x,s)\bigr)\bigr),
\end{equation*}
where
\begin{equation}\label{auxP}
N(x,s) = \nm{Notification}(x,s) \vee \nm{Accept}(x,s) \vee \nm{Reject}(x,s) \vee \nm{Revise}(x,s),
\end{equation}
is an FO-rewriting of $\q_{1}$ in the sense that, for any given data instance $\A$, the pair $(a,n)$ with an  article $a$ in $\A$ and a moment of time $n$ from the active domain of~$\A$  is a certain answer to $\q_1$ over $\A$ iff $\rq_1(a,n)$ holds true in the finite FO-structure given by $\A$. Similarly,
\begin{equation*}
\rq_2(x,t) ~=~ \exists s \,\bigl((s<t) \wedge \nm{Submission}(x,s)\bigr) \  \vee  \ \exists s\, \bigl((s \leq t) \land N(x,s)\bigr) \  \vee \ \exists s\,\bigl([(s \leq t) \lor (s = t + 1)] \land \nm{Publication}(x,s)\bigr),
\end{equation*}
where $s \leq t$ abbreviates $\neg(t < s)$, $s = (t + 1)$ stands for $(t < s) \land \neg \exists s' \,\bigl((t < s') \land (s' < s)\bigr)$ and  $N(x,s)$ is given again by~\eqref{auxP},
is an FO-rewriting of $\q_{2}$. (Note that we use the `strict' semantics for $\Ldiamond$ and other temporal operators, which does not include the current point.)
Observe that, in addition to the symbols from the OMQs, these rewritings use the precedence relation $<$, but no arithmetic operations. Also note that we always
use a single variable, $x$, as the answer variable for individuals, and
that we do not quantify over individuals. Thus, we can (and will) drop $x$
from all  rewritings and identify this target language for rewritings with $\FO(<)$, first-order logic with a built-in linear order.

We identify expressive fragments of \LTL{} together with query languages for which \emph{every} OMQ is rewritable into $\FO(<)$. To illustrate one of our main results,
let $\LTLpr$
denote the fragment of $\LTL{}$ containing all formulas constructed using arbitrary Boolean connectives and temporal operators $\Lbox$, $\Ldiamond$, $\Rbox$ and $\Rdiamond$. This language was one of the first temporal logics developed and studied in philosophical logic, going back to Prior~\cite{prior:1956b,ono1980on,Burgess84,DBLP:conf/spin/Vardi08}.
The ontology $\TO$ above is given in $\LTLpr$. We obtain the following rewritability result:

\medskip
\noindent
\textbf{Theorem A.}
\emph{All \OMQ{}s with ontologies and queries in
$\LTLpr$ are $\FO(<)$-rewritable.}

\medskip

It follows that evaluating any such OMQ is in \ACz{} for data complexity, and that
 evaluating them can be delegated to standard relational database management systems.
 Not all OMQs are $\FO(<)$-rewritable. In particular, even very simple Horn ontologies with operators $\Rnext$ and $\Lnext$ give rise to OMQs that fall outside the scope of $\FO(<)$. Consider, for example, the ontology $\TO_2$ with two axioms
\begin{equation*}
\nm{Even} \to \Rnext\nm{Odd}, \qquad \nm{Odd}\to \Rnext\nm{Even}
\end{equation*}
saying, that every `even' time point is always followed by an `odd' one, and the other way round. Using the fact that the set of even numbers cannot be defined in $\FO(<)$~\cite{Straubing94,Libkin},
one can show that the OMQ $(\TO_{2},\nm{Even})$ has no rewriting in $\FO(<)$.  The even numbers can, however, be defined in the extension $\FOE$ of $\FO(<)$ with the standard numeric predicates $t \equiv 0 \pmod n$, for any fixed~$n > 1$. We identify fragments of \LTL{} and query languages for which every OMQ is rewritable in $\FOE$. As a main example, let $\LTLkr$ be the fragment of \LTL{} with the unary temporal operators $\Rnext$, $\Lnext$, $\Rbox$, $\Lbox$, $\Rdiamond$ and $\Ldiamond$, in which all axioms are \emph{binary}, that is, contain at most two predicates (as in $\nm{Even} \to \Rnext\nm{Odd}$ or $\Rnext A \lor \Lnext B$). We show the following:

\medskip
\noindent
\textbf{Theorem B.}
\emph{All OMQs with $\LTLkr$ ontologies and atomic queries
are $\FOE$-rewritable.}

\medskip

Note that any such OMQ is still in \ACz{} for data complexity and that,
 because of the standard support for basic arithmetic, evaluating them can be delegated to relational database management systems.
 Not all OMQs are \mbox{$\FOE$}-rewritable, however. Non-binary Horn ontologies with axioms such as $\Rnext B \land A \to C$ can express parity (the number of time points $n$ with $A(n)$ in the active domain is even), which cannot be defined in $\FO(<)$ extended with arbitrary arithmetic predicates~\cite{DBLP:journals/mst/FurstSS84}.

Thus, our fine-grained classification of rewritability of \LTL{} OMQs distinguishes between the different temporal operators that can occur in ontology axioms and also takes account of the Boolean (non-temporal) structure of the axioms by distinguishing between the core, Krom, Horn and full Boolean fragments\footnote{Such fragments have also proved to be useful for studying FO-rewritability and data complexity of OMQs in the \DL{} family of description logics~\cite{DBLP:journals/jar/CalvaneseGLLR07,DBLP:journals/jair/ArtaleCKZ09}.} of \LTL. To provide a systematic analysis, it is useful to work with fragments of \LTL{} given in the clausal normal form
\begin{equation}\label{normal-dl}
C_1 \land \dots \land C_k ~\to~ C_{k+1} \lor \dots \lor C_{k+m},
\end{equation}
where the $C_i$ are predicate names, possibly prefixed with operators $\Rnext$, $\Lnext$, $\Rbox$ and $\Lbox$. Suppose that $\op \in \{\Box,\nxt,\Box\nxt\}$ and $\frag \in \{\bool, \horn, \krom,\core\}$. We denote by $\LTL_{\frag}^{\op}$ the temporal logic with clauses of the form~\eqref{normal-dl}, where the $C_i$ can only use the (future and past) operators indicated in $\op$, and $m\leq 1$ if $\frag = \horn$; $k + m\leq 2$ if $\frag = \krom$; $k + m\leq 2$ and $m \leq 1$ if $\frag = \core$; and arbitrary $k$, $m$ if $\frag = \bool$.
It follows from~\cite{DBLP:journals/logcom/Fisher97} that every \LTL-ontology can be converted (possibly with a linear blowup and by introducing fresh predicates)
in a canonical way into clausal normal form giving the same answers to queries as the original one. Observe that $\LTLpr$ and $\LTLkr$ ontologies as introduced above can be converted into $\LTL_{\bool}\Xbox$ and, respectively, $\LTL_{\krom}\Xallop$ ontologies.

We consider the following hierarchy of queries in our \OMQ{}s. \emph{Atomic} \OMQ{}s
(\OMAQ{}s, for short) use queries of the form $A(t)$ with $A$ a predicate.\!\footnote{ Recall that we do not give the variable $x$ ranging over the individuals in the database.}
\OMQ{}s using arbitrary \LTL{} formulas as queries are called ontology-mediated \emph{instance} queries (or OMIQs), and \OMQ{}s using only \emph{positive} \LTL{}
formulas as queries are called ontology-mediated  \emph{positive} instance
queries (or OMPIQs). The queries introduced so far have exactly one (implicit)
answer variable. To generalise our results to queries of arbitrary arity, we
introduce the language of \OMQ{}s with any $\MFO(<)$-formulas playing the role of  queries, for example   
\begin{equation*}
\psi(t,t') ~=~ \nm{Revise}(t) \wedge \nm{Accept}(t') \wedge \forall s \, \big( (t < s < t') \rightarrow \neg \nm{Revise}(s)\big),
\end{equation*}
asking for all pairs $(t,t')$ such that $t$ is the last revision date before the acceptance date $t'$. We also consider \OMQ{}s with \emph{quasi-positive} $\MFO(<)$-formulas as
queries that are constructed using
$\wedge$, $\vee$, $\forall$, $\exists$, as well as \emph{guarded} $\forall$
such as, for example, in $\forall s\,((t < s < t') \to \varphi)$.  
We show that quasi-positive $\MFO(<)$-formulas capture exactly the \emph{monotone} $\MFO(<)$-formulas (that are  preserved under adding time points to the extension of predicates). We also show that \OMQ{}s with quasi-positive $\MFO(<)$-formulas as queries behave in exactly the same way as OMPIQs.
Our main result about these expressive queries is as follows:

\medskip
\noindent
\textbf{Theorem C.}
All OMQs with arbitrary $\MFO(<)$-queries are $\FO(\RPR)$-rewritable. All OMQ{}s with $\LTL_{\horn}\Xbox$ or $\LTL_{\core}\Xallop$ ontologies and quasi-positive $\MFO(<)$-queries are $\FO(<)$-rewritable or $\FOE$-rewritable, respectively. 

\medskip

We summarise our rewritability results in Table~\ref{LTL-table}.
It is to be noted that the $\FO(<)$-rewritability
result for \OMQ{}s with $\LTLpr$-ontologies and queries stated
in Theorem~A follows from the $\FO(<)$-rewritability
of all OMAQs using $\LTL_{\bool}\Xbox$ ontologies as answering a $\LTLpr$ query $C$ can be reduced to answering an atomic query $A$ by
adding the axiom $C \to A$ to the ontology. Using similar reductions,
one can also extend other query languages given in Table~\ref{LTL-table} in the obvious way.

As Table~\ref{LTL-table} shows, all of our OMQs are rewritable into the extension $\FO(\RPR)$  of $\FO(<)$ with relational primitive recursion~\cite{DBLP:journals/iandc/ComptonL90}.
This implies that answering them is in $\NCo \subseteq \LogSpace$ for data complexity (and thus can be performed by an efficient parallel algorithm~\cite{Arora&Barak09}); it also means that answering OMQs can be done using finite automata~\cite{Libkin}. 
In terms of circuit complexity, both $\FO(<)$- and $\FOE$-rewritability of a given OMQ mean that answering this OMQ is in \LogTime{}-uniform $\ACz$ for data complexity~\cite{Immerman99}.
Note that the SQL:1999 ISO standard contains a \textsc{with recursive} construct that can represent various FO-queries with relational primitive recursion such as the query in Example~\ref{exampleBvNB:2} below, which cannot be expressed in FO without recursion.

\begin{table}[t]
\centering%
\tabcolsep=8pt%
\begin{tabular}{ccccc}\toprule
     \rule[-3pt]{0pt}{12pt}
     &\multicolumn{2}{c}{\OMAQ{}s}
     & \multicolumn{2}{c}{OMPIQs / quasi-positive $\MFO(<)$-queries}
     \\
     \rule[-3pt]{0pt}{13pt}$\frag$
     & $\LTL_\frag\Xbox$ & $\LTL_\frag\Xnext$ {\footnotesize and} $\LTL_\frag\Xallop$ & $\LTL_\frag\Xbox$ &  $\LTL_\frag\Xnext$ {\footnotesize and} $\LTL_\frag\Xallop$
     \\\midrule%
     \bool & \multirow{4}{*}{$\FO(<)$ {\scriptsize [Th.~\ref{thm:box}]}} &   $\FO(\RPR)$, {\small \NCo-hard}   & \multirow{2}{*}{$\FO(\RPR)$, {\small \NCo-hard} {\scriptsize [Th.~\ref{ex:NFA}]}}  & \multirow{3}{*}{FO(RPR) {\scriptsize [Th.~\ref{fullLTL},~\ref{fullLTL-OMQ}]}, {\small \NCo-hard}}  \\\cmidrule[0pt](lr){3-3}
     \krom &  & $\FOE$ {\scriptsize [Ths.~\ref{thm:krom-next},~\ref{thm:krom-nextbox}]}   &&     \\\cmidrule(lr){4-4}
     \horn & & $\FO(\RPR)$, {\small \NCo-hard} {\scriptsize [Th.~\ref{ex:NFA}]}  & \multirow{2}{*}{$\FO(<)$ {\scriptsize [Th.~\ref{thm:LTLrewritability},~\ref{thm:rewritability-monFO}]}}  & \\\cmidrule(lr){5-5}
     \core & & $\FOE$ & & $\FOE$ {\scriptsize [Th.~\ref{thm:LTLrewritability},~\ref{thm:rewritability-monFO}]}\\\bottomrule
   \end{tabular}
\caption{Rewritability and data complexity of \LTL{} \OMQ{}s.}
\label{LTL-table}
\end{table}

 It is known that under the open-world semantics of OMQs, answers to queries
 containing negation are often rather uninformative. For example, if one uses the query $\nm{Revise} \wedge \Lbox \neg \nm{Notification}$ mediated by the publication ontology~$\TO$ above to retrieve
 the date of the \emph{first} revise notification for an article, then one will only
 receive an answer if the submission date is just one time instant before the first revision notification. In classical OMQ answering, a way to obtain more meaningful answers to FO-queries is to interpret negation under the epistemic semantics as proposed by~\citeauthor{CalvaneseGLLR07}~\cite{CalvaneseGLLR07}. Under this semantics, we regard $\neg A(t)$ as true if $A(t)$ is not entailed. (The same semantics is used in the standard query language SPARQL for RDF datasets when interpreted with OWL ontologies under the entailment regimes~\cite{GlimmOgbuji13}.)
 To illustrate, the query~$\varphi_3= \nm{Revise}(t) \wedge \Lbox \neg \nm{Notification}(t)$ will now return the time instant of the first revise notification in the database.
It is well known that extending the expressive power of a query language in this way typically does not lead to an increase in data complexity. We confirm that this is the case for the query languages considered in this article too.

The plan of the article is as follows. In Section~\ref{sec:tdl}, we introduce the syntax and semantics of our ontology and query languages and define the basic
notions that are required in the sequel.
In Section~\ref{sec:LTL}, we show that \OMQ{}s with arbitrary \LTL{} ontologies and queries   are $\FO(\RPR)$-rewritable (and so also rewritable to monadic second-order logic $\MSO(<)$) and in $\NCo{}$ for  data complexity. We also establish $\NCo{}$-hardness results. The next five sections are devoted to the proofs
of Theorems~A, B and C formulated above. In Section~\ref{sec:ltl-tomaq:stutter}, we use partially ordered automata~\cite{DBLP:conf/dlt/SchwentickTV01} to prove Theorem~A.
In Sections~\ref{sec:ltl-tomaq:unary} and \ref{sec:ltl-tomaq:kromboxnext},
we use unary automata~\cite{chrobak-ufa} to prove Theorem~B.
In Sections~\ref{sec:ltl-tomiq} and \ref{sec:ltl-tomiq:2}, we use canonical models
to prove Theorem~C for OMPIQs. We lift our results for \LTL{}-queries
to MFO($<$)-queries in Section~\ref{FO-OMQs}. This completes the proof of Theorem~C. We also show an analogue of Kamp's Theorem for monotone formulas. Finally, in Section~\ref{epistecic}, we briefly discuss the epistemic semantics
for temporal queries. We conclude with a summary of the obtained results and a discussion of future work.


\subsection{Related Work}

As mentioned above, our approach to ontology-mediated query answering over temporal
data is motivated by the success of the ontology-based data management paradigm for
atemporal data using description logics or rule-based languages~\cite{DBLP:journals/jar/CalvaneseGLLR07,CaliGL12,DBLP:journals/ai/BagetLMS11,DBLP:conf/rweb/BienvenuO15}.
We first discuss the relationship between our results for `propositional' temporal data and the results obtained over the past 15 years for the rewritability and data complexity of ontology-mediated querying using `atemporal' description logics.

For standard DLs such as $\mathcal{ALCHI}$, one can prove the following dichotomy for the data complexity of answering an OMQ with an atomic query: evaluating such an OMQ is either in $\ACz$ or \LogSpace-hard~\cite{DBLP:journals/tods/BienvenuCLW14}. Thus, research has focused on either
the combined or parameterised complexity of OMQs that are in $\ACz$
in data complexity~\cite{DBLP:conf/dlog/BienvenuKKRZ17,DBLP:conf/pods/BienvenuKKPRZ17,DBLP:journals/jacm/BienvenuKKPZ18} or
on classifying further the data complexity of OMQs that are known to be \LogSpace-hard~\cite{DBLP:journals/tods/BienvenuCLW14,DBLP:conf/ijcai/LutzS17}.
For propositional temporal data, the situation is rather different. Indeed,
as $\ACz \subsetneq \NCo \subseteq \LogSpace$~\cite{Straubing94,Arora&Barak09}, the complexity class
of many of our OMQs does not play any role in standard atemporal ontology-mediated query answering.

We remind the reader that, in \DL{} and standard extensions such as $\mathcal{ALCHI}$, there is no need to distinguish between different target languages for FO-rewritings.
In fact, it is known that an OMQ with an $\mathcal{ALCHI}$ ontology
and a union of conjunctive queries (UCQ) is FO-rewritable iff it
is UCQ-rewritable~\cite{DBLP:journals/tods/BienvenuCLW14,DBLP:conf/pods/HernichLPW17}. In contrast, for the temporal data considered in this article and for \LTL{} ontologies, we show that there is a difference between rewritability into the first-order languages $\FO(<)$ and $\FOE$.

Finally, we remind the reader that, in \DL{} and other DLs, negation in queries results
in non-tractable (often undecidable) query evaluation~\cite{DBLP:conf/icdt/Rosati07,DBLP:journals/ws/Gutierrez-Basulto15}.
This is in contrast to the temporal case, where even OMQs with
arbitrary $\FO(<)$-queries are in $\NCo$ for data complexity.

Our article is closely related to the large body of work in temporal logic and automata for finite and infinite words, in particular to the B\"uchi--Elgot--Trakhtenbrot Theorem~\cite{Buchi60,Elgot61,Trakh62}, according to which monadic second-order sentences over finite strict linear orders define exactly the class of regular languages. Our data complexity results rely on the investigation of regular languages in terms of circuit and descriptive complexity~\cite{DBLP:journals/jcss/Barrington89,DBLP:journals/iandc/ComptonL90,Straubing94,Immerman99}.
We also use a few
more recent results on partially ordered and unary finite automata~\cite{chrobak-ufa,DBLP:conf/dlt/SchwentickTV01,to-ufa}.
To study OMQs with MFO($<$)-queries, we employ Kamp's Theorem~\cite{phd-kamp,DBLP:journals/corr/Rabinovich14}, according to which $\FO(<)$-formulas with one free variable have the same expressive power as \LTL{} formulas over the integers (or any other Dedekind-complete linear order).

As discussed above,
combinations of ontology languages with temporal formalisms
have been widely investigated since the beginning of the 1990s.
We are not aware, however, of any previous work that focuses on the
temporal dimension only without assuming the presence of a non-propositional, relational, domain as well.
Indeed, the focus of existing work has been on
\emph{adding} a temporal dimension to an existing ontology language
rather than on investigating an existing temporal logic from the
viewpoint of ontology-mediated querying. It is the latter what we do
in this article. We believe that this is worthwhile because $(i)$ the single-dimensional temporal languages are of interest by themselves and $(ii)$ investigating them first  allows one to investigate combined languages based on a good understanding of the computational complexity and rewritability properties
of their temporal components.

As the work on combining ontology and temporal languages is closely related to our research project, we give a brief overview.  Until the early 2010s, the main reasoning tasks investigated for the resulting logics were concept subsumption (is an inclusion between concepts entailed by a temporal ontology?) and knowledge base satisfiability (does a knowledge base consisting of a temporal data instance and a temporal ontology have a model?). Query answering and its complexity were not on the research agenda yet. Discussing this work is beyond the scope of this article, and so in the following we concentrate on briefly summarising recent research on combining ontology languages and temporal formalisms with the aim of ontology-mediated query answering over temporal data.  We focus on the discrete-time point-based approach as this is the approach we consider in this article and as it is fundamental for any other temporal data models.  One can distinguish between formalisms in which temporal constructs are added to \emph{both} ontology and query languages (as in this article) and formalisms in which only the query language is temporalised while the ontology language is a standard atemporal language. We call the latter approach \emph{query-centric} as no temporal connectives are added to the ontology language. The advantage of keeping the ontology language atemporal is that increases in the complexity of query answering compared to the atemporal case can only be caused by the new temporal constructs in the queries. OMQ answering in this framework has been investigated in depth, in particular, for the query language consisting of all \emph{LTL-CQs} that are obtained from \LTL{} formulas by replacing occurrences of propositional variables by arbitrary conjunctive queries (CQs).  \citeauthor{DBLP:conf/cade/BaaderBL13}~\cite{DBLP:conf/cade/BaaderBL13,DBLP:journals/ws/BaaderBL15} analyse the data and combined complexity of answering \LTL-CQs with respect to $\mathcal{ALC}$ and $\mathcal{SHQ}$ ontologies with and without rigid concept and role names (whose interpretation does not change over time). \citeauthor{DBLP:conf/gcai/BorgwardtT15}~\cite{DBLP:conf/gcai/BorgwardtT15,DBLP:conf/ijcai/BorgwardtT15} and \citeauthor{DBLP:conf/ausai/BaaderBL15}~\cite{DBLP:conf/ausai/BaaderBL15} investigate the complexity of answering \LTL-CQs with respect to weaker ontology languages such as $\EL$ and members of the \DL{} family. In this context, \citeauthor{DBLP:conf/frocos/BorgwardtLT13}~\cite{DBLP:conf/frocos/BorgwardtLT13,DBLP:journals/ws/BorgwardtLT15} study the rewritability properties of \LTL-CQs. \citeauthor{DBLP:journals/semweb/BourgauxKT19}~\cite{DBLP:journals/semweb/BourgauxKT19} investigate the problem of querying inconsistent data, and \citeauthor{DBLP:conf/aaai/Koopmann19}~\cite{DBLP:conf/aaai/Koopmann19} proposes an extension to probabilistic data.


As far as OMQ answering with temporal
ontologies is concerned, related work has been done
on querying temporal data with respect to temporalised
$\EL$ ontologies. In this case, since OMQ
answering with atemporal $\EL$ ontologies is
already \PTime-complete, a more expressive target language
than FO$(<)$ is required. \citeauthor{DBLP:conf/ijcai/Gutierrez-Basulto16}~\cite{DBLP:conf/ijcai/Gutierrez-Basulto16} consider a
temporal extension $\mathcal{TEL}$ of $\EL$ and investigate
the complexity and rewritability of atomic queries.
It is not known whether query answering in the full
language with rigid roles is decidable. However, it is
\PTime-complete for data and \PSpace-complete for combined
complexity in its fragment without rigid roles, and
\PSpace-complete in data and in \ExpTime{} for combined
complexity in the fragment where rigid roles can only
occur on the left-hand side of concept inclusions.
It is also shown that,
for acyclic ontologies, one obtains rewritability into
the extension of $\FO(<)$ with $+$, and that query answering is in \PTime{} for
combined complexity. Recent work of~\citeauthor{DBLP:conf/ruleml/BorgwardtFK19}~\cite{DBLP:conf/ruleml/BorgwardtFK19} investigates temporal ontology-mediated querying over sparse temporal data. A temporal extension of $\mathcal{ELH}_\bot$ is able to express different types of rigid concepts, with OMQ answering for rooted CQs with guarded negation and metric temporal operators under the minimal-world semantics being \PTime-complete for data and \ExpSpace-complete for combined complexity.

Extensions of datalog by constraints over an ordered domain representing time provide an alternative and well investigated approach to querying temporal data~\cite{DBLP:journals/tcs/Revesz93,DBLP:journals/jcss/KanellakisKR95,DBLP:journals/jlp/TomanC98}. In this approach, (possibly infinite) database relations are represented using constraints, and datalog programs with constraints play the role of both the ontology and the database query. 
A fundamental difference between datalog with constraints and our formalism is the arity of the relation symbols: our formalism is essentially monadic in the sense that the temporal precedence relation is the only non-unary relation symbol used, whereas datalog alone admits already arbitrary many relation symbols of arbitrary arity. 
A systematic comparison of the expressive power of the respective datalog and \LTL-based formalisms is beyond the scope of this article, but would be of great interest. It would also be of interest to see in how far datalog with constraints can be used as a target language for rewriting ontology-mediated temporal queries.

Our main target languages of query rewriting in this article are $\FO(<)$ and its extensions $\FOE$ and $\FO(\RPR)$. An alternative approach that could be of interest when studying the succinctness of rewritings is to consider as target languages \LTL{} and its second-order extensions such as \textsl{ETL} (\LTL{} with regular expressions)~\cite{DBLP:journals/iandc/Wolper83} and $\mu$\LTL{} (\LTL{} with fixpoints) \cite{DBLP:conf/tls/BanieqbalB87,DBLP:conf/popl/Vardi88}.


\section{Ontologies and Ontology-Mediated Queries in Linear Temporal Logic \LTL}
\label{sec:tdl}

We begin by defining our temporal ontology and query languages as fragments of the classical linear temporal logic \LTL{} (aka propositional temporal logic \textsl{PTL} or \textsl{PLTL}); see~\cite{Gabbayetal94,gkwz,DBLP:books/cu/Demri2016} and references therein.

\subsection{\LTL{} Knowledge Bases}\label{sec:KB}

Having in mind the specific application area for these languages, which has been described in the introduction, we somewhat modify the standard \LTL{} terminology. For instance, instead of propositional variables, we prefer to speak about atomic concepts that (similarly to concept names in Description Logic) are interpreted, in the temporal context, as sets of time points. Data instances are then membership assertions stating that a moment of time $\ell \in \Z$ is an instance of an atomic concept $A$. We construct  complex temporal concepts by applying temporal and Boolean operators to atomic ones. Finally, we define terminologies and capture background knowledge by means of ontology axioms, which are clauses representing inclusions (implications) between concepts that are supposed to hold at every moment of time.

Thus, in this article, we think of the alphabet of \LTL{} as a countably infinite set of \emph{atomic concepts} $A_i$, for $i < \omega$. \emph{Basic temporal concepts}, $C$, are defined by the grammar
\begin{equation}\label{tem-conc}
C \ \ ::=\ \ A_i  \ \ \mid\ \ \Rbox C \ \ \mid \ \ \Lbox C \ \ \mid\ \ \Rnext C \ \ \mid \ \ \Lnext C
\end{equation}
with the \emph{temporal operators} $\Rbox$ (always in the future), $\Lbox$ (always in the past), $\Rnext$ (at the next moment) and $\Lnext$ (at the previous moment). A \emph{temporal ontology}, $\TO$, is a finite set of \emph{clauses} of the form
\begin{equation}\label{axiom1}
C_1 \land \dots \land C_k ~\to~ C_{k+1} \lor \dots \lor C_{k+m},
\end{equation}
where $k,m \ge 0$ and the $C_i$ are basic temporal concepts. We often refer to the clauses in $\TO$ as (\emph{ontology}) \emph{axioms}. As usual, we denote the empty $\land$ by $\top$ and the empty $\lor$ by $\bot$.
We classify ontologies by the shape of their axioms and the temporal operators that occur in them. Let $\frag \in \{\bool, \horn, \krom,\core\}$ and $\op \in \{\Box, \nxt,\Box\nxt\}$. By an $\LTL_\frag^{\op}$-\emph{ontology} we mean any temporal ontology whose clauses satisfy the following restrictions on $k$ and $m$ in~\eqref{axiom1}  indicated by $\frag$:
\begin{description}\itemsep=0pt
\item[\textit{horn}:] $m\leq 1$,

\item[\textit{krom}:] $k + m\leq 2$,

\item[\textit{core}:] $k + m\leq 2$ and $m \leq 1$,

\item[\textit{bool}:] any $k,m \ge 0$,
\end{description}
and may only contain occurrences of the (future and past) temporal operators indicated in $\op$ (for example, $\op = \Box$ means that only $\Rbox$ and $\Lbox$ may occur in the temporal concepts). Note that any $\LTL_\frag^{\op}$-ontology may contain \emph{disjointness  axioms} of the form $C_1 \land C_2 \to \bot$. Although both $\LTL_\krom^{\op}$- and $\LTL_\core^{\op}$-ontologies may only have \emph{binary} clauses as axioms (with at most two concepts), only the former are allowed to contain \emph{universal covering axioms} such as $\top \to C_1 \lor C_2$; in other words, $\core = \krom \cap \horn$.

The definition above identifies a rather restricted set of \LTL-formulas as possible ontology axioms. For example, it completely disallows the use of the standard temporal operators $\Rdiamond$ (sometime in the future), $\Ldiamond$ (sometime in the past), $\U$ (until) and $\Si$ (since). Whether or not these operators can be expressed in a fragment $\LTL_\frag^{\op}$ (in the context of answering ontology-mediated queries) depends on $\frag$ and $\op$. We discuss this issue in Section~\ref{sec:expressivity}.

A  \emph{data instance}, $\A$, is a finite set of \emph{ground atoms} of the form $A_i(\ell)$, where  $\ell \in \Z$. We denote by $\min \A$ and $\max \A$ the minimal and maximal integer numbers occurring in~$\A$. The \emph{active domain} of a data instance~$\A$ is the set $\tem(\A) = \bigl\{\,n \in \Z\mid \min \A \leq n \leq \max \A\,\bigr\}$. To simplify constructions and without much loss of generality, we assume that $0 = \min \A$ and $1 \le \max \A$, implicitly adding `dummies' such as $D(0)$ and $D(1)$ if necessary, where $D$ is a fresh atomic concept (which will never be used in queries).  An $\LTL^{\op}_\frag$ \emph{knowledge base} (KB, for short) is a pair $(\TO,\A)$, where $\TO$ is an $\LTL^{\op}_\frag$-ontology and $\A$ a data instance. The \emph{size}~$|\TO|$ of an ontology $\TO$ is the number of occurrences of symbols in $\TO$.

We use the standard semantics for \LTL{} over $(\Z,<)$ with the \emph{strict} interpretation of temporal operators. A (\emph{temporal}) \emph{interpretation} $\Mmf$ associates with every atomic concept $A$ a subset $A^\Mmf\subseteq \Z$.
The \emph{extension} $C^\Mmf$ of a basic temporal concept $C$ in $\Mmf$ is defined inductively as follows:
\begin{align}
\label{eq:semantics:box}
(\Rbox C)^\Mmf & =  \bigl\{\, n\in\Z \mid k\in C^\Mmf, \text{  for all } k>n \,\bigr\}, &
(\Lbox C)^\Mmf & =  \bigl\{\, n\in\Z \mid k\in C^\Mmf, \text{  for all } k<n \,\bigr\},\\
\label{eq:semantics:next}
(\Rnext C)^\Mmf & =  \bigl\{\, n\in\Z \mid n + 1\in C^\Mmf \,\bigr\}, &
(\Lnext C)^\Mmf & =  \bigl\{\, n\in\Z \mid n - 1\in C^\Mmf \,\bigr\}.
\end{align}
A clause of the form~\eqref{axiom1} is interpreted in $\Mmf$ \emph{globally} in the sense that it is regarded to be \emph{true} in $\Mmf$ if
\begin{equation*}
C_1^{\Mmf} \cap \dots \cap C_k^{\Mmf} ~\subseteq~ C_{k+1}^{\Mmf} \cup \dots \cup C_{k+m}^{\Mmf},
\end{equation*}
where the empty $\cap$ is $\Z$ and the empty $\cup$ is $\emptyset$. Given a clause $\alpha$, we write $\Mmf \models \alpha$ if $\alpha$ is true in $\Mmf$.
We call $\Mmf$ a \emph{model} of $(\TO,\A)$ and write
$\Mmf\models (\TO,\A)$ if
\begin{equation*}
 \Mmf \models \alpha \text{ for all } \alpha \in \TO \qquad\text{ and }\qquad
 \ell\in A^{\smash{\Mmf}} \text{ for all } A(\ell)\in \A.
\end{equation*}
We say that $\TO$ is \emph{consistent} (or \emph{satisfiable}) if there is an interpretation $\Mmf$, called a \emph{model of $\TO$}, such that $\Mmf\models\alpha$,  for all~$\alpha \in \TO$; we also say that $\A$ is \emph{consistent with} $\TO$ (or that the KB $(\TO,\A)$ is \emph{satisfiable}) if there is a model of $(\TO,\A)$.  A basic temporal concept $C$ is \emph{consistent with} $\TO$ if there is a model $\Mmf$ of $\TO$ such that $C^{\Mmf}\ne\emptyset$. For a clause $\alpha$, we write~$\TO\models \alpha$ if $\Mmf \models \alpha$ for every model $\Mmf$ of $\TO$.

The combined complexity of the satisfiability problem for $\LTL_\frag^\op$ KBs $\K$ is shown in Table~\ref{LTL-comb-table}. In those results, we assume that the size of $\K$ is $|\TO|$ plus the size of the encoding of $\A$, and the numbers $\ell$ in $\A$ are assumed to be encoded in unary.
(It is to be noted that the $\LTL$-based languages studied in~\cite{AKRZ:LPAR13,DBLP:journals/tocl/ArtaleKRZ14} are somewhat different from those introduced in this article: in particular, the negative atoms in data instances~\cite{AKRZ:LPAR13} do not affect complexity, whereas implications in data instances~\cite{DBLP:journals/tocl/ArtaleKRZ14} change \NL-hardness to \NP-hardness.)

\begin{table}[t]
\centering%
\tabcolsep=20pt%
\begin{tabular}{cccc}\toprule
     \rule[-3pt]{0pt}{13pt}
     $\frag$ & $\LTL_\frag\Xbox$ & $\LTL_\frag\Xnext$ & $\LTL_\frag\Xallop$    \\\midrule%
     \bool & \NP~{\scriptsize\cite{ono1980on}} & \multicolumn{2}{c}{\PSpace~{\scriptsize\cite{DBLP:journals/jacm/SistlaC85}}} \\%
     \krom &  \NP & \NL$^*$ & \NP~{\scriptsize\cite{AKRZ:LPAR13}} \\
     \horn & \PTime~{\scriptsize\cite{AKRZ:LPAR13}} & \multicolumn{2}{c}{\PSpace~{\scriptsize\cite{Chen199495}}} \\
     \core & \NL~{\scriptsize\cite{AKRZ:LPAR13}} & \NL & \NP~{\scriptsize\cite{AKRZ:LPAR13}} \\\bottomrule
   \end{tabular}
\caption{Combined complexity of $\LTL$ KB satisfiability ($^*$the result follows from the proof of Theorem 1 in~\cite{AKRZ:LPAR13}).}
\label{LTL-comb-table}
\end{table}

\subsection{Ontology-Mediated Queries}

We next define languages for querying temporal knowledge bases. In classical atemporal OMQ answering with description logic ontologies, the standard language for retrieving data from KBs consists of \emph{conjunctive queries} (CQs, for short) or \emph{unions} thereof (UCQs)~\cite{DBLP:journals/jar/CalvaneseGLLR07}. In our present temporal setting, we consider significantly more expressive queries: we start by investigating queries that are arbitrary \LTL{}-formulas
or, equivalently, by Kamp's Theorem~\cite{phd-kamp,DBLP:journals/corr/Rabinovich14}, arbitrary monadic FO-formulas with a single
free variable and a built-in linear order relation. In Section~\ref{FO-OMQs},
we lift our results to formulas in the monadic FO with multiple answer variables.

A \emph{temporal concept}, $\varkappa$, is an arbitrary \LTL-formula defined by the grammar
\begin{align*}
& \varkappa  \ \ ::= \ \  \top \quad \mid\quad A_i \quad \mid\quad  \neg \varkappa \quad \mid\quad \varkappa_1 \land \varkappa_2 \quad \mid \quad \varkappa_1 \lor \varkappa_2 \quad \mid \quad \Rnext \varkappa \quad \mid \quad  \Rdiamond \varkappa \quad \mid \quad  \Rbox \varkappa \quad \mid \quad \varkappa_1 \U\varkappa_2\\
&\mbox{}\hspace*{8.8cm} \mid \quad \Lnext \varkappa \quad \mid \quad  \Ldiamond \varkappa \quad \mid \quad  \Lbox \varkappa \quad \mid \quad \varkappa_1 \Si\varkappa_2.
\end{align*}
A \emph{positive temporal concept} is a temporal concept without occurrences of $\neg$; note that positive temporal concepts~$\varkappa$ include all basic temporal concepts of the form~\eqref{tem-conc}. Let $\Mmf$ be a temporal interpretation. The \emph{extension} $\varkappa^{\Mmf}$ of a temporal concept~$\varkappa$ in~$\Mmf$ is given using~\eqref{eq:semantics:box}--\eqref{eq:semantics:next} and the following:
\begin{align*}
\top^\Mmf & = \Z,  &
(\neg\varkappa)^\Mmf & = \Z \setminus \varkappa^\Mmf,\\
(\varkappa_1 \land \varkappa_2)^\Mmf & = \varkappa_1^\Mmf \cap \varkappa_2^\Mmf, &
(\varkappa_1 \lor \varkappa_2)^\Mmf & = \varkappa_1^\Mmf \cup \varkappa_2^\Mmf,\\
(\Rdiamond \varkappa)^\Mmf & =  \bigl\{\, n\in\Z \mid \text{there is } k>n \text{ with } k\in\varkappa^\Mmf\,\bigr\}, &
(\Ldiamond \varkappa)^\Mmf & = \bigl\{\, n\in\Z \mid \text{there is } k<n \text{ with } k\in\varkappa^\Mmf\,\bigr\},\\
(\varkappa_1 \U \varkappa_2)^\Mmf & = \bigl\{\, n\in\Z \mid \text{there is } k>n \text{ with } k\in\varkappa_2^\Mmf \text{ and }m\in\varkappa_1^\Mmf \text{ for  } n < m < k \,\bigr\}, \hspace*{-20em}\\
(\varkappa_1 \Si \varkappa_2)^\Mmf & =\bigl\{\, n\in\Z \mid \text{there is } k<n \text{ with } k\in\varkappa_2^\Mmf \text{ and }m\in\varkappa_1^\Mmf \text{ for  } n > m > k \,\bigr\}. \hspace*{-20em}
\end{align*}
An $\LTL^{\op}_\frag$ \emph{ontology-mediated instance query} (\OMIQ, for short) is a pair of the form $\q = (\TO, \varkappa)$, where $\TO$ is an $\LTL^{\op}_\frag$ ontology and $\varkappa$ a temporal concept (which may contain arbitrary temporal operators, not only those indicated in $\op$). If $\varkappa$ is a positive temporal concept, then we refer to $\q$ as an \emph{ontology-mediated positive instance query} (OMPIQ). Finally, if $\varkappa$ is an atomic concept, we call $\q$ an \emph{ontology-mediated atomic query} (\OMAQ).

A \emph{certain answer} to an \OMIQ{} $\q = (\TO,\varkappa)$ over a data instance $\A$ is any number $\ell\in \tem(\A)$ such that  $\ell  \in \varkappa^{\smash{\Mmf}}$ for every model $\Mmf$ of~$(\TO, \A)$.
The set of all certain answers to $\q$ over $\A$ is denoted by $\ans(\q,\A)$.
As a technical tool in our constructions, we also require `certain answers' $\ell$ that range over the whole~$\Z$ rather than only the \emph{active temporal domain} $\tem(\A)$; we denote the set of such certain answers \emph{over $\A$ and $\Z$} by $\ans^\Z(\q,\A)$.

\begin{example}\label{ex1}\em
Suppose $\TO = \{\,\Lnext A \to B, \ \Lnext B \to A\,\}$ and $\A = \{\,A(0)\,\}$.
Then  $2n+1\in B^\Mmf$, for any $n \geq 0$ and any model~$\Mmf$ of $(\TO,\A)$. It follows that, for $\q = (\TO,\Rnext\Rnext B)$, we have $\ans^\Z(\q,\A) = \{\,2n-1 \mid  n\geq 0\,\}$, while $\ans(\q,\A) = \{1\}$ because $\tem(\A)=\{0,1\}$, as agreed in Section~\ref{sec:KB}.
\end{example}

By the \emph{\OMIQ{} answering problem for} $\LTL^{\op}_\frag$ we understand the decision problem for the set $\ans(\q,\A)$, where $\q$ is an $\LTL^{\op}_\frag$ \OMIQ{} and $\A$ a data instance, that is, given any $\ell \in \tem(\A)$, decide whether $\ell \in \ans(\q,\A)$. By restricting OMIQs to OMPIQs or OMAQs, we obtain the \emph{OMPIQ} or, respectively, \emph{OMAQ answering problem} for $\LTL^{\op}_\frag$.

The success of the classical ontology-based data management paradigm~\cite{IJCAI-18} has been largely underpinned by the fact that, for suitable ontology and query languages such as \DL{} and CQs, answering OMQs can be uniformly reduced to evaluating first-order queries (or standard SQL queries) directly over the data instance. Such `FO-rewritability' of OMQs implies that answering each of them can be done in $\ACz$ for \emph{data complexity}, that is, under the assumption that the ontology and query are fixed and the data instance is the only input.

Our main aim in this article is to investigate rewritability of $\LTL^{\op}_\frag$ OMIQs, OMPIQs and OMAQs into various types of first-order queries.
With this in mind, we think of any data instance $\A$ as a finite first-order structure $\SA$ with domain $\tem(\A)$ ordered by $<$, in which
\begin{equation*}
\SA \models A(\ell) \quad \text{iff} \quad  A(\ell)\in \A,
\end{equation*}
for any atomic concept $A$ and any $\ell\in\tem(\A)$.
The structure $\SA$ represents a temporal database over which we can evaluate various types of first-order formulas (queries). The smallest target language for rewritings comprises $\FO(<)$-formulas, that is, arbitrary first-order formulas with one built-in binary predicate $<$. A more expressive target language $\FOE$ extends $\FO(<)$ with the standard unary numeric predicates $t \equiv 0 \pmod n$, for $n > 1$, defined by taking $\SA \models \ell \equiv 0 \pmod n$ iff $\ell$ is divisible by $n$.
Evaluation of both $\FO(<)$- and $\FOE$-formulas can be done in \LogTime-uniform $\smash{\ACz{}}$ for data complexity~\cite{Immerman99}, one of the smallest complexity classes. It is to be noted that even though $\FO(<)$ and $\FOE$ lie in the same complexity class, their expressive power  differs substantially; see Example~\ref{exampleBvNB} and Remark~\ref{re:exp} below.

Our most expressive target language for rewritings is $\FO(\RPR)$ that extends $\FO$ with the successor relation and
\emph{relational primitive recursion} ($\RPR$, for short). (Note that we do not require the predicate $\textsc{bit}$ or, equivalently, the predicates $\plus$ and $\TIMES$ in this language; cf.~\cite{DBLP:journals/iandc/ComptonL90}.) Evaluation of $\FO(\RPR)$-formulas is known to be \NCo-complete for data complexity~\cite{DBLP:journals/iandc/ComptonL90}, with
$\ACz \subsetneq \NCo \subseteq \LogSpace$. We remind the reader that, using RPR, we can construct formulas such as
\begin{equation*}
\Phi(\avec{z},\avec{z}_1,\dots,\avec{z}_n) ~=~  \left[ \begin{array}{l}
Q_{1}(\avec{z}_1,t) \equiv \varTheta_1\big(\avec{z}_1,t,Q_1(\avec{z}_1,t-1),\dots,Q_n(\avec{z}_n,t-1)\big)\\
\dots\\
Q_{n}(\avec{z}_n,t) \equiv \varTheta_n\big(\avec{z}_n,t,Q_1(\avec{z}_1,t-1),\dots,Q_n(\avec{z}_n,t-1)\big)
\end{array}\right] \ \Psi(\avec{z},\avec{z}_1,\dots,\avec{z}_n),
\end{equation*}
where the part of $\Phi$ within $[\dots]$ defines recursively, via the $\FO(\RPR)$-formulas $\varTheta_i$, the interpretations of the predicates~$Q_i$ in the $\FO(\RPR)$-formula $\Psi$ (see Example~\ref{exampleBvNB:2} for an illustration).
Note that the recursion starts at $t=0$ and assumes that $Q_{i}(\avec{z}_i,-1)$ is false for all $Q_{i}$ and all $\avec{z}_i$, with $1 \leq i\leq n$. Thus, the truth value of $Q_{i}(\avec{z}_i,0)$ is computed by substituting falsehood $\bot$ for all $Q_{i}(\avec{z}_i,-1)$. For $t=1,2,\dots$, the recursion is then applied in the obvious way.
We assume that the relation variables $Q_i$ can only occur in one recursive definition $[\dots]$, so it makes sense to write $\SA\models Q_i(\avec{n}_i,k)$, for any tuple $\avec{n}_i$ in $\tem(\A)$ and $k \in \tem(\A)$, if the computed value is `true'.
Using thus defined truth-values, we compute inductively the truth-relation $\SA\models \Psi(\avec{n},\avec{n}_1,\dots,\avec{n}_n)$, and so $\SA\models \Phi(\avec{n},\avec{n}_1,\dots,\avec{n}_n)$, as usual in first-order logic.

We are now in a position to introduce the central notion of the paper that reduces answering \OMIQ{}s over data instances $\A$ to evaluation of first-order queries over $\SA$, which can be carried out by standard temporal databases.

\begin{definition}\label{rewriting}\em
Let $\lang$ be one of the three classes of FO-formulas introduced above: $\FO(<)$, $\FOE$ or $\FO(\RPR)$. Let $\q=(\TO, \varkappa)$ be an \OMIQ{} and $\rq(t)$ a constant-free $\lang$-formula with a single free variable $t$. We call $\rq(t)$ an $\lang$-\emph{rewriting of} $\q$ if, for any data instance~$\A$, we have $\ans(\q,\A) = \{\, \ell \in\tem(\A) \mid \SA \models \rq(\ell)\,\}$. We say that $\q$ is $\lang$-\emph{rewritable} if it has an $\lang$-rewriting.
\end{definition}

Answering both $\FO(<)$- and $\FOE$-rewritable \OMIQ{}s is clearly in $\ACz$ for data complexity, while answering $\FO(\RPR)$-rewritable \OMIQ{}s is in \NCo{} for data complexity.

\begin{remark}\label{rem:minmax}\em
$(i)$ In the definition above, we allowed no constants (numbers) from $\Z$ in rewritings. Note, however, that the \mbox{$\FO(<)$}-formulas $\neg \exists t'\, (t'< t)$ and $\neg \exists t'\, (t'>t)$ define the minimal and maximal numbers that occur in any given data instance. In view of this, we can use  the constants $0$ (or, $\min$) and $\max$ in $\FO(<)$-rewritings as syntactic sugar. We also use the following $\FO(<)$ abbreviations defined by induction on $a > 0$:
\begin{equation*}
(t = t' + a) \quad = \quad \begin{cases}
(t > t') \land \neg \exists s\,(t' < s < t), &  \text{if } a = 1,\\
\exists t''\,\bigl((t = t'' + (a - 1)) \land (t'' = t' + 1)\bigr),  & \text{if } a > 1,
\end{cases}
\end{equation*}
where $t' < s < t$ stands for $(t' < s) \land (s < t)$; formula $(t = t' + a)$ is also a shortcut for $(t' = t + (-a))$ if $a < 0$ and for~$(t = t')$ if~$a = 0$.

\smallskip

(\emph{ii}) In $\FOE$-rewritings, we require formulas that express membership in arithmetic progressions, that is, sets of the form
\begin{equation*}
a + b\N = \bigl\{\, a + bk \mid k \geq 0\,\bigr\},\qquad \text{ for } a, b \geq 0,
\end{equation*}
where $\N$ denotes the set of natural numbers (including $0$). So, for individual variables  $t$ and $t'$, we write $t - t' \in a + b\N$  to abbreviate $(t = t' + a)$ if $b = 0$, and $\exists t''\,\bigl[(t'' = t' + a) \land (t'' \leq t)]$ if $b = 1$; otherwise, that is, for $b > 1$, the formula $t - t' \in a + b\N$ stands for
\begin{equation*}
\exists t''\,\bigl[(t'' = t' + a) \land (t \geq t'') \land \bigvee_{0 \leq c < b} \bigl((t \equiv c\!\!\! \pmod b) \land (t'' \equiv c\!\!\! \pmod b)\bigr)\bigr],
\end{equation*}
where $t\equiv c \pmod b$ with $0 \leq c < b$ is an abbreviation for $\exists s \,\bigl[\bigl((s = t + (- c)) \lor (s = t + (b - c))\bigr)\land \bigl(s \equiv 0 \pmod b\bigr)\bigr]$.

\smallskip

(\emph{iii}) Observe that $\FO(\RPR)$ does not explicitly have the predicate $<$ because it can easily be expressed using $\RPR$; see~\cite[Proposition 4.1]{DBLP:journals/iandc/ComptonL90}. Therefore, every $\FO(<)$-formula is expressible in $\FO(\RPR)$. Similarly, every formula in $\FOE$ is expressible in $\FO(\RPR)$: indeed,
\begin{equation*}
\left[ \begin{array}{rcl}
Q_{0}(t) &\equiv& ((t=0) \lor Q_{b-1}(t-1))\\
Q_{b-1}(t) &\equiv& Q_{b-2}(t-1),\\
&\!\dots\!&\\
Q_{1}(t) &\equiv& Q_{0}(t-1)
\end{array}\right] \ Q_0(s),
\end{equation*}
expresses $s \equiv 0 \pmod b$, for $b \geq 1$.
\end{remark}

We illustrate the given definitions by a few examples.

\begin{example}\label{exampleBvNB}\em
Consider the \OMAQ{} $\q = (\TO,A)$, where $\TO$ is the same as in Example~\ref{ex1}. It is not hard to see that
\begin{equation*}
\rq(t) ~=~ \exists s\,\bigl(A(s) \land (t - s \in 0 + 2\N)\bigr)\  \lor\ \exists s\,\bigl(B(s) \land (t - s \in 1 + 2\N)\bigr)
\end{equation*}
is an $\FOE$-rewriting of $\q$.
Note, however, that $\q$ is \emph{not} $\FO(<)$-rewritable since properties such as `$t$ is even' are not definable by $\FO(<)$-formulas, which can be established using a standard Ehrenfeucht-Fra\"iss\'e argument~\cite{Straubing94,Libkin}.
\end{example}

\begin{example}\label{exampleBvNB:2}\em
Next, consider the \OMAQ{} $\q = (\TO,B_0)$, where $\TO$ consists of the axioms
\begin{equation*}
\Lnext B_k \land A_0 \to B_k \ \ \text{ and } \ \ \Lnext B_{1-k} \land A_1 \to B_k,  \ \ \ \text{for} \ \ k = 0, 1.
\end{equation*}
For any binary word  $\avec{e} = (e_1,\dots,e_{n}) \in \{0,1\}^n$, we take the data instance
$\A_{\avec{e}}  =   \{\,B_0(0)\,\} \cup  \{\, A_{e_i}(i) \mid 0 < i \leq n\,\}$.
It is not hard to check that $n$ is a certain answer to $\q$  over $\A_{\avec{e}}$ iff the number of 1s in~$\avec{e}$ is even (\parity): intuitively, the word is processed starting from the minimal timestamp and moving towards the maximal one, and the first axiom preserves $B_i$ if the current symbol is~0, whereas the second axiom toggles $B_i$ if the current symbol is 1. As \parity{} is not in $\ACz$~\cite{DBLP:journals/mst/FurstSS84}, it follows that $\q$ is not FO-rewritable even if \emph{arbitrary} numeric predicates are allowed in rewritings. However, it can be rewritten to the following $\FO(\text{RPR})$-formula:
\begin{equation*}
\rq(t) ~=~  \left[ \begin{array}{l}
Q_{0}(t) \equiv \varTheta_0\\
Q_{1}(t) \equiv \varTheta_1
\end{array}\right] \ Q_0(t),
\end{equation*}
where
\begin{equation*}
\varTheta_k(t, Q_{0}(t-1), Q_{1}(t-1))  \quad =\quad B_k(t) \ \ \lor \ \ \bigl(Q_k(t-1)\land A_0(t)\bigr) \ \ \lor \ \ \bigl(Q_{1-k}(t-1)\land A_1(t)\bigr),   \ \ \ \text{for} \ \ k = 0, 1.
\end{equation*}
As noted above, the recursion starts from the minimal timestamp 0 in the data instance (with $Q_i(-1)$ regarded  false) and proceeds to the maximal one.
\end{example}

As a technical tool in our constructions of $\FO(<)$- and $\FOE$-rewritings, we also use infinite first-order structures~$\smash{\SA^\Z}$ with domain $\Z$ that are defined in the same way as $\SA$ but over the whole $\Z$.
If in Definition~\ref{rewriting} we replace $\SA$ with $\smash{\SA^\Z}$, then we can speak of $\smash{\FO^\Z}(<)$ -or  $\FOEZ$-\emph{rewritings $\rq(t)$ of} $\q$.

\begin{example}\label{ex:2-3}\em
Suppose $\q = (\TO,\varkappa)$, where $\TO =\{\,A \to \Rnext^2 A,\ B \to \Rnext^3 B\,\}$, $\varkappa = \Rdiamond (A\land B)$ and $\Rnext^k$ is a sequence of $k$-many operators $\Rnext$. Then
\begin{equation*}
\exists s \, \bigl[(t < s) \land \exists u\,\bigl(A(u) \land (s - u \in 0 + 2\N)\bigr) \land  \exists v\,\bigl(B(v) \land (s - v \in 0+3\N)\bigr)\bigr]
\end{equation*}
is an $\FOEZ$-rewriting of the \OMPIQ{} $\q$, but not an $\FOE$-rewriting because although $u$ and $v$ always belong to  the active temporal domain $\tem(\A)$ of $\A$, $s$ can be outside $\tem(\A)$. Interestingly, $\exists u,v\, \bigl[A(u) \land B(v)\bigr]$ is both an $\FO(<)$- and $\FO^{\smash{\Z}}(<)$-rewriting of $\q$.
\end{example}

\color{black}

We conclude this section by discussing the expressive power of the most important languages $\LTL_{\frag}^{\op}$ in comparison with full \LTL.


\subsection{Remarks on Expressivity}\label{sec:expressivity}

We are interested in expressive power modulo the introduction of fresh symbols (atomic concepts). By the \emph{signature} of an ontology we mean the set of atomic concepts that occur in it. An ontology $\TO'$ is called a \emph{model conservative extension} of an ontology $\TO$ if $\TO'\models \TO$, the signature of $\TO$ is contained in the signature of $\TO'$, and every model of $\TO$ can be expanded to a model of $\TO'$ by providing an interpretation of the fresh symbols of $\TO'$ but leaving the domain and the interpretation of the symbols in $\TO$ unchanged. Observe that if $\q = (\TO, \varkappa)$ is an OMIQ and $\TO'$ a model conservative extension of $\TO$, then the certain answers to $\q$ over a data instance $\A$ in the signature of $\TO$ coincide with the certain answers to $\q' = (\TO',\varkappa)$ over $\A$. Thus, any rewriting of $\q'$ is also a rewriting of $\q$.

Observe first that, if arbitrary \LTL-formulas are used as axioms of an ontology $\TO$, then one can construct an $\LTL_{\bool}\Xallop$ ontology $\TO'$ that is a model conservative extension of $\TO$~\cite{FisherDP01,AKRZ:LPAR13}. We do not repeat the proof here but indicate a few important steps. First, we note that $A \to \Rnext B$ is equivalent to $\Lnext A \to B$ and $\Ldiamond A \to B$ is equivalent to $A \to \Rbox B$.
Then the implication $A \to \Rdiamond B$ can be simulated by two $\krom$ clauses with $\Rbox$ and a fresh atomic concept~$C$: namely, $A \land \Rbox C \to \bot$ and $\top \to C \lor B$. To simulate $A \to B \U C$, we use the well-known fixed-point unfolding of $B\U C$ as $\Rnext C\lor(\Rnext B \land \Rnext (B\U C))$, which gives rise to four clauses $A \to U$,  $U \to \Rnext C \lor \Rnext B$,  $U \to \Rnext C \lor \Rnext U$ and $A \to \Rdiamond C$ with a fresh $U$ (the $\Rdiamond$ in the last clause can be replaced with $\Rbox$ as described above). The implication $B \U C \to A$ can be replaced with $\Rnext C \to U$, $\Rnext U \land \Rnext B \to U$ and $U \to A$, for  a fresh~$U$.

We now discuss in more detail the clausal form fragments corresponding to the languages $\textsl{Prior-}\LTL$ and $\textsl{Krom-}\LTL$ used in Theorems~A and~B.


Recall that $\LTLpr$ denotes the set of \LTL{}-formulas constructed using arbitrary Boolean connectives, $\Lbox$, $\Rbox$,~$\Ldiamond$ and $\Rdiamond$. 
By employing the equivalent rewritings of formulas with $\Ldiamond$ or $\Rdiamond$ from the previous paragraph, it is easy to see that, for every
ontology $\TO$ in $\LTLpr$, there exists an $\LTL_{\smash{\bool}}\Xbox$ ontology $\TO'$ that is a model conservative extension of $\TO$. It follows that any FO-rewritability result for OMIQs with ontologies given in $\LTL_{\smash{\bool}}\Xbox$ holds for OMIQs with ontologies in $\LTLpr$.

Also recall that by $\LTLkr$ we denote the set of all \LTL{}-formulas constructed using arbitrary Boolean operators from at most two \LTL{} atomic concepts prefixed with any sequence of unary operators from $\Lbox$, $\Lnext$, $\Ldiamond$, $\Rbox$, $\Rnext$ and~$\Rdiamond$. Then, using  the transformations introduced above, one can construct for every
$\LTLkr$ ontology $\TO$ an $\LTL_{\smash{\krom}}\Xallop$-ontology $\TO'$ that
is a model conservative extension of $\TO$.

Denote by $\LTLho$ the set of implications $\varkappa_1\rightarrow \varkappa_2$, where $\varkappa_1$ is constructed using $\wedge$, $\Lbox$, $\Rbox$, $\Ldiamond$, $\Rdiamond$, and $\varkappa_{2}$ is constructed using
$\wedge$, $\Lbox$ and $\Rbox$. Then, for every
ontology $\TO$ consisting of formulas in $\LTLho$, there exists an $\LTL_{\smash{\horn}}\Xbox$ ontology $\TO'$ that is a model conservative extension of $\TO$.
%
%

The following observation, also based on model conservative extensions, will be required  in Section~\ref{sec:ltl-tomaq:unary}. It is not hard to see that, for any $\LTL_\frag^{\op}$ ontology $\TO$, one can construct an $\LTL_\frag^{\op}$ ontology $\TO'$, possibly using some fresh atomic concepts, such that $\TO'$ contains no nested temporal operators, and $\TO'$ is a model conservative  extension
of $\TO$. For example, $\Rbox\Lnext A \to B$ in $\TO$ can be replaced with two clauses $\Lnext A \to C$ and $\Rbox C \to B$, for  a fresh atomic concept~$C$. Thus, in what follows and where convenient, we can assume without loss of generality that our ontologies do not contain nested temporal operators.

\section{Rewriting $\LTL\Xallop_\bool$ OMIQs into $\FO(\RPR)$}
\label{sec:LTL}

It follows from Example~\ref{exampleBvNB:2} that the languages $\FO(<)$ and $\FOE$ are not sufficiently expressive as target languages for rewritings of arbitrary \OMIQ{}s. In the next theorem, we show, however, that all of them can be rewritten into $\FO(\RPR)$. As follows from~\cite[Proposition~4.3]{DBLP:journals/iandc/ComptonL90}, this means that  we can also rewrite \OMIQ{}s into the language $\MSO(<)$ of \emph{monadic second-order} formulas that are built from atoms of the form $A(t)$ and $t < t'$ using the Booleans, first-order quantifiers $\forall t$ and $\exists t$, and second-order quantifiers $\forall A$ and $\exists A$~\cite{Buchi60}.
(An $\LTL\Xallop_\bool$ \OMIQ{} $\q = (\TO,\varkappa)$ is said to be \emph{$\MSO(<)$-rewritable} if there is an $\MSO(<)$-formula $\rq(t)$ such that $\ans(\q,\A) = \{\,\ell\in\tem(\A) \mid \SA \models \rq(\ell)\,\}$, for any data instance $\A$.)

\begin{remark}\label{re:exp}\em
It is worth reminding the reader (see~\cite{Straubing94,abs-1011-6491,DBLP:books/ws/phaunRS01/ComptonS01} for details) that, by the B\"uchi--Elgot--Trakhtenbrot Theorem~\cite{Buchi60,Elgot61,Trakh62}, $\MSO(<)$-sentences define exactly the class of regular languages, $\FOE$-sentences define exactly the class of regular languages in (non-uniform) \ACz, and $\FO(<)$-sentences define the class of star-free regular languages. $\FO(\RPR)$, extended with the predicates $\plus$ and $\TIMES$ or, equivalently, with one predicate \textsc{bit}~\cite{Immerman99}, captures exactly the languages in \NCo{} (which are not necessarily regular)~\cite{DBLP:journals/iandc/ComptonL90}.
On the other hand, a close connection between \LTL{} and finite automata has been known in formal verification (model checking) since the 1980s~\cite{VardiW86}.
\end{remark}

\begin{theorem}\label{fullLTL}
All $\LTL\Xallop_\bool$ \OMIQ{}s are $\FO(\RPR)$- and $\MSO(<)$-rewritable, and so  answering such \OMIQ{}s is in \NCo{} for data complexity.
\end{theorem}
\begin{proof}
Let $\q = (\TO,\varkappa_0)$ be an $\LTL\Xallop_\bool$ \OMIQ{}. Denote by $\subq$ the set of subconcepts of temporal concepts in $\TO$ and~$\varkappa_0$ together with their negations. A \emph{type} for $\q$ is a  maximal subset $\tp$ of $\subq$ consistent with $\TO$. Let $\tpset = \{\tp_1, \dots, \tp_n \}$ be the set of all such types for $\q$. Given a model $\Mmf$ of $\TO$ and any $k \in \Z$, we denote by $\tp_\Mmf(k) = \{\, \varkappa\in \subq \mid k \in \varkappa^{\Mmf} \,\}$ the \emph{type of $k$ in} $\Mmf$. We write $\suc(\tp, \tp')$ to say that there exist a model $\Mmf$ of $\TO$ and a number $k \in \Z$ such that $\tp = \tp_\Mmf(k)$ and $\tp' = \tp_\Mmf(k+1)$.
Given a data instance $\A$ and $\ell\in\tem(\A)$, we denote $\A|_{\le \ell} = \{\,A(k) \in \A \mid k \le \ell\,\}$ and $\A|_{\ge \ell} = \{\,A(k) \in \A \mid k \ge \ell\,\}$. Note that $\tem(\A|_{\le \ell}) = \{\,0,\dots,\ell\,\}$, while $\tem(\A|_{\ge \ell}) = \tem(\A) = \{\,0,\dots,\max(\A) \,\}$.

For every type $\tau$ for $\q$, we construct two $\FO(\RPR)$-formulas $\varphi_\tp(t_0)$ and $\psi_\tp(t_0)$ by taking:
\begin{align*}
\varphi_\tp(t_0) \ \ = \ \ \left[ \begin{array}{l}
R_{\tp_1}(t) \equiv \vartheta_{\tp_1}\\
\dots\\
R_{\tp_n}(t) \equiv \vartheta_{\tp_n}
\end{array}\right] \ \ R_{\tp}(t_0),
\hspace*{2cm}
\psi_\tp(t_0) \ \ = \ \ \left[ \begin{array}{l}
Q_{\tp_1}(t_0,t) \equiv \eta^{\tp}_{\tp_1}\\
\dots\\
Q_{\tp_n}(t_0,t) \equiv \eta^{\tp}_{\tp_n}
\end{array}\right] \ \ \bigvee_{\tp_i \in \tpset} Q_{\tp_i}(t_0,\max),
\end{align*}
where $R_{\tp_i}(t)$ and $Q_{\tp_i}(t_0,t)$, for $\tp_i \in \tpset$, are relation variables  and
\begin{align*}
& \vartheta_{\tp_i}(t,R_{\tp_1}(t-1),\dots,R_{\tp_n}(t-1)) ~=~ \typef_{\tp_i}(t) \land \bigl((t=0) \lor \bigvee_{\tp' \in \tpset,\ \mathsf{suc}(\tp', \tp_i)} R_{\tp'}(t-1)\bigr),\\[5pt]
& \eta^{\tp}_{\tp_i}(t_0, t, Q_{\tp_1}(t_0,t-1),\dots,Q_{\tp_n}(t_0,t-1)) ~=~
\begin{cases}\displaystyle
\typef_{\tp_i}(t) \land \bigl((t=t_0) \lor \bigvee_{\tp' \in \tpset,\ \mathsf{suc}(\tp', \tp_i)} Q_{\tp'}(t_0,t-1)\bigr), & \text{ if } \tp_i = \tp;\\ \displaystyle
\typef_{\tp_i}(t) \land  \bigvee_{\tp' \in \tpset,\ \mathsf{suc}(\tp', \tp_i)} Q_{\tp'}(t_0,t-1), & \text{ if } \tp_i \ne \tp,
\end{cases}
\end{align*}
and $\typef_\tp(t)$  is the conjunction of all
$\neg A(t)$ with~$\neg A \in \tp$.\nz{changed} We require the following auxiliary lemma.

\begin{lemma}\label{recursion}
For any data instance $\A$ and any $\ell \in \tem(\A)$,
\begin{itemize}
\item[$(i)$] $\SA \models \varphi_\tau(\ell)$ iff there is a model $\Mmf$ of $\TO$ and $\A|_{\le \ell}$ such that $\tp = \tp_\Mmf(\ell)$;

\item[$(ii)$] $\SA \models \psi_\tau(\ell)$ iff there is a model $\Mmf$ of $\TO$ and $\A|_{\ge \ell}$ such that $\tp = \tp_\Mmf(\ell)$;

\item[$(iii)$] $\SA \models \varphi_\tau(\ell) \land \psi_\tau(\ell)$ iff there is a model $\Mmf$ of $\TO$ and $\A$ such that $\tp = \tp_\Mmf(\ell)$.
\end{itemize}
\end{lemma}
\begin{proof}
$(i)$ We proceed by induction on $\ell$. Suppose $\ell = \min \A = 0$. If $\SA \models R_{\tp} (0)$, then $A(0) \in \A$ implies $A \in \tp$, and so, since $\tp$ is consistent with $\TO$ (which means that $\TO$ is satisfiable in some model $\Mmf$ with $\tp_\Mmf(0) = \tp$), there is a model $\Mmf$ of $\TO$ and $\A|_{\le 0}$ with $\tp = \tp_\Mmf(0)$. Conversely, let $\Mmf$ be a model of $\TO$ and $\A|_{\le 0}$, and  $\tp = \tp_\Mmf(0)$. It follows from the structure of $\vartheta_\tau$ that $\SA \models \typef_{\tp}(t) \land (t=0)$, and so $\SA \models R_\tau(0)$ and $\SA \models \varphi_\tau(0)$. 

Assume now that $(i)$ holds for $\ell - 1 \ge 0$. Suppose $\SA \models R_{\tp} (\ell)$. Then there is $\tp'$ such that $\suc(\tp', \tp)$ and $\SA \models\typef_{\tp}(\ell) \land R_{\tp'}(\ell-1)$. By the induction hypothesis, there is a model $\Mmf'$ of $\TO$ and $\A|_{\le \ell -1}$ such that $\tp' = \tp_{\Mmf'}(\ell-1)$. As $\suc(\tp', \tp)$ holds, there is a model $\Mmf''$ of $\TO$ with $\tp' = \tp_{\Mmf''}(\ell -1)$ and $\tp = \tp_{\Mmf''}(\ell)$. The required model $\Mmf$ of $\TO$ and~$\A|_{\le \ell}$ is obtained by taking, for any atomic concept $A$,
\begin{equation}\label{glue}
A^\Mmf = \bigl\{\, k < \ell \mid k \in A^{\Mmf'}\,\bigr\} \cup \bigl\{\, k \ge \ell \mid k \in A^{\Mmf''}\,\bigr\}.
\end{equation}
The converse implication is straightforward by the induction hypothesis. 

\smallskip

$(ii)$ Suppose $\SA \models \psi_\tau(\ell)$. We claim  that there are types $\tp_{\ell}, \tp_{\ell+1},\dots,\tp_{\max(\A)}$ for $\q$ such that $\tp_{\ell} = \tp$, 
$\mathsf{suc}(\tp_{{j}},\tp_{{j+1}})$, for $\ell \leq j < \max(\A)$, and 
\begin{equation*}
\SA \models \typef_{\tp_\ell}(\ell) \land \typef_{\tp_{\ell+1}} (\ell+1) \land \dots \land \typef_{\tp_{{\max(\A)}}}(\max(\A)).
\end{equation*}
Assuming that this claim holds, the remainder of the proof for $(ii)$ is similar to that for $(i)$.  

To prove the claim, observe first that there is a type $\tau_{\max(\A)}$ with $\SA \models Q_{\tau_{\max(\A)}}(\ell, \max(\A))$. Because of $\eta^{\tp}_{\tp_{\max(\A)}}$, we immediately obtain $\SA \models \typef_{\tp_{{\max(\A)}}}(\max(\A))$. Moreover, if $\max(\A) > \ell$, then there is a type $\tau_{\max(\A)-1}$ such that $\suc(\tau_{\max(\A)-1},\tp_{\max(\A)})$, $\SA \models \typef_{\tau_{\max(\A)-1}}(\max(\A)-1)$ and $\SA \models Q_{\tau_{\max(\A)-1}}(\ell, \max(\A)-1)$. If $\max(\A)-1 > \ell$, we can find $\tp_{\max(\A)-2}$ with the properties required by the claim. We proceed in this way until we find the required sequence. Now, it remains to check that $\tp_{\ell} = \tp$. Observe that $\SA \not \models Q_{\tp'}(\ell, \ell')$ for all $\ell' < \ell$ and $\tp' \in \mathfrak T$. Therefore, $\SA \models Q_{\tp_\ell}(\ell, \ell)$ implies by $\eta^{\tp}_{\tp_{\ell}}$ that $\tp = \tp_{\ell}$. The converse implication is  straightforward. 

\smallskip

$(iii)$ We take the models $\Mmf'$ and $\Mmf''$ provided by $(i)$ and $(ii)$, respectively, and construct the required model $\Mmf$ using~\eqref{glue}. This completes the proof of the lemma.
\qed
\end{proof}

It follows immediately from Lemma~\ref{recursion} that the formula
\begin{align*}
\rq(t_0) \ \ = \ \ \neg \bigvee_{\tp \in \tpset,\ \varkappa_0 \notin \tp} ( \varphi_\tp(t_0) \land \psi_\tp(t_0)).
\end{align*}
is an $\FO(\RPR)$-rewriting of $\q = (\TO, \varkappa_0)$. It can be transformed into an $\MSO(<)$-rewriting  in the same way as in the proof of~\cite[Proposition~4.3]{DBLP:journals/iandc/ComptonL90}. 
It is, however, instructive to compute the $\MSO(<)$-rewriting directly from the OMIQ~$\q$. To this end, for every $\varkappa$ in $\subq$, we take a fresh unary predicate
$\varkappa^*(t)$ with the intuitive meaning `$\varkappa$ is true at $t$'\!.
Let~$\Xi(t)$ be the conjunction of the following FO$(<)$-formulas:
\begin{align*}
& A(t) \to A^\ast(t), && \text{ for every atomic } A\in \subq,\\
& (\neg\varkappa)^\ast(t) \leftrightarrow \neg \varkappa^\ast(t),&&\text{ for every } \neg\varkappa \in \subq,\\
& (\varkappa_1\land \varkappa_2)^{\ast}(t) \leftrightarrow \varkappa_1^{\ast}(t) \land \varkappa_2^{\ast}(t), &&\text{ for every }\varkappa_1\land \varkappa_2 \in \subq,\\
& (t < \max) \to \bigl((\Rnext \varkappa)^\ast(t) \leftrightarrow \varkappa^\ast(t+1)\bigr), && \text{ for every }\Rnext\varkappa \in \subq,\\
& (t < \max) \to \bigl((\Rbox \varkappa)^\ast(t) \leftrightarrow \varkappa^\ast(t+1) \land (\Rbox\varkappa)^\ast(t + 1)\bigr), && \text{ for every }\Rbox\varkappa\in\subq,\\
& (t < \max) \to \bigl((\Rdiamond \varkappa)^\ast(t) \leftrightarrow \varkappa^\ast(t+1) \lor (\Rdiamond\varkappa)^\ast(t + 1)\bigr), && \text{ for every }\Rdiamond\varkappa\in\subq,\\
& (t < \max) \to \bigl((\varkappa_1 \U \varkappa_2)^{\ast}(t) \leftrightarrow  \varkappa_2^{\ast}(t+1)
\vee (\varkappa_1^{\ast}(t+1) \wedge (\varkappa_1 \U \varkappa_2)^{\ast}(t+1))\bigr), &&\text{ for every } \varkappa_1\U \varkappa_2 \in \subq,
\end{align*}
and the corresponding formulas for $\Lnext$, $\Lbox$, $\Ldiamond$, and $\Si$ in which $(t < \max)$ and $(t+1)$ are replaced by $(\min < t)$ and $(t-1)$, respectively.
Let  $\tpset$ be the set of all \emph{types} for $\q$.
The following formula is an MSO$(<)$-rewriting of $\q$:
\begin{equation*}
\rq(t_0) \ \ = \ \ \forall \pmb{\varkappa^{\ast}} \,\Bigl(\forall t\, \Bigl( \Xi(t) \land \bigvee_{\tp \in \tpset}\bigwedge_{\varkappa\in \tp} \varkappa^{\ast}(t)\Bigr) \to \varkappa_0^{\ast}(t_0) \Bigr),
\end{equation*}
where $\pmb{\varkappa^{\ast}}$ is a list of all predicate names $\varkappa^{\ast}$ for $\varkappa \in \subq$. To show this, consider a data instance $\A$ and $\ell\in\tem(\A)$.

Suppose first that $\ell\not\in \ans(\q,\A)$. Take a model $\Mmf$ of $(\TO,\A)$ with $\ell\notin \varkappa_0^\Mmf$. Interpret the predicates $\varkappa^{\ast}$ with $\varkappa\in \subq$ in $\SA$ by setting $\mathfrak{a}(\varkappa^{\ast})= \varkappa^\Mmf \cap \tem(\A)$. Then $\SA \models^{\mathfrak{a}} \forall t\, \big[ \Xi(t) \land \bigvee_{\tp\in \tpset}\bigwedge_{\varkappa\in \tp} \varkappa^{\ast}(t)\big]$ but $\SA\not\models^{\mathfrak{a}} \varkappa_{0}^*(\ell)$, and so $\SA\not\models \rq(\ell)$.

Conversely, suppose $\SA \not\models \rq(\ell)$. Let $\mathfrak{a}$ be an interpretation of the predicates $\varkappa^{\ast}$ with $\varkappa\in \subq$ in $\SA$
such that $\SA \models^{\mathfrak{a}} \forall t\, \big[ \Xi(t) \land \bigvee_{\tp\in \tpset}\bigwedge_{\varkappa\in \tp} \varkappa^{\ast}(t)\big]$ and $\SA\not\models^{\mathfrak{a}} \varkappa_{0}^*(\ell)$.
Since every $\tp\in \tpset$ is consistent with $\TO$, there are models $\Mmf_{\min}$ and $\Mmf_{\max}$ of $\TO$ that correspond to the types selected by the disjunction in $\rq(t_0)$ at $\min \A$ and $\max \A$: for all $\varkappa \in \subq$,
\begin{equation*}
\min \A\in \varkappa^{\Mmf_{\min}} \ \ \text{ iff }\ \   \min\A \in \mathfrak{a}(\varkappa^\ast) \qquad\text{ and } \qquad
\max \A\in \varkappa^{\Mmf_{\max}} \ \ \text{ iff }\ \   \max\A \in \mathfrak{a}(\varkappa^\ast).
\end{equation*}
Define an interpretation $\Mmf$ by `stitching' together $\Mmf_{\min}$, $\SA$ with $\mathfrak{a}$ and $\Mmf_{\max}$: for $k< \min\A$, we set $k \in A^{\Mmf}$ iff $k\in A^{\Mmf_{\min}}$; for~$k\in \tem(\A)$, we set $k \in A^{\Mmf}$ iff $k\in \mathfrak{a}(A^{\ast})$; and for $k> \max\A$, we set $k \in A^{\Mmf}$ iff  $k\in A^{\Mmf_{\max}}$. Then $\Mmf$ is a model of~$(\TO,\A)$ but $\ell\not\in \varkappa_0^{\Mmf}$, and so $\ell\notin \ans(\q,\A)$.
\qed
\end{proof}

The next theorem establishes a matching $\NCo$ lower bound for $\LTL_\horn\Xnext$ \OMAQ{}s as well as $\LTL_\krom\Xbox$ and $\LTL_\krom\Xnext$ \OMPIQ{}s by generalising Example~\ref{exampleBvNB:2} and using the fact that  there are $\NCo$-complete regular languages~\cite{DBLP:journals/jcss/Barrington89}.
\begin{theorem}\label{ex:NFA}
There exist \textup{(}i\textup{)} an $\LTL_\horn\Xnext$ \OMAQ{} and \textup{(}ii\textup{)} $\LTL_\krom\Xbox$ and $\LTL_\krom\Xnext$ \OMPIQ{}s, the answering problem for which is $\smash{\NCo}$-hard for data complexity.
\end{theorem}
\begin{proof}
Let $\mathfrak A$ be a DFA with a tape alphabet $\Gamma$, a set of states $Q$, an initial state $q_0\in Q$, an accepting state $q_1\in Q$ and a transition function $\to$: we write $q \to_e q'$ if $\mathfrak{A}$ moves to a state $q'\in Q$ from a state $q\in Q$ while reading $e\in\Gamma$. (Without loss of generality we assume that $\mathfrak A$ has only one accepting state.)  We take atomic concepts $A_e$ for  tape symbols $e\in\Gamma$ and atomic concepts $B_q$ for states $q\in Q$, and consider the \OMAQ{} $ \q = (\TO,B_{q_1})$, where
\begin{equation*}
\TO \ \ = \ \ \bigl\{\,\Lnext B_{q} \land A_e \to B_{q'} \mid q \to_e q'\,\bigr\}.
\end{equation*}
For any input word $\avec{e}=(e_1\dots e_{n})\in \Gamma^*$, we set
\begin{equation*}
\A_{\avec{e}} \ \ = \ \  \bigl\{\,B_{q_0}(0)\,\bigr\} \ \ \cup \ \ \bigl\{\,A_{e_i}(i) \mid 0 < i \leq n\,\bigr\}.
\end{equation*}
It is easy to see that $\mathfrak A$ accepts $\avec{e}$ iff
$\max(\A_{\avec{e}}) \in \ans(\q,\A_{\avec{e}})$. Thus, we obtain
(\emph{i}).
For (\emph{ii}), we take $\q' = (\TO',\varkappa')$ with
\begin{align*}
 \TO' \ \ &= \ \ \bigl\{ \, \overline{B}_q \land B_q \to \bot, \ \top \to \overline{B}_q \lor B_q  \ \mid \ q \in Q \, \bigr\}, &
\varkappa' \ \ &= \ \ \Bigl[\bigvee_{q\to_e q'} \Ldiamond^+ ( \Lnext B_{q} \land A_e \land \overline{B}_{q'})\Bigr] \ \ \lor \ \ B_{q_1} ,
\end{align*}
where $\Ldiamond^{\scriptscriptstyle+} C$ is an abbreviation for $C \lor \Ldiamond C$.
(Intuitively, $\overline{B}_q$ represents the complement of $B_q$, and $\varkappa'$\nz{was $\varkappa$} is equivalent to formula $\bigl[\,\bigwedge_{q\to_e q'} \Lbox^{+} (\Lnext B_{q} \land A_e \to B_{q'})\,\bigr] \to B_{q_1}$, where $\Lbox^+C$ is an abbreviation for $C\land \Lbox C$.) It follows that $\mathfrak A$ accepts~$\avec{e}$ iff $\max(\A_{\avec{e}}) \in \ans(\q',\A_{\avec{e}})$.
\qed
\end{proof}

We now establish the $\FO(<)$- and $\FO(<,+)$-rewritability results in Table~\ref{LTL-table} for various subsets of $\LTL\Xallop_\bool$ \OMPIQ{}s. We begin with ontology-mediated \emph{atomic} queries (\OMAQ{}s) and describe two types of automata-based constructions of $\FO(<)$- and $\FOE$-rewritings.

\section{$\LTL_{\smash{\bool}}^\Box$ \OMAQ{}s: Partially Ordered Automata}
\label{sec:ltl-tomaq:stutter}

Our first rewriting technique for \OMAQ{}s is based on the NFA construction by~\citeauthor{VardiW86}~\cite{VardiW86}. Consider  an $\LTL_\bool\Xallop$ ontology $\TO$. We define an NFA~$\nfao$ that recognises data instances~$\A$ such that $(\TO,\A)$ is consistent. The data instances  are represented as words of the form
$X_{\min \A}, X_{\min \A + 1},\dots, X_{\max \A}$, where
\begin{equation*}
X_i = \bigl\{\, B \mid B(i) \in \A  \text{ and } B \text{ occurs in } \TO \,\bigr\},\qquad \text{ for } i \in\tem(\A).
\end{equation*}
Thus, the alphabet of $\nfao$ comprises all the subsets of atomic concepts $B$ that occur in $\TO$. To define its states, denote by $\subo$ the set of subconcepts of $\TO$ and their negations. Then the set $\tpset$ of states of the automaton~$\nfao$ is the set of types for $\TO$, that is, maximal subsets $\tp$ of $\subo$ consistent with $\TO$.  Note that every $\tp\in \tpset$ contains either $B$ or $\neg B$, for each atomic concept $B$ in $\TO$.  Finally, for any states $\tp,\tp' \in \tpset$ and an alphabet symbol~$X$, the NFA $\nfao$ has a transition $\tp \to_X \tp'$ just in case the following conditions hold:
\begin{align}
\label{eq:path:abox}
& X \subseteq \tp',\\
\label{eq:path:next-consistency}
& \Rnext C \in \tp  ~\text{ iff }~ C \in \tp', && \Lnext C \in \tp' ~\text{ iff }~ C \in \tp,\\
\label{eq:path:box-monotonicity}
& \Rbox C \in \tp  ~\text{ iff }~ C, \Rbox C \in \tp', && \Lbox C \in \tp' ~\text{ iff }~ C, \Lbox C \in \tp.
\end{align}
Clearly, $\tp\to_X\tp'$ implies $\tp\to_\emptyset\tp'$, for all $X$, and so we omit the $\emptyset$ subscript in $\to_{\emptyset}$ in the sequel.
Since all $\tp$ in $\tpset$ are consistent with $\TO$,  every state in $\nfao$ has a $\to$-predecessor and a $\to$-successor, and all states in $\nfao$ are initial and accepting. We say that a sequence
\begin{equation}\label{path}
\pi \ \ = \ \  \tp_0\to \tp_1\to \ldots \to \tp_{m-1} \to \tp_m
\end{equation}
of states in $\nfao$  (also called a \emph{path})
\emph{accepts} a word~$X_0,X_1,\dots,X_m$ if $X_i \subseteq \tp_i$, for all $i$ with $0 \leq i \leq m$, which means that the NFA~$\nfao$ contains a path $\tp_{-1}\to_{X_0} \tp_0\to_{X_1} \tp_1\to_{X_2} \ldots \to_{X_{m-1}} \tp_{m-1} \to_{X_m} \tp_m$, for some $\tp_{-1}\in\tpset$.

Let $\q = (\TO,A)$ be an $\LTL\Xallop_\bool$ \OMAQ{}, for an atomic concept $A$ that occurs in $\TO$. It should be clear now that, for any data instance~$\A$ with $\min\A = 0$ represented as $X_0,X_1,\dots,X_m$ and any $\ell$, $0 \leq \ell\leq m$, we have $\ell \notin \ans(\q,\A)$ iff  there exists a path $\tp_0\to\ldots\to\tp_m$ in $\nfao$ accepting $X_0,X_1,\dots,X_m$ with $A \notin \tp_\ell$ (which is just another way of saying that there is a model of~$(\TO,\A)$ where~$A$ does not hold at $\ell$). This criterion can be encoded by an \emph{infinite} FO-expression
\begin{equation*}
\Psi(t) \ \ = \ \
\neg \Bigl[\bigvee_{\begin{subarray}{c}\tp_0 \to \ldots\to \tp_m\\\text{is a path in } \nfao\end{subarray}} \ \ \Bigl(
\bigwedge_{0 \leq i \leq m} \typef_{\tp_i}(i) \ \ \land \
\bigvee_{\begin{subarray}{c}0 \leq i \leq m\\A \notin\tp_i\end{subarray}} (t = i)\Bigr)\Bigr],
\end{equation*}
where the disjunction is over all (possibly infinitely many) paths, and $\typef_\tp(t)$  is a conjunction of all
$\neg B(t)$ with~$B \notin \tp$, for atomic concepts $B$ in $\TO$: the first conjunct in $\Psi(t)$ ensures, by contraposition, that any $B$ from $X_i$ also belongs to~$\tp_i$, for all~$i$, and so the path $\tp_0\to\ldots\to\tp_m$ accepts $X_0,\dots,X_m$, while the second conjunct guarantees that $A\notin\tp_\ell$ in case~$\ell = t$.
Needless to say that, for some OMAQs $\q$, the infinitary `FO-rewriting' above cannot be made  finite.

We now show that, for every $\LTL_\bool^\Box$ ontology $\TO$, the NFA
$\nfao$ constructed above can be converted into an equivalent
\emph{partially ordered} NFA $\nfao^\circ$, that is, an NFA
whose transition relation contains no cycles other than
trivial self-loops~\cite{DBLP:conf/dlt/SchwentickTV01}, which will
allow us to \emph{finitely} represent the infinitary `rewriting'
$\Psi(t)$.
\begin{theorem}\label{thm:box}
All $\LTL_\bool^\Box$ \OMAQ{}s are $\FO(<)$-rewritable.
\end{theorem}
\begin{proof}
Let $\q = (\TO,A)$ be an $\LTL_\bool^\Box$ \OMAQ{} and $\nfao$ the NFA constructed above.
Define an equivalence relation~$\sim$ on the set $\tpset$ of states of $\nfao$ by taking $\tp \sim \tp'$ iff $\tp = \tp'$ or $\nfao$ has a cycle through both $\tp$ and $\tp'$. Denote by $[\tp]$ the~$\sim$-equivalence class of $\tp\in\tpset$. It is readily seen that, if $\tp \sim \tp'$, then $\tp$ and $\tp'$ contain the same boxed concepts:
\begin{equation*}
\Rbox C \in \tp \ \ \text{ iff }\ \  \Rbox C \in \tp' \qquad\text{ and }\qquad \Lbox C \in \tp \ \ \text{ iff } \ \ \Lbox C \in \tp'.
\end{equation*}
Moreover, if $\Rbox C\in \tp$ (or, equivalently, $\Lbox C\in \tp$) and $\nfao$ contains a cycle through elements of $[\tp]$ (equivalently, if $[\tp]$ has at least two elements), then $C\in \tp'$ for all $\tp'\sim \tp$. It follows that
\begin{equation}\label{path-correctness}
\tp_1 \to_X \tp_2 \ \ \text{ iff } \ \ \tp_1' \to_X \tp_2, \qquad \text{ for all } \tp_1 \sim \tp_1',
\end{equation}
and, as observed above, if $\nfao$ contains a cycle through $[\tp]$, then $\tp \to_X \tp$ iff $\tp' \to_X \tp$, for any $\tp' \sim \tp$.
Denote by $\nfao^\circ$ an NFA with the states~$[\tp]$, for $\tp \in \tpset$, and transitions\nz{was $\tp_1' \sim\tp_2$}
\begin{equation*}
[\tp_1] \to_X [\tp_2] \quad\text{ iff }\quad \tp_1' \to_X \tp'_2, \text{ for some } \tp_1' \sim \tp_1 \text{ and } \tp'_2 \sim \tp_2.
\end{equation*}
Again, all states in $\nfao^\circ$ are initial and accepting. Clearly, the NFA $\nfao^\circ$ is partially ordered and accepts the same language as $\nfao$: for any path $\pi$ of the form~\eqref{path} in $\nfao$ accepting $X_0,X_1,\dots,X_m$,  the sequence
\begin{equation}\label{path'}
[\pi] \ \ = \ \  [\tp_0] \to [\tp_1]  \to \ldots \to  [\tp_{m-1}] \to [\tp_{m}]
\end{equation}
of states is a path in $\nfao^\circ$ accepting the same word $X_0,X_1,\dots,X_m$; the converse is due to~\eqref{path-correctness}.

Given a path $[\pi]$ of the form~\eqref{path'}, we now select the longest possible sequence of indexes
\begin{equation*}
0 = s_0 < s_1 < \dots < s_{d-1} < s_d = m,
\end{equation*}
where each $s_i$, for $i < m$, is the smallest $j$ with $\tp_j \sim \tp_{s_i}$. The respective  sequence
$[\tp_{s_0}] \to [\tp_{s_1}] \to \ldots \to [\tp_{s_{d-1}}] \to [\tp_{s_d}]$
of states, by~\eqref{path-correctness}, is also a path in $\nfao^\circ$ (we could, equivalently, remove from $[\pi]$ all $[\tp_j]$ for which there is $j' < j$ with~$\tp_{j'} \sim \tp_j$). Observe that, for any $\Lbox C\in\subo$, if $\Lbox C, \neg C\in \tp_i$ in a path of the form~\eqref{path'}, then $\Lbox C, C\in\tp_j$ for all~$j < i$, and $\Lbox C\notin\tp_j$ for all $j > i$; and symmetrically for $\Rbox C$. On the other hand, given two states $\tp$ and $\tp'$, if either~$\Box C, C\in\tp,\tp'$ or $\Box C\notin\tp,\tp'$, for each $\Box C\in\subo$,  $\Box\in\{\Lbox, \Rbox\}$, then there is a cycle through $\tp$ and $\tp'$. Indeed, by~\eqref{eq:path:box-monotonicity}, we have $\tp\to\tp'$ and $\tp'\to\tp$, and so $\tp\sim\tp'$; we call such an equivalence class~$[\tp]$ a \emph{loop}. Therefore, $[\pi]$ comprises at most~$|\TO|$ equivalence classes, each containing a pair of the form $\Box C,\neg C$, possibly separated by repeated loops: more precisely, for all $i$ with $0 \leq i < d$, we have
\begin{equation}\label{eq:path:extension}
\tp_j \to \tp_{s_i} \text{ and }\tp_{s_i} \to \tp_j, \text{ for all } j \text{ with } s_i < j < s_{i+1}.
\end{equation}
Thus, by selecting the indexes $s_0,s_1,\dots, s_{d-1}$, we eliminate loop repetitions;  the resulting paths will be called \emph{stutter-free}.
In stutter-free paths, the number of distinct loops  does not exceed $|\TO|+1$ and the number of non-loops is at most~$|\TO|$, and so we obtain $d\leq 2|\TO| + 1$.

\begin{example}\label{ex:NFAo}\em
Suppose $\TO = \{\, C \to  \Lbox B, \, \Lbox B \to A \,\}$ and $\q = (\TO,A)$.
\begin{figure}[t]
\vspace*{-3mm}\centering%
\begin{tikzpicture}[>=latex,semithick,
 state/.style={rectangle, rounded corners=2mm, fill=gray!40,draw=black,text=black,inner sep=2pt,minimum width=30mm},
 substate/.style={rectangle, rounded corners=1mm, fill=gray!10,draw=black,text=black,inner xsep=2pt,inner ysep=4pt,minimum width=20mm}]\scriptsize
  \node[state,minimum height=23mm] (A) at (-0.25,-0.1)  {};
  \node[substate] (A1) at (0,0.4) {$ A\ B\ C\ \Lbox B$};
  \node[substate] (A2) at (0,-0.8) {$A\ B\ \Lbox B$};
  \node[state,minimum height=9mm,minimum width=24mm]  (B) at (3.7,-0.8) {};
  \node[substate]  at (3.7,-0.8) {$A\ \Lbox B$};
  \node[state,minimum height=9mm,minimum width=24mm]  (Bp) at (3.7,0.6) {};
  \node[substate]  at (3.7,0.6) {$A\ C\ \Lbox B$};
  \node[state,minimum height=23mm,minimum width=70mm]   (C) at (10,-0.1) {};
  \node[substate]  (C1) at (8.1,0.5) {$A\ B$};
  \node[substate]  (C2) at (8.1,-0.7) {$B$};
  \node[substate]  (C3) at (11.9,0.5) {$A$};
  \node[substate]  (C4) at (11.9,-0.7) {$\emptyset$};
  \draw[->] (A) to node[below] {$A$} (B);
  \draw[->] (A) to node[above] {$AC$} (Bp);
  \draw[->,out=60,in=120,looseness=2] (A) to node[below,midway] {$ABC$} (A);
  \draw[->] (B) to node[below] {$AB$} (C);
  \draw[->] (Bp) to node[above] {$AB$} (C);
  \draw[->,out=60,in=120,looseness=2] (C) to node[below,midway] {$AB$} (C);
  \begin{scope}[ultra thin]\tiny\color{black!70}
  \draw[->,bend right] (A1) to node[left] {$AB$} (A2);
  \draw[->,out=150,in=210,looseness=2] (A1.north west) to node[below,sloped] {$ABC$} (A1.south west);
  \draw[->,bend right] (A2) to node[right] {$ABC$} (A1);
  \draw[->,out=150,in=210,looseness=2] (A2.north west) to node[below,sloped] {$AB$} (A2.south west);
  \draw[->,out=150,in=210,looseness=2] (C1.north west) to node[below,sloped] {$AB$} (C1.south west);
  \draw[->,out=150,in=210,looseness=2] (C2.north west) to node[below,sloped] {$B$} (C2.south west);
  \draw[->,out=30,in=-30,looseness=2] (C3.north east) to node[above,sloped] {$A$} (C3.south east);
  \draw[->,out=30,in=-30,looseness=2] (C4.north east) to node[above,sloped] {} (C4.south east);
  \draw[->,bend right] (C1) to node[left] {$B$} (C2);
  \draw[->,bend right] (C2) to node[right] {$AB$} (C1);
  \draw[->,bend right] (C3) to node[left] {} (C4);
  \draw[->,bend right] (C4) to node[right] {$A$} (C3);
  \draw[->,bend left] (C1.east) to node[above] {$A$} (C3.west);
  \draw[->,bend left] (C3.west) to node[above] {$AB$} (C1.east);
  \draw[->,bend left] (C2.east) to node[below] {} (C4.west);
  \draw[->,bend left] (C4.west) to node[below] {$B$} (C2.east);
  \draw[->] (C2) to (C3.south west);
  \draw[->] (C3) to (C2.north east);
  \draw[->] (C1) to (C4.north west);
  \draw[->] (C4) to (C1.south east);
  \end{scope}
\end{tikzpicture}
\caption{The NFA $\nfao^\circ$ in Example~\ref{ex:NFAo}.}\label{fig:NFAo}
\end{figure}
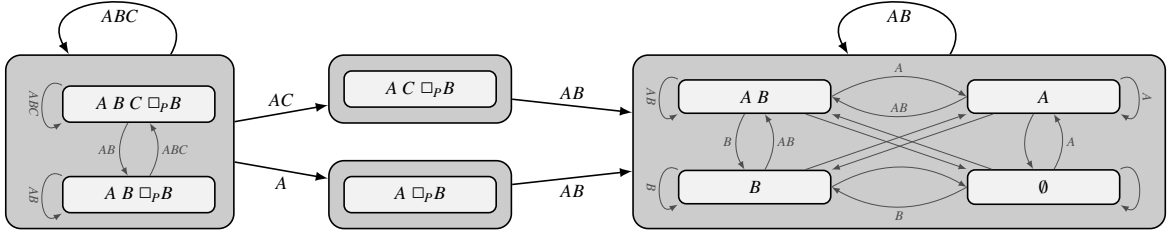
Consider an input $\{A,C\}, \{A,B\}, \{C\}, \emptyset, \{B\}, \emptyset, \emptyset, \{A\}$ and one of the paths in the NFA $\nfao$ accepting the input:
\begin{equation*}
\pi \ \ = \ \ \{A,B,C,\Lbox B\} \to \{A,B,\Lbox B\} \to \{A,B,C,\Lbox B\} \to \{A,C,\Lbox B\} \to \{A,B\}  \to \{A\}  \to \emptyset  \to \{A,B\},
\end{equation*}
where a set $X$ represents the type comprising all $B\in X$ and all $\neg B$, for $B \in\subo\setminus X$.
The NFAs $\nfao$ and $\nfao^\circ$ are shown in Fig.~\ref{fig:NFAo}: the 8 states of $\nfao$ are smaller light-gray nodes, and the 3 states of $\nfao^\circ$ are larger dark-gray nodes. An arrow labelled by $X$ denotes a set of transitions, each of which corresponds to reading any \emph{subset} of~$X$. As we mentioned above, all the states are initial and accepting; only positive literals are shown in the labels of states, and some labels on the transitions of $\nfao$ are not shown to avoid clutter.
The respective path in $\nfao^\circ$ accepting the input is
\begin{equation*}
[\pi] \ \ = \ \ [A,B,\Lbox B] \to \underline{[A,B,\Lbox B]} \to \underline{[A,B,\Lbox B]} \to [A,C,\Lbox B] \to [\emptyset]  \to \underline{[\emptyset]} \to \underline{[\emptyset]} \to [\emptyset].
\end{equation*}
By removing repeated loops (underlined above), we obtain a stutter-free path
$[A,B,\Lbox B] \to [A,C,\Lbox B]    \to [\emptyset] \to [\emptyset]$ (note that the last loop occurs twice because $s_d$ is always $m$). In fact, every stutter-free accepting path in $\nfao^\circ$ is a subsequence of either this path or $[A,B,\Lbox B] \to [A,\Lbox B]    \to [\emptyset] \to [\emptyset]$.

Consider now an extension $\TO'$ of $\TO$ with $D \to \Rbox E$ and $\Rbox E \to F$. It should be clear that the corresponding NFA~$\mathfrak{A}_{\TO'}$ contains 64 states. The stutter-free accepting paths in $\mathfrak{A}_{\TO'}^\circ$ will now have to choose whether a state with~$\Lbox B,\neg B$ occurs (\emph{i}) before, (\emph{ii}) at the same moment or (\emph{iii}) after a state with $\Rbox E, \neg E$, which gives rise to stutter-free accepting paths of length not exceeding~6: for example,
\begin{equation*}
[A,B,\Lbox B] \to [A,C,\Lbox B]  \to [\emptyset] \to [F,\Rbox E]   \to [E,F,\Rbox E] \to [E,F,\Rbox E].
\end{equation*}
\end{example}

We return to the proof of Theorem~\ref{thm:box} now. The criterion for certain answers given above in terms of $\nfao$ can be reformulated in terms of $\nfao^\circ$ by using the observation that all accepting paths in $\nfao^\circ$ have a `bounded' representation: each accepting path in $\nfao^\circ$ is determined by a stutter-free path in $\nfao^\circ$ and a corresponding sequence of indexes $0 = s_0 < s_1 < \dots < s_{d-1} < s_d = m$ because the path elements between the chosen indexes are given by~\eqref{eq:path:extension}. It should be clear that, in fact, we can consider all paths in~$\nfao^\circ$ of length not exceeding $2|\TO|+1$ (rather than stutter-free paths only); moreover, it will be convenient to include an additional type for the time instant $\ell$. So, for any data instance~$\A$ and any~$\ell \in \tem(\A)$, we have $\ell \notin \ans(\q,\A)$ iff there are $d \leq 2|\TO| + 2$, a path $[\tp_0] \to [\tp_1] \to \ldots \to [\tp_d]$ in $\nfao^\circ$ and a sequence $\min \A = s_0 < s_1 < \dots < s_{d-1} < s_d = \max \A$  satisfying the following conditions:
\begin{align}
\label{eq:path-conditions:1}
& X_{s_i} \subseteq \tp, \text{ for some } \tp \sim \tp_i, && \text{for } 0 \le i \le d; \\
\label{eq:path-conditions:2}
& X_j \subseteq \tp, \text{ for some } \tp \text{ with } \tp\to\tp_i \text{ and } \tp_i \to \tp,   && \text{for } s_i < j < s_{i+1}  \text{ and } 0 \le i < d;\\
\label{eq:path-conditions:3}
& \ell = s_i \text{ for some } 0 \leq i \leq d  \text{ and some } \tp \sim \tp_i \text{ with } A \notin \tp.
\end{align}
Conditions~\eqref{eq:path-conditions:1} and~\eqref{eq:path-conditions:2} ensure that $X_0,\dots,X_m$ is accepted by the path obtained from $[\tp_0] \to [\tp_1] \to \ldots \to [\tp_d]$  by placing each $[\tp_i]$ at position $s_i$ and filling in the gaps according to~\eqref{eq:path:extension}; condition~\eqref{eq:path-conditions:2} guarantees, in particular,  that any $[\tp_i]$ that needs to be repeated is a loop. Finally, condition~\eqref{eq:path-conditions:3} says that the type at position $\ell$, which must be one of the~$s_i$, does not contain $A$.
We now encode the new criterion by an $\FO(<)$-formula: let
\begin{equation*}
\rq(t) ~=~  \neg  \Bigl[\bigvee_{d \leq 2|\TO| + 2} \ \ \bigvee_{\begin{subarray}{c}[\tp_0]\to\ldots\to[\tp_d]\\\text{is a path in }\nfao^\circ\end{subarray}} \exists
  t_0,\dots,t_d\, \Bigl(\pathf_{[\tp_0]\to\ldots\to[\tp_d]}(t_0, \dots, t_d)  \ \ \land \ \ \bigvee_{\begin{subarray}{c}0\leq i \leq d\\A\notin\tp \text{ for some } \tp
   \sim\tp_i\end{subarray}}\hspace*{-1.5em} (t = t_i) \Bigr)\Bigr],
\end{equation*}
where the formula $\pathf_{[\tp_0]\to\ldots\to[\tp_d]}$ encodes conditions~\eqref{eq:path-conditions:1} and \eqref{eq:path-conditions:2}:
\begin{multline*}
\pathf_{[\tp_0]\to\ldots\to[\tp_d]}(t_0, \dots,t_d)  \ \ = \ \ (\min = t_0)  \ \land \bigwedge_{0 \leq i < d} (t_i < t_{i+1}) \ \land \ (t_d = \max) \ \ \land {}\\
\bigwedge_{0 \leq i \leq d}  \ \ \bigvee_{\tp \sim \tp_i} \typef_{\tp}(t_i) \ \ \ \land \ \
\bigwedge_{0 \leq i < d} \forall t\, \bigl((t_i < t < t_{i+1}) \to \bigvee_{\tp \to \tp_i \to \tp} \typef_{\tp}(t) \bigr).
\end{multline*}
Observe the similarity of this rewriting to the infinitary `rewriting' $\Psi(t)$ constructed above.
\qed
\end{proof}

We illustrate the construction of $\rq(t)$ by the following example.

\begin{example}\em
In the context of Example~\ref{ex:NFAo}, the accepting path $[\pi]$ can be decomposed into a stutter-free path $[\pi]^\circ$ and a sequence of indexes in the following way:
\begin{align*}
[\pi] \  &=&  [A,B,\Lbox B] & \to& [A,B,\Lbox B] & \to& [A,B,\Lbox B] & \to& [A,C,\Lbox B] & \to& [\emptyset]  & \to& [\emptyset] & \to& [\emptyset] & \to& [\emptyset],\\
[\pi]^\circ \  &=& [A,B,\Lbox B] & \to&  &&  && [A,C,\Lbox B] & \to& [\emptyset]  & \to&  &&  && [\emptyset],\\
&& s_0 = 0 &  < & && && s_1 = 3 &  < & \hspace*{-0.5em}s_2 = 4 &  < & && && \hspace*{-1em}s_3 = 7.
\end{align*}
Note that $[A,C,\Lbox B]$ is not a loop, whereas $ [A,B,\Lbox B]$ and $[\emptyset]$ are loops. So, if $\ell$ is one of these $s_i$, then we can use
\begin{align*}
\pathf_{[\pi]^\circ}(t_0, t_1, t_2,t_3) \  = \ \ & (\min = t_0) \land (t_0 < t_1) \land (t_1 < t_2) \land (t_2 < t_3)\land (t_3 = \max) \land{}\\
&  \typef_{[AB\Lbox B]}(t_0) \land \typef_{[AC\Lbox B]}(t_1) \land \typef_{[\emptyset]}(t_2) \land \typef_{[\emptyset]}(t_3) \land{} \\
& \forall t\,\bigl((t_0 < t < t_1) \to \typef_{[AB\Lbox B]}(t)\bigr) \  \land \
\forall t\,\neg(t_1 < t < t_2) \ \land \
 \forall t \, \bigl((t_2 < t < t_3) \to  \typef_{[\emptyset]}(t)\bigr),
 \end{align*}
where $\typef_{[X]}$ abbreviates $\bigvee_{\tp\sim X}\typef_{\tp}$. If $\ell$ is not one of the selected $s_i$, then we decompose $[\pi]$ into a (non-stutter-free) path and a sequence of indexes with 5 elements, which effectively means that either the extra state $[A,B,\Lbox B]$ is inserted into $[\pi]^\circ$ between indexes 0 and 3, or the extra state $[\emptyset]$ is inserted into $[\pi]^\circ$  between indexes~4 and~7. Since the paths in~$\nfao^\circ$ we need to consider have at most 5 states, formula $\rq(t)$ is an $\FO(<)$-rewriting of the OMAQ.
\end{example}


\section{$\LTL_{\smash{\krom}}\Xnext$ \OMAQ{}s: Unary Automata and Arithmetic Progressions}
\label{sec:ltl-tomaq:unary}

We use unary NFAs, that is, automata with a single-letter alphabet, to construct rewritings for $\LTL_\krom\Xnext$ \OMAQ{}s.

\begin{theorem}\label{thm:krom-next}
All $\LTL_\krom\Xnext$ \OMAQ{}s are $\FOE$-rewritable.
\end{theorem}
\begin{proof}
Suppose $\q = (\TO,A)$ is an $\LTL_\krom\Xnext$ \OMAQ{} such that $A$
occurs in $\TO$. As observed in Section~\ref{sec:expressivity},
without loss of generality, we can assume that $\TO$ does not contain
nested $\Rnext$ and $\Lnext$.
A \emph{simple literal}, $L$, for $\TO$ is an atomic concept from $\TO$ or its negation; we set $\neg\neg L = L$. We also use $\nxt^n L$ to abbreviate $\Rnext^n L$ if $n > 0$, $L$ if $n = 0$, and $\Lnext^{-n} L$ if $n < 0$.

It is known that checking satisfiability of a 2CNF (a propositional Krom formula) boils down to finding cycles in a directed graph whose nodes are literals and edges represent clauses  (see, e.g.,~\cite[Lemma~8.3.1]{Borgeretal97}). To deal with $\LTL_\krom\Xnext$, we need an infinite graph with timestamped literals. More precisely, denote by $G_\TO$ the infinite directed graph whose vertices are pairs $(L, n)$, for a simple literal $L$ for $\TO$ and $n \in \Z$, and which contains an edge from $(L,n)$ to~$(L',n+k)$, for $k \in \{-1,0,1\}$, iff $\TO \models L \to \nxt^k L'$.
We write $(L_1,n_1)  \leadsto (L_2,n_2)$ iff $G_\TO$ has a (possibly empty) path from $(L_1,n_1)$ to $(L_2,n_2)$. In other words, $\leadsto$ is the reflexive and transitive closure of the edge relation in $G_\TO$.
It is readily seen that $(L_1,n_1)  \leadsto (L_2,n_2)$ is another way of saying that $\TO \models \nxt^{n_1} L_1 \to \nxt^{n_2} L_2$.
We require the following key properties distinguishing $\LTL_\krom\Xnext$ (cf.~\cite[Theorem~5.1]{DBLP:journals/tocl/ArtaleKRZ14}): if $\TO$ is consistent, then, for any data instance $\A$,
\begin{description}
\item[\rm (\emph{i})] if $\A$  is consistent with $\TO$, then $\ell \in \ans^\Z(\q,\A)$ iff $(B,n) \leadsto (A,\ell)$, for some $B(n) \in \A$;

\item[\rm (\emph{ii})] $\A$ is inconsistent with $\TO$ iff $(B,n) \leadsto (\neg B',n')$, for some $B(n),B'(n') \in \A$.
\end{description}
(Observe that,  if $\TO\models  \top \to A$, then $G_\TO$ contains an edge to $(A,\ell)$ from any $(B,n)$, and so the right-hand side of~(\emph{i}) is trivially satisfied because data instances $\A$ are nonempty.)
We leave the proof of these properties to the reader as an exercise (see also the proof of a more general  Lemma~\ref{lem:kromcond}).

Given literals $L$ and $L'$, let $\mathfrak A_{L,L'}$ be an NFA whose tape alphabet is $\{0\}$, states are the simple literals, with $L$ initial and $L'$ accepting, and there is a transition $L_1 \to_0 L_2$ iff $\TO\models L_1 \to \Rnext L_2$ (as we agreed at the end of Section~\ref{sec:expressivity},\nz{changed} $\TO$ does not contain nested temporal operators). It is easy to see that, for any distinct $n,n'\in\Z$,
\begin{equation*}
(L,n) \leadsto (L',n')\quad \text{ iff }\quad
\begin{cases}
\mathfrak A_{L,L'} \text{ accepts } 0^{n' - n}, &  \text{ for } n < n',\\
\mathfrak A_{\neg L', \neg L} \text{ accepts } 0^{n - n'}, &  \text{ for } n > n';
\end{cases}
\end{equation*}
the latter is the case because $(L,n) \leadsto (L',n')$ iff $(\neg L',n') \leadsto (\neg L,n)$.
We can now use the results on the normal form of unary finite automata~\cite{chrobak-ufa,to-ufa}, according to which, for every \emph{unary} NFA $\mathfrak
A_{L,L'}$, there are $N = O(|\mathfrak A_{L,L'}|^2)$ arithmetic
progressions $a_i + b_i\mathbb{N} = \bigl\{\, a_i + b_i \cdot m \mid m\geq 0
\,\bigr\}$, $1 \leq i \leq N$, such that $0 \leq a_i, b_i \leq |\mathfrak
A_{L,L'}|$ and
\begin{equation*}
\mathfrak{A}_{L,L'}
\text{ accepts $0^n\ (n>0)$}\quad \text{ iff }\quad n \in a_i + b_i\mathbb{N} \ \ \text{ for some }
1\leq i \leq N.
\end{equation*}

\begin{example}\label{ex:unaryFA}\em
Suppose $\q = (\TO,A)$ and
\begin{equation*}
\TO = \bigl\{ \, A \to \Rnext B, \  \ B \to \Rnext C, \ \ C \to \Rnext D, \ \ D \to \Rnext A,
\ \
D \to \Rnext E, \ \ E \to \Rnext D \, \bigr\}.
\end{equation*}
The NFA $\mathfrak A_{B,A}$ (the states reachable from $B$, to be more precise) is shown in Fig.~\ref{fig:unaryFA}a. (For $L \in \{A,C,D,E\}$, the NFA~$\mathfrak A_{L,A}$ is obtained from $\mathfrak A_{B,A}$ by shifting the initial state to $L$; see Fig.~\ref{fig:unaryFA}b for $L=E$.)
\begin{figure}[t]
\centering%
\begin{tikzpicture}[>=latex,xscale=2.5,yscale=1,thick,
state/.style={circle,draw,fill=gray!50,thick,minimum size=4mm,inner sep=0pt}]\small
\node at (-0.5,1) {a)};
\node[state,label=below:{$B$}] (B) at (0,0) {};
\node[state,label=below:{$A$},fill=white] (A) at (1,0) {};
\node[state,label=above:{$C$}] (C) at (0,1) {};
\node[state,label=above:{$D$}] (D) at (1,1) {};
\node[state,label=below:{$E$}] (E) at (1.75,1) {};
\node[circle,inner sep=0pt,draw,minimum size=2.5mm,thick,fill=gray!50] at (A) {};
\fill (B.west) -- ++(-0.1,0.1) -- ++(0,-0.2) -- cycle;
\draw[->] (B) -- (C);
\draw[->] (C) -- (D);
\draw[->] (D) -- (A);
\draw[->] (A) -- (B);
\draw[->,out=60,in=120] (D) to (E);
\draw[->,in=-60,out=-120] (E) to (D);
\node at (1.5,-0.25) {\normalsize $3 + 2\mathbb{N}$};
\end{tikzpicture}
\hspace*{8em}
\begin{tikzpicture}[>=latex,xscale=2.5,yscale=1,thick,
state/.style={circle,draw,fill=gray!50,thick,minimum size=4mm,inner sep=0pt}]\small
\node at (-0.5,1) {b)};
\node[state,label=below:{$B$}] (B) at (0,0) {};
\node[state,label=below:{$A$},fill=white] (A) at (1,0) {};
\node[state,label=above:{$C$}] (C) at (0,1) {};
\node[state,label=above:{$D$}] (D) at (1,1) {};
\node[state,label=below:{$E$}] (E) at (1.75,1) {};
\node[circle,inner sep=0pt,draw,minimum size=2.5mm,thick,fill=gray!50] at (A) {};
\fill (E.east) -- ++(0.1,0.1) -- ++(0,-0.2) -- cycle;
\draw[->] (B) -- (C);
\draw[->] (C) -- (D);
\draw[->] (D) -- (A);
\draw[->] (A) -- (B);
\draw[->,out=60,in=120] (D) to (E);
\draw[->,in=-60,out=-120] (E) to (D);
\node at (1.5,-0.25) {\normalsize$2 + 2\mathbb{N}$};
\end{tikzpicture}
\caption{NFAs $\mathfrak{A}_{B,A}$ and $\mathfrak{A}_{E,A}$  in Example~\ref{ex:unaryFA}.}\label{fig:unaryFA}
\end{figure}
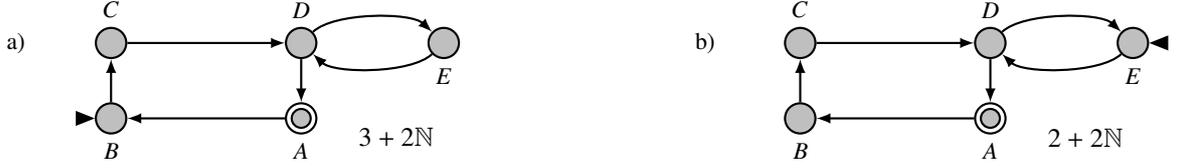
It is readily seen that $\smash{\mathfrak A_{B,A}}$ accepts~$0^n$ iff
$n \in 3 + 2 \mathbb{N}$, which means that all certain answers to $(\TO,A)$ satisfy the formula
$\exists s\, \bigl(B(s) \ \land \ \bigl(t - s \in 3 + 2\mathbb N\bigr)\bigr)$.
Similarly, for $\mathfrak A_{E, A}$, we need $\exists s\,\bigl(E(s) \land \bigl(t - s \in 2 + 2\mathbb N\bigr)\bigr)$. (In general, more than one progression is required to characterise the NFAs $\mathfrak A_{L,L'}$.)
\end{example}

Returning to the proof of Theorem~\ref{thm:krom-next}, let
\begin{equation*}
\entailsf_{L,L'}(s,t) \ \ = \ \ \entailsf^{\smash{>0}}_{L,L'}(s,t) \ \ \lor \ \ \entailsf^{\smash{0}}_{L,L'}(s,t) \ \ \lor \ \ \entailsf^{\smash{>0}}_{\neg L',\neg L}(t,s),
\end{equation*}
where $\entailsf^0_{L,L'}(s,t)$ is $(s = t)$ if $\TO\models L \to L'$ and $\bot$ otherwise, and  $\entailsf^{>0}_{L,L'}(s,t)$ denotes the disjunction of all
formulas $t -  s \in a_i + b_i\mathbb N$,
for arithmetic progressions $a_i + b_i\mathbb{N}$ associated with $\mathfrak A_{L,L'}$.
Now, conditions (\emph{i}) and~(\emph{ii}) can be encoded as an
\mbox{$\FOE$}-formula $\rq(t)$, which is $\top$ if
$\TO$ is inconsistent (this condition can be checked in \NLogSpace~\cite{Borgeretal97}, independently of the data instance), and otherwise the following formula:
\begin{align*}
\rq(t) \ \ \ \ = \ \bigvee_{B \text{ occurs in } \TO} \hspace*{-1em}\exists s\,\bigl(B(s) \land
\entailsf_{B,A}(s,t) \bigr) \ \ \ \lor\ \ \ \bigvee_{B,B' \text{ occur in } \TO} \hspace*{-1em}\exists s,s' \, \bigl(B(s) \land B'(s') \land \entailsf_{B,\neg B'}(s,s')\bigr).
\end{align*}
It follows immediately that $\rq(t)$ is an $\FOE$-rewriting of $\q$.
\qed
\end{proof}


\section{$\LTL\Xallop_{\smash{\krom}}$ \OMAQ{}s: Unary Automata with Stutter-Free Sequences of Types}
\label{sec:ltl-tomaq:kromboxnext}

We now combine the ideas of the two previous sections to obtain an $\FOE$-rewritability result for $\LTL\Xallop_\krom$ \OMAQ{}s: we separate reasoning with $\Box$-operators (as in Theorem~\ref{thm:box}) from  reasoning with $\nxt$-operators (as in Theorem~\ref{thm:krom-next}); see also a similar separation idea in the proof of~\cite[Theorem~5.1]{DBLP:journals/tocl/ArtaleKRZ14}.
\begin{theorem}\label{thm:krom-nextbox}
All $\LTL\Xallop_\krom$ \OMAQ{}s are $\FOE$-rewritable.
\end{theorem}
\begin{proof}
Let $\q = (\TO,A)$ be an $\LTL\Xallop_\krom$ \OMAQ{} such that $A$ occurs in $\TO$, the ontology $\TO$ has no nested occurrences of temporal operators and contains the axioms $\nxt B \to A_{{\scriptscriptstyle\bigcirc} B}$ and $A_{{\scriptscriptstyle\bigcirc} B} \to \nxt B$, for every $\nxt B$ in $\TO$ and $\nxt \in \{\Rnext,\Lnext\}$.

As before, a type  $\tp$ is a maximal subset of $\subo$ consistent with $\TO$. We write $\tp \to^\Box \tp'$ if $\tp$ and $\tp'$ satisfy~\eqref{eq:path:box-monotonicity}; note, however, that such $\tp$ and $\tp'$ do not have to satisfy~\eqref{eq:path:next-consistency}. Using an argument similar to the observation on stutter-free paths in the proof of Theorem~\ref{thm:box}, one can show the following: given any  path $\tp_0\to \tp_1\to \ldots \to \tp_{m-1} \to \tp_m$  in the NFA~$\nfao$ corresponding to $\TO$, we can extract a subsequence
\begin{equation*}
\tp_{s_0} \to^\Box \tp_{s_1} \to^\Box \ldots \to^\Box \tp_{s_{d-1}} \to^\Box \tp_{s_d}
\end{equation*}
such that $0 = s_0 < s_1 < \dots < s_{d-1} < s_d = m$ for $d \leq 2|\TO|+1$ and, for all $i < d$,
\begin{align}
\text{either } \ \ \Box C, C\in\tp_{s_i},\tp_j \ \ \text{ or } \ \ \Box C \notin\tp_{s_i},\tp_j, && \text{for all } \Box C\in\subo, \Box \in\{\Lbox, \Rbox\},  \text{ and all }  j \text{ with } s_i < j < s_{i+1}.
\end{align}
Observe that the last condition plays the same role as~\eqref{eq:path:extension}, but, because of the $\nxt$-operators,  the NFA $\nfao$ for $\TO$ does not necessarily contains a cycle through $\tp_j$ and $\tp_{s_i}$.
We deal with $\nxt$-operators by means of unary NFAs, as in the proof of Theorem~\ref{thm:krom-next}. Let $\tilde{\TO}$ be the $\LTL\Xnext_\krom$ ontology obtained from $\TO$  by first extending it with the following axioms:
\begin{equation}\label{eq:box-semantics}
\Lbox C \to \Lnext \Lbox C \ \text{ and } \  \Lbox C \to \Lnext C, \ \  \text{ for all } \Lbox C\in\subo,
\quad
\Rbox C \to \Rnext \Rbox C \  \text{ and } \ \Rbox C \to \Rnext C, \ \  \text{ for all } \Rbox C\in\subo,
\end{equation}
which are obvious $\LTL\Xallop_\krom$ tautologies, and then replacing every $\Lbox C$ and $\Rbox C$ with its \emph{surrogate}, which is a fresh atomic concept.
Let~$G_{\tilde{\TO}}$ be the directed graph for $\tilde{\TO}$ defined in the proof of Theorem~\ref{thm:krom-next}. As before, we write $(L_1,n_1) \leadsto (L_2,n_2)$ iff $G_{\tilde{\TO}}$ has a (possibly empty) path from $(L_1,n_1)$ to $(L_2,n_2)$. Note, however, that the $L_i$ can contain surrogates, and we slightly abuse notation and write, for example, $L\in\tp$ in case $L$ is the surrogate for $\Lbox C$ and type~$\tp$ contains~$\Lbox C$.
The following lemma generalises $(i)$ and $(ii)$ in the proof of Theorem~\ref{thm:krom-next} (cf.~conditions~\eqref{eq:path-conditions:1}--\eqref{eq:path-conditions:3} in the proof of Theorem~\ref{thm:box}):

\begin{lemma}\label{lem:kromcond}
For any data instance~$\A$ and any~$\ell \in \tem(\A)$, we have $\ell \notin \ans(\q,\A)$ iff there are $d \leq 2|\TO| + 2$, a sequence $\tp_0 \to^\Box \tp_1 \to^\Box \ldots \to^\Box \tp_d$  of types for $\TO$ and a sequence $\min \A = s_0 < s_1 < \dots < s_{d-1} < s_d = \max \A$ of indexes satisfying the following conditions\textup{:}
\begin{align}
\label{eq:path-conditions:1:p}
\tag{\ref{eq:path-conditions:1}$'$}
& B\in\tp_i, \text{ for each } B(s_i)\in\A,  && \text{ for } 1 \leq i \leq d\textup{;}\\
\label{eq:path-conditions:2:c}
\tag{\ref{eq:path-conditions:2}$_1'$}
& (B,n) \not\leadsto (\neg B',n'), \text{ for } s_i < n, n' < s_{i+1} \text{ with } B(n), B'(n') \in \A, && \text{ for } 1 \leq i < d\textup{;}\\
\label{eq:path-conditions:2:e}
\tag{\ref{eq:path-conditions:2}$_2'$}
& (L,s_i) \not\leadsto (\neg B',n'), \text{ for } L \in \tp_i \text{ and } s_i < n' < s_{i+1} \text{ with } B'(n') \in \A, && \text{ for } 1 \leq i < d\textup{;}\\
\label{eq:path-conditions:2:d}
\tag{\ref{eq:path-conditions:2}$_3'$}
& (B,n) \not\leadsto (\neg L',s_{i+1}), \text{ for } s_i < n < s_{i+1} \text{ with  } B(n) \in \A \text{ and  } L' \in \tp_{i+1}, && \text{ for } 1 \leq i < d\textup{;}\\
\label{eq:path-conditions:2:f}
\tag{\ref{eq:path-conditions:2}$_4'$}
& (L,s_i) \not\leadsto (\neg L',s_{i+1}), \text{ for } L \in \tp_i \text{ and } L' \in \tp_{i+1}, && \text{ for } 1 \leq i < d\textup{;}\\
\label{eq:path-conditions:3:p}
\tag{\ref{eq:path-conditions:3}$'$}
& \ell = s_i, \text{ for some $i$ with } 0 \leq i \leq d \text{ such that } A \notin \tp_i.
\end{align}
\end{lemma}
\begin{proof}
$(\Leftarrow)$ Suppose $\ell \notin \ans(\q,\A)$.  Then there is a model $\Mmf$ of
$\TO$ and $\A$ such that $\ell \notin A^\Mmf$. We consider the sequence $\bar{\tp}_0\to \bar{\tp}_1\to \ldots \to \bar{\tp}_{m-1} \to \bar{\tp}_m$ of types for $\TO$ given by $\Mmf$, where $0 = \min \A$ and $m = \max \A$. As argued above, we can find subsequences of types and respective indexes between $\min\A$ and $\max\A$ whose length does not exceed~$2|\TO| + 1$. We add $\bar{\tp}_\ell$ to obtain the sequences satisfying conditions~\eqref{eq:path-conditions:1:p}--\eqref{eq:path-conditions:3:p}.

\smallskip

$(\Rightarrow)$ Suppose we have a sequence $\tp_0\to^\Box \tp_1 \to^\Box \cdots \to^\Box  \tp_d$ of types for $\TO$ and a sequence $\min \A = s_0 < s_1 < \dots < s_d = \max \A$ of indexes satisfying  conditions~\eqref{eq:path-conditions:1:p}--\eqref{eq:path-conditions:3:p}. We construct a model $\Mmf$ of $\TO$ and $\A$ with $\ell \notin A^\Mmf$. The model is defined as a sequence of types $\bar{\tp}_n$, for $n \in \Z$. We begin by setting  $\bar{\tp}_{s_i} = \tp_i$, for $0 \leq i \leq d$. Then, since $\bar{\tp}_{\min \A} = \tp_0$ is consistent with $\TO$, there is a model~$\Mmf_{\min}$ of~$\TO$ with type $\tp_0$ at $\min \A$, and so we take the types~$\bar{\tp}_n$ given by $\Mmf_{\min}$ for $n< \min \A$. Similarly, the types~$\bar{\tp}_n$, for $n > \max\A$, are provided by a model of $\TO$ with~$\bar{\tp}_{\max} = \tp_d$ at~$\max \A$. Now, let $1 \le i < d$. We show how to construct the~$\bar{\tp}_j$, for $s_i < j < s_{i+1}$, in a step-by-step manner.
\begin{description}
\item[\normalfont\textit{Step} 0:] for all $j$ with $s_i < j < s_{i+1}$, set $\bar{\tp}_j = \bigl\{\, B \mid B(j) \in \A\,\bigr\}$.

\item[\normalfont\textit{Step} 1:] for all $k$  with $s_i \leq k \leq s_{i+1}$ and all $j$ with $s_i < j < s_{i+1}$, if $L \in \bar{\tp}_k$ and $(L,k) \leadsto (L',j)$, then add $L'$ to $\bar{\tp}_j$.

\item[\normalfont\textit{Step} $m > 1$:]  pick $\bar{\tp}_k$, for $s_i < k < s_{i+1}$, and a literal $L$ with $L,\neg L \notin \bar{\tp}_k$, terminating the construction if
there are none. Add $L$ to $\bar{\tp}_k$, and, for all $j$ with  $s_i < j < s_{i+1}$, if
$(L,k) \leadsto (L',j)$, then add $L'$ to $\bar{\tp}_j$.
\end{description}
Note that $\bar{\tp}_k$ could also be extended with $\neg L$---either choice is consistent with the previously constructed types $\bar{\tp}_j$.

By induction on $m$,  we show that the $\bar{\tp}_k$
constructed in Step $m$ is \emph{conflict-free}
in the sense that there is no $k$, $s_i < k < s_{i+1}$, and no literal $L_0$ with
$L_0,\neg L_0 \in \bar{\tp}_k$, which is obvious for $m=0$. Suppose that  $\bar{\tp}_k$  is not conflict-free after Step 1. Then one of the following six cases has happened in Step 1 (we assume $s_i < n,n' < s_{i+1}$, if relevant, below):
\begin{itemize}
\item[--] if $L, L' \in \bar{\tp}_{s_i}$, $(L, s_i) \leadsto (L_0, k)$, $(L', s_i) \leadsto (\neg L_0, k)$, then $(L, s_i) \leadsto (\neg L', s_i)$, contrary to consistency of  $\tp_i$ with~$\TO$;
\item[--] if $L, L' \in \bar{\tp}_{s_{i+1}}$, $(L, s_{i+1}) \leadsto (L_0, k)$, $(L', s_{i+1}) \leadsto (\neg L_0, k)$, then $(L, s_{i+1}) \leadsto (\neg L', s_{i+1})$, which is also impossible;
\item[--] if $B(n), B'(n') \in \A$ with $(B, n) \leadsto (L_0, k)$, $(B', n') \leadsto (\neg L_0,k)$, then $(B, n) \leadsto (\neg B', n')$, contrary to~\eqref{eq:path-conditions:2:c};
\item[--] if $L \in \bar{\tp}_{s_i}$, $(L, s_i) \leadsto (L_0, k)$ and $B'(n') \in \A$, $(B', n') \leadsto (\neg L_0,k)$, then $(L, s_i) \leadsto (\neg B', n')$, contrary to~\eqref{eq:path-conditions:2:e};
\item[--] if $B(n) \in \A$, $(B, n) \leadsto (L_0,k)$ and $L' \in \bar{\tp}_{s_{i+1}}$, $(L', s_{i+1}) \leadsto (\neg L_0, k)$, then $(B, n) \leadsto (\neg L',s_{i+1})$, contrary to~\eqref{eq:path-conditions:2:d};
\item[--] if $L \in \bar{\tp}_{s_i}$, $(L, s_i) \leadsto (L_0, k)$ and $L' \in \bar{\tp}_{s_{i+1}}$,  $(L', s_{i+1}) \leadsto (\neg L_0, k)$, then $(L, s_i) \leadsto (\neg L', s_{i+1})$, contrary to~\eqref{eq:path-conditions:2:f}.
\end{itemize}
Thus, the $\bar{\tp}_k$ constructed in Step 1 are conflict-free.
Suppose now that all of the $\bar{\tp}_k$ are conflict-free after Step $m$, $m  \ge 1$, while some $\bar{\tp}_j$ after Step $m+1$ is not. It follows that some $\bar{\tp}_k$ is extended with $L$ in Step~$m+1$, $(L,k) \leadsto (L',j)$, but~$\bar{\tp}_j$ contained $\neg L'$ (at least) since Step $m$. Now, as $(\neg L',j) \leadsto (\neg L,k)$, the type $\bar{\tp}_k$ already contained $\neg L$ in Step $m$, and so~$L$~could not be added in Step $m+1$.

\smallskip

Let $\bar{\tp}_n$, $n \in \Z$, be the resulting sequence of types for $\TO$. Define an interpretation $\Mmf$ by taking $n \in B^\Mmf$ iff $B \in \bar{\tp}_n$, for every atomic concept $B$ in $\TO$. In view of~\eqref{eq:path-conditions:1:p} and Step 0, we have $\Mmf\models\A$ but $\ell\notin A^\Mmf$. We show~$\Mmf\models\TO$.
Since the~$\bar{\tp}_n$ are conflict-free, and thus
consistent with $\tilde{\TO}$, and since $L_1 \in \bar{\tp}_n$ implies
$L_2 \in \bar{\tp}_n$ if $\tilde{\TO}$ contains $L_1 \to
L_2$, and $L_1\in\bar{\tp}_n$ implies $L_2 \notin \bar{\tp}_n$ if $\tilde{\TO}$ contains $L_1 \land L_2 \to \bot$, it is sufficient to prove that
\begin{equation*}
A_{{\scriptscriptstyle \bigcirc_F} B} \in \bar{\tp}_n \ \  \text{ iff } \ \ B \in
\bar{\tp}_{n+1}\qquad\text{ and }\qquad \Rbox C \in \bar{\tp}_n \ \ \text{ iff } \ \ C
\in \bar{\tp}_k \text{ for all } k > n,
\end{equation*}
and the past
counterparts of these equivalences. We readily
obtain the first equivalence since $(A_{{\scriptscriptstyle \bigcirc_F}
B}, n) \leadsto (B, n+1)$ and $(B, n+1) \leadsto
(A_{{\scriptscriptstyle \bigcirc_F} B}, n)$, and similarly for $\Lnext$. It thus remains to show the second equivalence.
For all $n \geq \max \A$, the claim is immediate from the choice of the $\bar{\tp}_k$ for $k \ge s_d = \max \A$. We then proceed by induction on $i$ from $d - 1$ to~$0$ assuming that the claim holds for all $n \geq s_{i +1}$. We consider the following three options.
If $\Rbox C\in\bar{\tp}_{s_i}$, then, by~\eqref{eq:path:box-monotonicity}, we have $\Rbox C, C\in\bar{\tp}_{s_{i+1}}$, and the claim for all $n \geq s_i$ follows from the fact that  $\tilde{\TO}$ contains inclusions~\eqref{eq:box-semantics}. If $\Rbox C\notin\bar{\tp}_{s_i}$, then, by~\eqref{eq:path:box-monotonicity}, either $\Rbox C\notin\bar{\tp}_{s_{i+1}}$ or $C\notin \bar{\tp}_{s_{i+1}}$; in either case,  the claim for all $n \geq s_i$ is immediate from the induction hypothesis and the fact that  $\tilde{\TO}$ contains inclusions~\eqref{eq:box-semantics}. This finishes the inductive argument, and
the claim for $n < s_0$  then follows from the choice of the $\bar{\tp}_k$ for $k \le s_0 = \min \A$.
A symmetric argument shows that $\Lbox C \in \bar{\tp}_n$ iff~$L \in \bar{\tp}_k$ for all~$k < n$.
This completes the proof of Lemma~\ref{lem:kromcond}.
\qed
\end{proof}

Returning to the proof of Theorem~\ref{thm:krom-nextbox}, we are now in a position to define an $\FOE$-rewriting $\rq(t)$ of the OMAQ $\q = (\TO,A)$ by encoding the conditions of Lemma~\ref{lem:kromcond} as follows:
\begin{equation*}
\rq(t) ~=~ \neg \Bigl[\bigvee_{d \leq 2|\TO| + 2} \ \ \bigvee_{\tp_0 \to^\Box \ldots \to^\Box \tp_d}  \exists t_0,\dots,t_d \,
\Bigl(\pathf_{\tp_0\to^\Box\ldots\to^\Box\tp_d}(t_0,\dots,t_d) \land  \bigvee_{\begin{subarray}{c}0 \leq i \leq d\\A \notin \tp_i\end{subarray}} (t = t_i)
\Bigr)\Bigr],
\end{equation*}
where
\begin{equation*}
\pathf_{\tp_0\to^\Box\ldots\to^\Box\tp_d}(t_0, \dots,t_d)  \   = \   (t_0 = \min) \ \land  \bigwedge_{0 \leq i < d} (t_i < t_{i+1})  \ \land \ (t_d = \max)\ \ \land \
\bigwedge_{0 \leq i \leq d}\hspace*{-0.25em}\typef_{\tp_i}(t_i) \ \ \land \
\bigwedge_{0 \leq i < d}\hspace*{-0.25em}\intervalf_{\tp_i,\tp_{i+1}}(t_i, t_{i+1}),
\end{equation*}
the formulas $\typef_{\tp}$ are given at the beginning of Section~\ref{sec:LTL}, the  formulas $\entailsf_{L,L'}$ are from the proof of Theorem~\ref{thm:krom-next} (but for~$\tilde{\TO}$ rather than $\TO$ alone) and
\begin{align*}
\intervalf_{\tp_i \tp_{i+1}}(t_i, t_{i+1}) \ \ =   \bigwedge_{B,B' \text{ occur  in } \TO} & \forall t,t' \, \bigl((t_i < t< t_{i+1}) \land B(t) \land (t_i < t' < t_{i+1} )  \land B'(t') \to \neg\entailsf_{B, \neg B'}(t,t') \bigr) \ \ \land \\
 \bigwedge_{L \in\tp_i \text{ and } B' \text{ occurs in } \TO} & \forall  t' \,\bigl((t_i < t' < t_{i+1})  \land B'(t') \to \neg\entailsf_{L, \neg B'}(t_i, t')\bigr) \ \ \land\\
  \bigwedge_{B \text{ occurs in } \TO \text{ and } L'\in\tp_{i+1}} & \forall  t\,\bigl( (t_i < t < t_{i+1}) \land B(t)  \to\neg\entailsf_{B,\neg L'}(t, t_{i+1})\bigr) \ \ \land\\
  \bigwedge_{L\in\tp_i \text{ and }L' \in \tp_{i+1}} & \neg\entailsf_{L, \neg L'}(t_i, t_{i+1}).
\end{align*}
Note that, if $\TO$ is inconsistent, then the disjunctions over sequences of types $\tp_0 \to^\Box \ldots \to^\Box \tp_d$ for $\TO$ are empty,  and the rewriting $\rq(t)$
is $\neg \bot$, that is, $\top$. \qed
\end{proof}


\section{Canonical Models for \OMPIQ{}s with Horn Ontologies}
\label{sec:ltl-tomiq}

We next use Theorems~\ref{thm:box} and~\ref{thm:krom-next} to construct $\FO(<)$-rewritings for $\LTL_{\smash{\horn}}\Xbox$ \OMPIQ{}s and $\FOE$-rewritings for $\LTL_\core\Xallop$ \OMPIQ{}s, thereby completing Table~\ref{LTL-table} for OMPIQs. First, for any $\LTL_{\smash{\horn}}\Xallop$ ontology $\TO$ and any data instance~$\A$ consistent with $\TO$, we define  an interpretation~$\C_{\TO,\A}$, called the \emph{canonical model of $\TO$ and $\A$}, that gives the certain answers to all \OMPIQ{}s of the form $(\TO,\varkappa)$.

The canonical model is defined by transfinite recursion~\cite{DBLP:books/daglib/0090259,Hajnal99} (examples illustrating why infinite ordinals are required in the construction of the model are given after the definition). Let $\Lambda$ be a countable set of atoms of the form~$\bot$ and $C(n)$, where $C$ is a basic temporal concept and $n \in \Z$. Denote by $\cl_\TO(\Lambda)$ the result of applying non-recursively the following rules to~$\Lambda$:
\begin{description}
\item[\rm (mp)] if $\TO$ contains $C_1 \land \dots \land C_m \to C$ and $C_i(n) \in \Lambda$ for all $i$, $1 \le i \le m$, then we add $C(n)$ to $\Lambda$;

\item[\rm (cls)] if $\TO$ contains $C_1 \land \dots \land C_m \to \bot$ and $C_i(n) \in \Lambda$ for  all $i$, $1 \le i \le m$, then we add $\bot$;

\item[\rm ({$\Rbox^\to$})] if $\Rbox C(n) \in \Lambda$, then we  add all $C(k)$ with $k > n$;

\item[\rm ({$\Rbox^\leftarrow$})] if $C(k) \in \Lambda$ for all $k > n$, then we add $\Rbox C(n)$;

\item[\rm ({$\Rnext^\to$})] if $\Rnext C(n) \in \Lambda$,  then we  add $C(n+1)$;

\item[\rm ({$\Rnext^\leftarrow$})] if $C(n + 1) \in \Lambda$,  then we add $\Rnext C(n)$;
\end{description}
and symmetric rules $(\Lbox^\to)$, ($\Lbox^\leftarrow$), ($\Lnext^\to$) and ($\Lnext^\leftarrow$) for the corresponding past-time operators. Note that the concepts  introduced in some of the rules above do not necessarily occur in $\TO$: for example, repeated applications of $\cl_\TO$ to $\Lambda$ containing~$C(n)$  will extend it with (infinitely many) atoms of the form $\Rnext^k C(n - k)$ and $\Lnext^k C(n + k)$, for $k > 0$. Note also that the rules for the $\Box$-operators are infinitary. 

Now, starting with $\cl_\TO^0(\Lambda) = \Lambda$, we then set, for any successor ordinal $\xi +1$ and any limit ordinal $\zeta$,
\begin{equation}\label{closures}
\cl_\TO^{\smash{\xi +1}}(\Lambda) = \cl_\TO(\cl_\TO^{\smash{\xi}}(\Lambda)) \qquad \text{ and }\qquad  \cl^{\smash{\zeta}}_\TO (\Lambda) = \bigcup\nolimits_{\xi<\zeta} \cl_\TO^{\smash{\xi}}(\Lambda).
\end{equation}
Let  $\C_{\TO,\A} = \cl^{\smash{\omega_1}}_\TO(\Lambda)$, where $\omega_1$ is the first uncountable ordinal. Note that, as
$\cl_\TO^{\smash{\omega_1}}(\Lambda)$ is countable (because its cardinality does not exceed the cardinality of $\mathbb Z \times \mathbb N$), there is an ordinal $\alpha < \omega_1$ such that $\cl^{\smash{\alpha}}_\TO(\Lambda) = \cl^{\smash{\beta}}_\TO (\Lambda)$, for all $\beta \ge \alpha$. 
For example, to derive $\Rbox A(n)$ from $A(n)$ and $A \to \Rnext A$, one needs $\omega + 1$ steps of the recursion; so deriving $B(n)$ from $A(n)$ and $\TO = \{\, A \to \Rnext A, \ \Rbox A \to B\,\}$ requires $\omega + 2$ steps.

In the sequel, we regard $\can$ as both a set of atoms of the form $\bot$ and $C(n)$, and an interpretation where, for any atomic concept~$A$ (but not $\bot$), we have $A^\can  = \{\, n\in \Z \mid A(n) \in \can\,\}$.

\begin{theorem}\label{canonical}
Let $\TO$ be an $\LTL_\horn\Xallop$ ontology and $\A$ a data instance. Then the following hold\textup{:}
\begin{description}
\item[\rm (\emph{i})] for any basic temporal concept $C$ and any $n \in \Z$, we have $C^\can = \{\, n\in\Z\mid C(n) \in \can\,\}$\textup{;}

\item[\rm (\emph{ii})] for any model $\Mmf$ of $\TO$ and $\A$, any basic temporal concept $C$ and any $n \in \Z$, if $C(n) \in \can$ then $n\in C^\Mmf$\textup{;}

\item[\rm (\emph{iii})] if $\bot \in \C_{\TO,\A}$, then $\TO$ and $\A$ are inconsistent\textup{;} otherwise, $\can$ is a model of $\TO$ and $\A$\textup{;}

\item[\rm (\emph{iv})] if $\TO$ and $\A$ are consistent, then, for any \OMPIQ{} $\q = (\TO,\varkappa)$, we have
$\ans^\Z(\q,\A) =  \varkappa^\can$.
\end{description}
\end{theorem}
\begin{proof}
Claim (\emph{i}) is proved by induction on the construction of $C$. The basis is immediate from the definition of~$\can$. For $C =  \Rbox D$, suppose first that $n\in (\Rbox D)^\can$. Then $k\in D^\can$ for all $k > n$, whence, by the induction hypothesis, we have~$D(k) \in \can$, and so, by ({$\Rbox^\leftarrow$}), we obtain $\Rbox D(n) \in\can$. Conversely, if $\Rbox D(n) \in\can$ then, by ({$\Rbox^\to$}), we have $D(k) \in \can$ for all $k > n$. By the induction hypothesis, $k \in D^\can$ for all $k > n$, and so $n\in (\Rbox D)^\can$. The other temporal operators, $\Lbox$, $\Rnext$ and $\Lnext$, are treated similarly.

Claim (\emph{ii}) is proved by transfinite induction on the construction of
$\C_{\TO,\A}$. If $\xi=0$, then the claim holds because $\cl^0_\TO(\A)=\A$, and so $C(n)\in \cl^0_\TO(\A)$ implies $n\in C^\Mmf$ for every model $\Mmf$
of $\A$. Now, suppose $C(n) \in \cl_\TO^{\xi}(\A)$ implies $n \in C^{\Mmf}$, for all $n \in \Z$ and all ordinals $\xi < \zeta$. If $\zeta$ is a limit ordinal, then $\cl^{\smash{\zeta}}_\TO (\A) = \bigcup\nolimits_{\xi<\zeta} \cl_\TO^{\smash{\xi}}(\A)$, and so we are done. If $\zeta = \xi + 1$, then $\cl_\TO^{\smash{\xi +1}}(\A) = \cl_\TO(\cl_\TO^{\smash{\xi}}(\A))$, and the claim follows from the induction hypothesis  and the fact that $\Mmf$ is a model of $\TO$, for rules (mp) and (cls), and definitions \eqref{eq:semantics:box} and \eqref{eq:semantics:next} of $C^{\Mmf}$, for the remaining rules. 

To show (\emph{iii}), assume first that $\bot\in \can$. Then there exist an axiom  $C_1 \land \dots \land C_m \to \bot$ in $\TO$ and $n\in \Z$ such that $C_i(n) \in \can$ for all $i$ $(1 \le i \le m)$. If there were a model $\Mmf$ of $\TO$ and $\A$, then, by (\emph{ii}), we would have $n \in C_{\smash{i}}^\Mmf$ for all $i$, which is impossible. It follows that there is no model of $\TO$ and $\A$, that is, they are inconsistent.
If $\bot\notin\can$, then $\can$
is a model of $\TO$ by (\emph{i}) as it is closed under rules~(mp) and~(cls); $\can$ is a model of $\A$ by definition.

To show (\emph{iv}), observe first that, for any interpretations $\Mmf$ and $\Mmf'$, if $A^\Mmf\subseteq A^{\Mmf'}$
for all atomic concepts $A$, then $\varkappa^\Mmf\subseteq\varkappa^{\Mmf'}$ for all positive
temporal concepts $\varkappa$.
Now, that $n \in \ans^\Z(\q,\A)$ implies $n\in\varkappa^\can$
follows by (\emph{iii}) from the fact that $\can$ is a model of $\TO$ and $\A$, while the converse direction follows
from (\emph{ii}) and the observation above.
\qed
\end{proof}

If $\A$ is consistent with $\TO$,  then, by Theorem~\ref{canonical}~(\emph{ii}) and~(\emph{iii}), for any $C$ in $\TO$ and $\A$ and any $n \in \Z$,
\begin{equation}\label{intersection}
n\in C^\can \quad\text{ iff }\quad n\in C^\Mmf \ \ \text{ for all models }\Mmf \text{ of } \TO \text{ and }\A.
\end{equation}
Denote by $\type_{\TO,\A}(n)$ the set of all $C\in\subo$ with $n\in C^\can$. It is known that every satisfiable $\LTL$-formula is satisfied in an \emph{ultimately periodic model}~\cite{DBLP:journals/jacm/SistlaC85}. Since $\C_{\TO,\A}$ is a `minimal' model in the sense of~\eqref{intersection},  it is also ultimately periodic, which is formalised in the next two lemmas. This periodic structure will be used in defining our FO-rewritings below.
\begin{lemma}\label{period:A}
{\rm (\emph{i})} For any $\LTL_{\smash{\horn}}\Xallop$ ontology $\TO$ and any $\A$ consistent with $\TO$, there are positive integers $s_{\TO,\A} \le 2^{|\TO|}$ and $p_{\TO,\A} \le 2^{2|\TO|}$ such that
\begin{equation}
\label{eq:period:A}
\type_{\TO,\A}(n) = \type_{\TO,\A}(n - p_{\TO,\A}), \text{ for } n \le \min \A -s_{\TO,\A},  \quad\text{ and }\quad
\type_{\TO,\A}(n) = \type_{\TO,\A}(n + p_{\TO,\A}), \text{ for }n \ge \max \A + s_{\TO,\A}.
\end{equation}
{\rm (\emph{ii})} For any $\LTL_\horn\Xbox$ ontology $\TO$ and any $\A$ consistent with $\TO$, there is a positive integer $s_{\TO,\A} \le |\TO|$ such that~\eqref{eq:period:A} holds with $p_{\TO,\A} =1$.
\end{lemma}
\begin{proof}
(\emph{i}) Take the minimal number $k > \max \A$ such that $\type_{\TO,\A}(k) = \type_{\TO,\A}(k')$, for some $k'$ with $\max \A < k' < k$, and set $s^\FF = k' - \max \A$ and $p^\FF = k - k'$. As the number of distinct $\type_{\TO,\A}(n)$, for $n\notin\tem(\A)$, does not exceed $2^{|\TO|}$, we have $\smash{s^\FF,p^\FF \le 2^{|\TO|}}$. By using a symmetric construction, we obtain $\smash{s^\PP}$ and $\smash{p^\PP}$ and set $\smash{s_{\TO,\A} = \max(s^\FF, s^\PP)}$ and $\smash{p_{\TO,\A} = p^\FF \times p^\PP}$.

To show~\eqref{eq:period:A}, consider the interpretation $\Mmf_1$ where, for any atomic concept $A$ occurring in $\TO$ or $\A$ and any $n \in \Z$, we have $n \in A^\Mmf$ iff $A \in \type_{\TO,\A}(n)$, for $n < k'$, and $n \in A^{\Mmf_1}$ iff $A \in \type_{\TO,\A}(n + p_{\TO,\A})$, for $n \ge k'$. (Intuitively, $\Mmf_1$ is obtained by cutting out the fragment $\type_{\TO,\A}(k'), \dots, \type_{\TO,\A}(k -1)$ of length~$p_{\TO,\A}$ from $\C_{\TO,\A}$.) Denote by $\type_1(n)$ the set of basic temporal concepts $C$ from $\TO$ and $\A$ such that $n\in C^{\Mmf_1}$. By induction on the construction of $C$, it is not hard to see that $\type_1(n) = \type_{\TO,\A}(n)$, for $n < k'$, and $\type_1(n) = \type_{\TO,\A}(n + p_{\TO,\A})$, for $n \ge k'$. It follows that $\Mmf_1$ is a model of $\TO$ and~$\A$. In view of~\eqref{intersection}, we obtain $\type_{\TO,\A}(n) \subseteq \type_{\TO,\A}(n+p_{\TO,\A})$, for every $n \geq \max\A + s_{\TO,\A}$.

For the converse inclusion, consider $\Mmf_2$ obtained by replacing the fragment $\type_{\TO,\A}(k'), \dots, \type_{\TO,\A}(k -1)$ of length~$p_{\TO,\A}$ in~$\C_{\TO,\A}$  with the doubled sequence $\type_{\TO,\A}(k'), \dots, \type_{\TO,\A}(k-1)$, $\type_{\TO,\A}(k'), \dots, \type_{\TO,\A}(k-1)$. Again, $\Mmf_2$ is a model of $\TO$ and~$\A$, but now, for any $n \geq \max \A + s_{\TO,\A}$, the set of all  $C$ in $\TO$ and $\A$ with $n+p_{\TO,\A} \in C^{\Mmf_2}$ coincides with $\type_{\TO,\A}(n)$. By~\eqref{intersection}, we obtain $\type_{\TO,\A}(n+p_{\TO,\A}) \subseteq \type_{\TO,\A}(n)$.

This proves the second claim in~\eqref{eq:period:A}; a similar argument establishes the first one.

\smallskip

(\emph{ii}) Observe that $\Rbox C \in \type_{\TO,\A}(k')$ implies $\Rbox C \in \type_{\TO,\A}(k)$, and $\Lbox C \in \type_{\TO,\A}(k)$ implies $\Lbox C \in \type_{\TO,\A}(k')$, for any $k' <k$, and that if $\type_{\TO,\A}(k')$ and $\type_{\TO,\A}(k)$, for $\max \A < k' <k$ or $\min \A > k' > k$, have the same boxed concepts, then $\type_{\TO,\A}(k') = \type_{\TO,\A}(k)$. To show the latter, consider the interpretation $\Mmf$ obtained from $\C_{\TO,\A}$
by setting $n\in A^\Mmf$ if $A\in \bigcap_{k' \le n \le k}\type_{\TO,\A}(n)$, for $k' \le n \le k$. It is readily seen that $\Mmf$ satisfies the same boxed concepts as $\C_{\TO,\A}$ at every $n\in \Z$, and so it is a model of $\TO$ and $\A$. Thus, by~\eqref{intersection}, all the $\type_{\TO,\A}(n)$ must coincide, for $k' \le n \le k$. The case
$k'>k$ is symmetric. Now (\emph{ii}) follows from the observation that there are at most $|\TO|$-many $\type_{\TO,\A}(n)$ having distinct boxed concepts and the proof of (\emph{i}).
\qed
\end{proof}

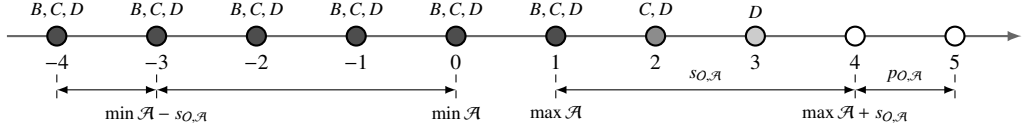
\begin{figure}
\centering%
\begin{tikzpicture}[>=latex, yscale=0.75, xscale=0.875, semithick,
point/.style={draw,thick,circle,inner sep=0pt,minimum size=2.5mm,fill=white}]\footnotesize
\draw[thick,->, black!60] (-6.75,0) -- ++(15.25,0);
%
\node (p1p) at (-6,0) [point, fill=black!70,label=below:{$-4$}, label=above:{\scriptsize\shortstack{$B,C,D$}}]{}; 
\node (p1) at (-4.5,0) [point, fill=black!70,label=below:{$-3$}, label=above:{\scriptsize\shortstack{$B,C,D$}}]{}; 
\node at (-3,0) [point, fill=black!70,label=below:{$-2$}, label=above:{\scriptsize\shortstack{$B,C,D$}}]{}; 
\node at (-1.5,0) [point, fill=black!70,label=below:{$-1$}, label=above:{\scriptsize\shortstack{$B,C,D$}}]{}; 
\node (n0) at (0,0) [point, fill=black!70,label=below:{$0$}, label=above:{\scriptsize\shortstack{$B,C,D$}}]{}; 
\node (n1) at (1.5,0) [point, fill=black!70,label=below:{$1$}, label=above:{\scriptsize\shortstack{$B,C,D$}}]{}; 
\node at (3,0) [point, fill=black!45,label=below:{$2$}, label=above:{\scriptsize\shortstack{$C,D$}}]{}; 
\node at (4.5,0) [point, fill=black!20,label=below:{$3$}, label=above:{\scriptsize\shortstack{$D$}}]{}; 
\node (n2) at (6,0) [point, label=below:{$4$}, label=above:{\scriptsize}]{};
\node (n2p) at (7.5,0) [point, label=below:{$5$}, label=above:{\scriptsize}]{};
%
%
\begin{scope}[dashed,thin]
\draw ($(n0)+(0,-1.1)$) node[below] {\scriptsize$\min\A$} -- ++(0,0.5);
\draw ($(n1)+(0,-1.1)$) node[below] {\scriptsize$\max\A$} -- ++(0,0.5);
\draw ($(n2)+(0,-1.1)$) node[below] {\scriptsize$\max\A+s_{\TO,\A}$} -- ++(0,0.5);
\draw ($(p1)+(0,-1.1)$) node[below] {\scriptsize$\min\A-s_{\TO,\A}$} -- ++(0,0.5);
\draw ($(p1p)+(0,-1.1)$)  -- ++(0,0.5);
\draw ($(n2p)+(0,-1.1)$)  -- ++(0,0.5);
\end{scope}
\draw[<->,thin] ($(n1)+(0,-0.95)$) to node [above] {\scriptsize$s_{\TO,\A}$} ($(n2)+(0,-0.95)$);
\draw[<->,thin] ($(n0)+(0,-0.95)$) to ($(p1)+(0,-0.95)$); 
\draw[<->,thin] ($(p1)+(0,-0.95)$) to ($(p1p)+(0,-0.95)$);
\draw[<->,thin] ($(n2)+(0,-0.95)$) to node [above] {\scriptsize$p_{\TO,\A}$} ($(n2p)+(0,-0.95)$);
\end{tikzpicture}
\caption{The canonical model in Example~\ref{ex:period} (\emph{i}): $s_{\TO,\A}=3$ and $p_{\TO,\A}=1$.}\label{fig:period:1}
\end{figure}

\begin{figure}
\centering%
\begin{tikzpicture}[>=latex, yscale=0.75, xscale=1.3, semithick,
point/.style={draw,thick,circle,inner sep=0pt,minimum size=2.5mm,fill=white}]\footnotesize
\begin{scope}[line width=1mm,gray]
\draw (2,1.5) -- ++(0,-5.4);
\draw (8,1.5) -- ++(0,-5.4);
\draw[ultra thick,<->] (2,1.2) to node[above,black] {$p_{\TO}$} ++(6,0);
\end{scope}
\draw[thick, ->, black!60] (-1.5,0) -- ++(10.25,0);
%
\node (p1) at (-1,0) [point, label=below:{$-1$}, label=above:{\scriptsize}]{}; 
\node (n0) at (0,0) [point, fill=gray, label=below:{$0$}, label=above:{\scriptsize\shortstack{$A$}}]{}; 
\node (n1) at (1,0) [point, label=below:{$1$}, label=above:{\scriptsize\shortstack{}}]{}; 
\node (n2) at (2,0) [point, fill=gray, label=below:{\colorbox{white}{$2$}}, label=above:{\scriptsize\colorbox{white}{\shortstack{$A$}}}]{}; 
\node (p2) at (3,0) [point, label=below:{$3$}, label=above:{\scriptsize\shortstack{}}]{}; 
\node (p3) at (4,0) [point, fill=gray, label=below:{$4$}, label=above:{\scriptsize\shortstack{$A$}}]{}; 
\node (p4) at (5,0) [point, label=below:{$5$}, label=above:{\scriptsize\shortstack{}}]{}; 
\node (p5) at (6,0) [point, fill=gray, label=below:{$6$}, label=above:{\scriptsize\shortstack{$A$}}]{}; 
\node (p6) at (7,0) [point, label=below:{$7$}, label=above:{\scriptsize\shortstack{}}]{}; 
\node (p7) at (8,0) [point, fill=gray, label=below:{\colorbox{white}{$8$}}, label=above:{\scriptsize\colorbox{white}{\shortstack{$A$}}}]{}; 
\begin{scope}[dashed,thin]
\draw ($(n0)+(0,-1.1)$) node[below] {\scriptsize$\min\A_1$} -- ++(0,0.4);
\draw ($(n1)+(0,-1.1)$) node[below,fill=white] {\scriptsize$\max\A_1$} -- ++(0,0.4);
\draw ($(n2)+(0,-1.2)$) -- ++(0,0.6);
\draw ($(p3)+(0,-1.2)$) -- ++(0,0.6);
\draw ($(p5)+(0,-1.2)$) -- ++(0,0.6);
\draw ($(p7)+(0,-1.2)$) -- ++(0,0.6);
\end{scope}
\draw[<->,thin] ($(n2)+(0,-0.9)$) to node [below] {\scriptsize$p_{\TO,\A_1}$} ($(p3)+(0,-0.9)$);
\draw[<->,thin] ($(p3)+(0,-1.1)$) to node [above] {\scriptsize$p_{\TO,\A_1}$} ($(p5)+(0,-1.1)$);
\draw[<->,thin] ($(p5)+(0,-0.9)$) to node [below] {\scriptsize$p_{\TO,\A_1}$} ($(p7)+(0,-0.9)$);
\begin{scope}[yshift=-24mm]
\draw[->, thick, black!60] (-1.5,0) -- ++(10.25,0);
%
\node (p1) at (-1,0) [point, label=below:{$-1$}, label=above:{\scriptsize\shortstack{}}]{};
\node (n0) at (0,0) [point, label=below:{$0$}, label=above:{\scriptsize }]{};
\node (n1) at (1,0) [point, fill=gray, label=below:{$1$}, label=above:{\scriptsize\shortstack{$B$}}]{}; 
\node (n2) at (2,0) [point, label=below:{\colorbox{white}{$2$}}, label=above:{\scriptsize\colorbox{white}{\shortstack{}}}]{}; 
\node (p2) at (3,0) [point, label=below:{$3$}, label=above:{\scriptsize\shortstack{}}]{}; 
\node (p3) at (4,0) [point, fill=gray, label=below:{$4$}, label=above:{\scriptsize\shortstack{$B$}}]{}; 
\node (p4) at (5,0) [point, label=below:{$5$}, label=above:{\scriptsize\shortstack{}}]{}; 
\node (p5) at (6,0) [point, label=below:{$6$}, label=above:{\scriptsize\shortstack{}}]{}; 
\node (p6) at (7,0) [point, fill=gray, label=below:{$7$}, label=above:{\scriptsize\shortstack{$B$}}]{}; 
\node (p7) at (8,0) [point, label=below:{\colorbox{white}{$8$}}, label=above:{\scriptsize\colorbox{white}{\shortstack{}}}]{}; 
\begin{scope}[dashed,thin]
\draw ($(n0)+(0,-1.1)$) node[below] {\scriptsize$\min\A_2$} -- ++(0,0.4);
\draw ($(n1)+(0,-1.1)$) node[below,fill=white] {\scriptsize$\max\A_2$} -- ++(0,0.4);
\draw ($(n2)+(0,-1.2)$) -- ++(0,0.6);
\draw ($(p4)+(0,-1.2)$) -- ++(0,0.6);
\draw ($(p7)+(0,-1.2)$) -- ++(0,0.6);
\end{scope}
\draw[<->,thin] ($(n2)+(0,-0.9)$) to node [below] {\scriptsize$p_{\TO,\A_2}$} ($(p4)+(0,-0.9)$);
\draw[<->,thin] ($(p4)+(0,-1.1)$) to node [above] {\scriptsize$p_{\TO,\A_2}$} ($(p7)+(0,-1.1)$);
\end{scope}
\end{tikzpicture}%
\caption{The canonical model in Example~\ref{ex:period} (\emph{ii}): $s_{\TO}=3$ and $p_{\TO}=6$.}\label{fig:period:2}
\end{figure}
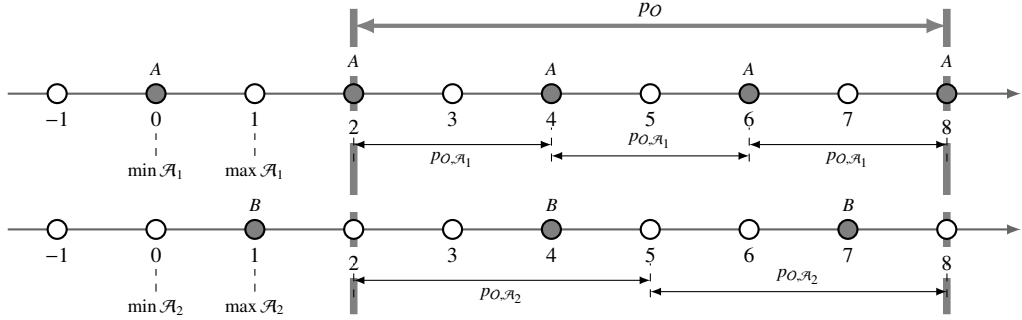

\begin{example}\label{ex:period}\em
We illustrate Lemma~\ref{period:A} by two examples. (\emph{i}) The periodic structure of the canonical model of the $\LTL_\core\Xbox$ ontology $\TO=\bigl\{\,B \to \Lbox B, \ \Lbox B \to C, \ \Lbox C\to D\,\bigr\}$ and the data instance $\A=\{\,B(1)\,\}$ is shown in Fig.~\ref{fig:period:1}.

(\emph{ii}) Consider now the $\LTL_\core\Xnext$ ontology $\TO = \bigl\{\,A \to \Rnext^2 A, \  \ B \to \Rnext^3 B\,\bigr\}$ and the data instances $\A_1 = \{\,A(0)\,\}$ and $\A_2 = \{\,B(1)\,\}$. It is readily seen that $s_{\TO,\A_1} = 1$ and $ p_{\TO,\A_1} = 2$, while $s_{\TO,\A_2} = 1$ and $p_{\TO,\A_2} = 3$; see  Fig.~\ref{fig:period:2}. On the other hand, we can also take $s_{\TO,\A_1} = s_{\TO,\A_2} =1$ and $p_{\TO,\A_1} = p_{\TO,\A_2} = 6$.
\end{example}

In fact, as we shall see in Lemma~\ref{period} below, these numbers can always  be chosen independently of the data instance: $s_{\TO}$ is the maximum of the lengths of prefixes, and $p_{\TO}$ is the least common multiple of all possible periods.
\begin{lemma}\label{period}
{\rm (\emph{i})}  For any $\LTL_{\smash{\horn}}\Xallop$ ontology $\TO$, there are positive integers $s_{\TO} \le 2^{|\TO|}$ and $p_{\TO} \le 2^{2|\TO|\cdot 2^{|\TO|}}$ such that, for any data instance $\A$ consistent with $\TO$,
\begin{equation}
\label{eq:period}
\type_{\TO,\A}(n) = \type_{\TO,\A}(n - p_{\TO}),  \text{ for } n \le \min \A -s_{\TO}, \quad\text{ and }\quad
%
\type_{\TO,\A}(n) = \type_{\TO,\A}(n + p_{\TO}),  \text{ for }n \ge \max \A + s_{\TO}.
\end{equation}
{\rm (\emph{ii})} For any $\LTL_{\smash{\horn}}\Xbox$ ontology $\TO$, there is a positive integer $s_{\TO} \le |\TO|$ such that~\eqref{eq:period} holds with $p_{\TO} =1$, for any $\A$ consistent with $\TO$.
\end{lemma}
\begin{proof}
Observe that the numbers $s_{\TO,\A}$ and $p_{\TO,\A}$ provided by Lemma~\ref{period:A}  depend only on the `types' $\type_{\TO,\A}(\min \A)$ and $\type_{\TO,\A}(\max \A)$. It follows that we can take $s_{\TO}$ to be the maximum over the $s_{\TO,\A}$ for all such types, and  $p_{\TO}$ to be the product of all, at most $2^{|\TO|}$-many,  numbers $p_{\TO,\A}$.
\qed
\end{proof}

Note that, for the ontology $\TO$ in Example~\ref{ex:period}~(\emph{ii}), a single data instance $\A_3 = \A_1 \cup \A_2$ gives rise to the period of $6$ (the canonical model for $\TO$ and $\A_3$ has no shorter period). In general, however, there may be no single data instance with the minimal period of $p_{\TO}$. Indeed, consider an extension $\TO'$ of $\TO$ with $A\land B \to \Rnext A$: it can be seen that $p_{\TO'} = 6$, but the canonical models of all individual data instances have also shorter periods.

Given a positive temporal concept $\varkappa$, we denote by $\len_\varkappa$ the number of temporal operators in $\varkappa$ (thus, $\len_A = 0$ for an atomic concept $A$).

\begin{lemma}\label{period1}
For any $\LTL_\horn\Xallop$ ontology $\TO$, any data instance $\A$ consistent with $\TO$ and any positive temporal concept~$\varkappa$, %
\begin{align*}
&n\in \varkappa^\can\quad \text{ iff  }\quad n + p_\TO  \in\varkappa^\can,\qquad\text{ for }n \ge \max \A + s_\TO + \len_\varkappa p_\TO,\\
&n \in \varkappa^\can\quad \text{ iff } \quad n - p_\TO \in \varkappa^\can,\qquad\text{ for } n \le \min \A-(s_\TO + \len_\varkappa p_\TO).
\end{align*}
\end{lemma}
\begin{proof}
The proof is by induction on the construction of $\varkappa$. The basis of induction follows from Lemma~\ref{period}.

Suppose $\varkappa = \varkappa_1 \Si \varkappa_2$ and $ n \in \varkappa^{\smash{\can}}$  for some $n \ge \max \A + s_\TO + \len_\varkappa p_\TO$. Then there is $k < n$ such that $k \in \varkappa_2^{\smash{\can}}$ and $m \in \varkappa_1^{\smash{\can}}$ for all $m$ with $k < m < n$. Two cases are possible; see Fig.~\ref{fig:periodic-two-cases}.

\begin{figure}
\centering%
\begin{tikzpicture}[>=latex,semithick,xscale=2.3]\footnotesize
\fill[gray!40] (5.5,0.8) rectangle +(1.5,0.3);
\fill[gray!40] (3.5,0) rectangle +(1.5,0.3);
\draw[thick,->, black!60] (1,0) -- (7.5,0);
%
\draw[dashed] (2,-0.3) -- ++(0,1.6);
\node at (2,-0.6) {$\max \A + s_\TO+ (\len_\varkappa - 1)p_\TO$};
\draw[dashed] (4,-0.3) -- ++(0,1.6);
\draw[->,thin] (2,-0.2) -- (4,-0.2);
\node at (3,-0.5) {\scriptsize $p_\TO$};
\draw[ultra thick] (3.5,0) -- ++(0,0.9);
\node at (3.5,1.1) {$k$};
\node at (3.35,0.3) {$\varkappa_2$};
\draw[->,thin] (3.5,0.8) -- (5.5,0.8);
\node at (4.5,1) {\scriptsize $p_\TO$};
\draw[ultra thick] (5,-0.2) -- ++(0,0.5);
\node at (5,0.5) {$n$};
\node at (5,-0.4) {$\varkappa_1 \Si \varkappa_2$};
\node[circle,inner sep=0pt,fill=white] at (4.25,0.15) {$\varkappa_1$};
\draw[ultra thick] (7,-0.2) -- ++(0,1.3);
\node at (7,1.3) {$n+p_\TO$};
\draw[ultra thick] (5.5,0) -- ++(0,1.2);
\node at (5.35,0.5) {$\varkappa_2$};
\node at (5.5,1.4) {$k+p_\TO$};
\node[circle,inner sep=0pt,fill=white] at (6.25,0.95) {$\varkappa_1$};
\draw[->,thin] (5,0.2) -- (7,0.2);
\node at (6,0.4) {\scriptsize $p_\TO$};
\node at (7,-0.4) {$\varkappa_1 \Si \varkappa_2$};
%
%
\begin{scope}[yshift=-28mm]
\fill[gray!40] (4,0.8) rectangle +(2,0.3);
\fill[gray!40] (6,1.1) rectangle +(1.5,0.3);
\fill[gray!40] (1.5,0) rectangle +(3.5,0.3);
\draw[thick,->,black!60] (1,0) -- (7.5,0);
%
\draw[dashed] (2,-0.3) -- ++(0,1.8);
\node at (2,-0.6) {$\max \A + s_\TO+ (\len_\varkappa - 1)p_\TO$};
\draw[dashed] (4,-0.3) -- ++(0,1.8);
\draw[dashed] (6,-0.3) -- ++(0,1.8);
\draw[ultra thick] (1.5,0) -- ++(0,0.5);
\node at (1.5,0.7) {$k$};
\node at (1.35,0.3) {$\varkappa_2$};
\draw[->,thin] (2,0.8) -- (4,0.8);
\node at (3,1) {\scriptsize $p_\TO$};
\draw[ultra thick] (5,-0.2) -- ++(0,0.5);
\node at (5,0.5) {$n$};
\node at (5,-0.4) {$\varkappa_1 \Si \varkappa_2$};
\node[circle,inner sep=0pt,fill=white] at (3.25,0.15) {$\varkappa_1$};
\draw[ultra thick] (7,-0.2) -- ++(0,1.6);
\node at (7,1.6) {$n+p_\TO$};
\node[circle,inner sep=0pt,fill=white] at (5,0.95) {$\varkappa_1$};
\draw[->,thin] (4,1.1) -- (6,1.1);
\node at (5,1.3) {\scriptsize $p_\TO$};
\node[circle,inner sep=0pt,fill=white] at (6.5,1.25) {$\varkappa_1$};
\draw[-,thin] (6,1.4) -- (7.5,1.4);
\draw[->,thin] (5,0.2) -- (7,0.2);
\node at (6.5,0.4) {\scriptsize $p_\TO$};
\node at (7,-0.4) {$\varkappa_1 \Si \varkappa_2$};
\end{scope}
\end{tikzpicture}
\caption{The two cases in the proof of Lemma~\ref{period1}.}\label{fig:periodic-two-cases}
\end{figure}
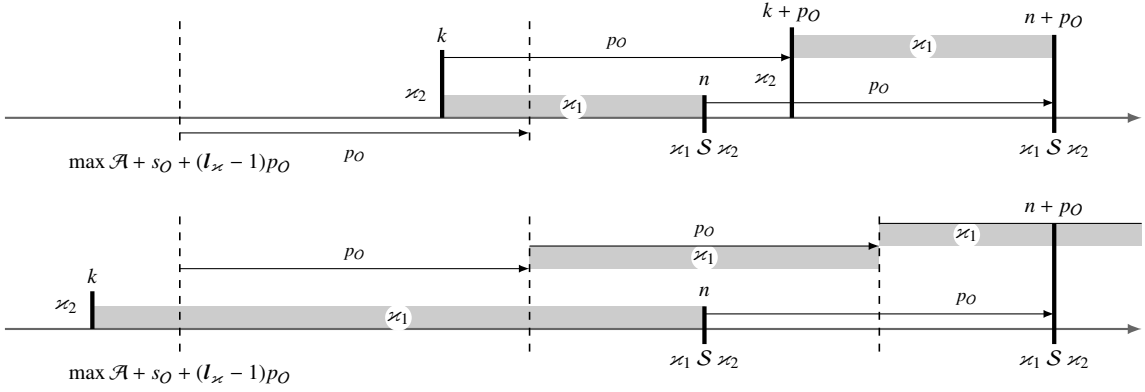

If $k \ge \max \A + s_\TO + (\len_\varkappa - 1)p_\TO$, then, since $\len_{\varkappa_1},\len_{\varkappa_2} \le \len_\varkappa -1$, we have, by the induction hypothesis, $k + p_\TO \in \varkappa_2^{\smash{\can}}$ and $m \in \varkappa_1^{\smash{\can}}$ for all $m$ with $k + p_\TO < m < n + p_\TO$. It follows that $n + p_\TO \in \varkappa^\can$, as required.

Otherwise, $k < \max \A + s_\TO + (\len_\varkappa - 1)p_\TO$, and we have $m'\in \varkappa_1^{\smash{\can}}$ for all $m'$ such that $\max \A + s_\TO + (\len_\varkappa - 1)p_\TO < m' \leq  \max \A + s_\TO + \len_\varkappa p_\TO$. Since $\len_{\varkappa_1} \le \len_\varkappa -1$, by the repeated application of the induction hypothesis, we obtain $m' \in \varkappa_1^{\smash{\can}}$ for all \mbox{$m' > \max \A + s_\TO + (\len_\varkappa - 1)p_\TO$}. Therefore, $n+p_\TO \in \varkappa^{\smash{\can}}$ as well.

The converse direction and other temporal operators are treated analogously. The claim for $n \le \min \A-s_\TO - \len_\varkappa p_\TO$ is proved by a symmetric argument.
\qed
\end{proof}

\section{From \OMAQ{}s to \OMPIQ{}s with Horn Ontologies: Phantoms and Rewritings}
\label{sec:ltl-tomiq:2}

It remains to put together the results on the periodic structure of canonical models $\C_{\TO,\A}$ with the FO-rewritings for $\LTL_{\smash{\horn}}\Xbox$ and  $\LTL_\core\Xallop$  \OMAQ{}s constructed in Theorems~\ref{thm:box} and~\ref{thm:krom-nextbox}.  As follows from Lemma~\ref{period1}, when constructing rewritings for \OMPIQ{}s, for example, of the form $(\TO,\Rdiamond A)$, we have to consider time instants beyond the active temporal domain $\tem(\A)$. For this purpose, we require the following definition.
\begin{definition}\label{def:phantom}\em
Let $\lang$ be one of $\FO(<)$, $\FOE$ or $\FO(\RPR)$.
Given an \OMAQ{} $\q = (\TO,A)$,
\emph{$\lang$-phantoms of $\q$}  are
$\lang$-sentences $\Phi^\kpar$, for $k \ne 0$, such that, for every data instance
$\A$ consistent with $\TO$,
\begin{equation*}
\SA \models \Phi^\kpar \ \ \text{ iff } \ \  \sigma_\A(\kpar)\in \ans^\Z(\q,\A), \qquad\text{ where } \sigma_\A(k) = \begin{cases}\max \A + k, & \text{if } k > 0,\\
\min \A+ k, & \text{if } k < 0.
\end{cases}
\end{equation*}
\end{definition}

$\FO(<)$-phantoms $\Phi^\kpar$ of an $\LTL_\horn\Xbox$ \OMAQ{} $\q=(\TO,A)$ can be constructed as follows.
Suppose $\kpar > 0$. For a state $[\tp]$ from the proof of Theorem~\ref{thm:box}, we write $[\tp] \to^\kpar [\tp']$ to say that $[\tp']$ can be reached from $[\tp]$ in~$\nfao^\circ$ by~$\kpar$-many transitions~$\to$. It is readily seen that
\begin{equation*}
\Phi^\kpar \ \ = \ \
\neg    \Bigl[\bigvee_{d \leq 2|\TO| + 1} \ \ \bigvee_{\begin{subarray}{c}[\tp_0]\to\ldots\to[\tp_d]\\\text{is a path in }\nfao^\circ \end{subarray}} \ \ \bigvee_{[\tp_d] \to^\kpar [\tp] \text{ with } A\notin\tp }
\exists t_0,\dots,t_d \,
\pathf_{[\tp_0]\to\ldots\to[\tp_d]}(t_0,\dots,t_d)\Bigr]
\end{equation*}
is as required (note that in this case there is no need to include an extra type to account for the answer $\ell$). The case of~$\kpar < 0$ is symmetric and left to the reader.

$\FOE$-phantoms $\Phi^\kpar$
of an $\LTL\Xallop_\core$ \OMAQ{} $\q=(\TO,A)$ can be constructed as follows.
Suppose $\kpar > 0$. Given a  type $\tp$ for
$\TO$, we write $\tp \to^\kpar \tp'$ to say that $\tp'$ can be reached from $\tp$ in~$\nfao$ by~$\kpar$-many transitions~$\to$ (in other words,
if there is a model $\Mmf$ of $\TO$ such that $\tp$ and $\tp'$ comprise the literals that are true  in $\Mmf$ at $0$ and $\kpar$, respectively).
It is not hard to see that
\begin{equation*}
\Phi^\kpar ~=~
\neg \Bigl[\bigvee_{d \leq 2|\TO| + 1} \ \ \bigvee_{\tp_0 \to^\Box \ldots \to^\Box \tp_d}  \ \ \bigvee_{\tp_d \to^k\tp \text{ with } A\notin\tp} \exists t_0,\dots,t_d \,
\pathf_{\tp_0\to^\Box\ldots\to^\Box\tp_d}(t_0,\dots,t_d)
\Bigr]
\end{equation*}
is as required. The case of $\kpar < 0$ is symmetric.

Equipped with these constructions, we are now in a position to prove the following:
\begin{theorem}\label{thm:LTLrewritability}
All $\LTL\Xallop_\core$ \OMPIQ{}s are
$\FOE$-rewritable, and all
$\LTL_{\smash{\horn}}\Xbox$ \OMPIQ{}s are
\mbox{$\FO(<)$}-rewritable.
\end{theorem}
\begin{proof}
Suppose $\q = (\TO,\varkappa)$ is an $\LTL_\core\Xallop$ \OMPIQ. By induction on the construction of $\varkappa$, we define an $\FOE$-rewriting $\rq(t)$ of $\q$ and $\FOE$-phantoms $\Phi^\kpar$ of $\q$ for $\kpar \ne 0$.  For an $\LTL_\horn\Xbox$ \OMPIQ, we define an $\FO(<)$-rewriting and $\FO(<)$-phantoms. As the only difference between these two cases is in the basis of induction, we simply refer to a rewriting and phantoms.
The basis (for rewritings of \OMAQ{}s) was established in
Sections~\ref{sec:ltl-tomaq:kromboxnext}
and~\ref{sec:ltl-tomaq:stutter}, respectively.

Suppose now that $\varkappa = \varkappa_1 \Si \varkappa_2$ and, for both $\q_i=(\TO, \varkappa_i)$, $i = 1,2$, we have the required rewriting $\rq_i(t)$ and phantoms $\Phi^\kpar_i$ for $\kpar \ne 0$. We claim that
\begin{equation*}
\rq(t)  \ = \ \exists s \, \bigl[ (s < t) \land
\rq_{2}(s) \land \forall r \, \bigl((s < r < t) \to
\rq_1(r)\bigr)\bigr]\  \lor \
 \bigl[ \forall s \,
  \bigl((s < t) \to \rq_1(s)\bigr)\ \  \land \hspace*{-1em}\bigvee_{-s_\TO - (\len_{\varkappa_2}+1) p_\TO  < k < 0} \hspace*{-1.2em}\bigl(
  \Phi^k_2  \ \land \bigwedge_{ k < i < 0}\hspace*{-0.5em}
  \Phi^i_1\bigr)\bigr]
\end{equation*}
is a rewriting of $\q$: for any data instance $\A$ and any $n \in \tem(\A)$, we have $n \in \ans(\q,\A)$ iff $\SA \models \rq(n)$. Indeed, suppose first that $n \in \ans(\q,\A)$. If $\A$ is inconsistent with $\TO$ then, by the induction hypothesis, $\SA \models \rq_i(n)$ for all $n \in \tem(\A)$, and $\SA \models \Phi^\kpar_i$ for all $\kpar \ne 0$, whence $\SA \models \rq(n)$ for all $n \in \tem(\A)$.
Otherwise, $\A$ is consistent with~$\TO$ and then, by Theorem~\ref{canonical}~(\emph{iv}), $n \in \varkappa^{\smash{\can}}$, which means that there is some $k < n$ such that $k \models \varkappa_2^{\smash{\can}}$ and $m \models \varkappa_1^{\smash{\can}}$ for $k < m < n$.
If $k \geq \min\A$, then, by the induction hypothesis, the first disjunct of $\rq$ holds at~$n$. If $k <  \min\A$, then, by Lemma~\ref{period1}, we can always  find such a~$k$ with $-s_\TO - (\len_{\varkappa_2} +1) p_\TO < k < 0$, and so, by the induction hypothesis, the second disjunct of $\rq$ holds at~$n$.
Conversely, suppose $\SA \models \rq(n)$. If the first disjunct of $\rq$ holds at~$n$, then there is $k \in \tem(\A)$ with $k < n$ such that $\SA \models \rq_2(k)$ and $\SA \models \rq_1(m)$ for $k < m < n$. By the induction hypothesis, $k \in \ans(\q_2,\A)$ and $m \in \ans(\q_1,\A)$ for all $k < m < n$, whence, by Theorem~\ref{canonical}~(\emph{iv}), $k \in \varkappa_2^{\smash{\can}}$ and $m \in \varkappa_1^{\smash{\can}}$  for all $k < m < n$, and therefore $n \in \varkappa^{\smash{\can}}$ and $n \in \ans(\q,\A)$. The case when the second disjunct of~$\rq$ holds at $n$ is considered analogously.

By Lemma~\ref{period1}, Theorem~\ref{canonical}~(\emph{iv}) and the induction hypothesis, we can define the phantoms $\Phi^\kpar$, for $k > 0$, by taking
\begin{multline*}
\Phi^\kpar \ \ \ = \ \ \  \bigvee_{0 < i < \kpar } \bigl(
  \Phi^i_2 \ \land \bigwedge_{ i < j < \kpar}\hspace*{-0.1em}
  \Phi^j_1\bigr) \ \ \ \lor{} \ \ \  \exists s \, \bigl[\,
\rq_{2}(s) \land \forall r \, \bigl((s < r) \to
\rq_1(r)\bigr) \land \bigwedge_{0 < i < \kpar}\hspace*{-0.1em} \Phi^i_1 \ \bigr] ~ \lor{}\hspace*{-10em} \\
 \Bigl[
  \bigvee_{-s_\TO - (\len_{\varkappa_2}+1) p_\TO  < i < 0}\hspace*{-0.7em} \bigl(
  \Phi^i_2 \ \land \bigwedge_{ i < j < 0}\hspace*{-0.1em}
  \Phi^j_1\bigr) \ \ \land \ \ \forall s \,\rq_1(s) \ \ \land \bigwedge_{0 < i < k} \hspace*{-0.1em}\Phi^i_1\ \Bigr].
\end{multline*}
Observe that, by Lemma~\ref{period1}, the sentence $\Phi^{\kpar + p_\TO}$ is equivalent to $\Phi^\kpar$, for all $\kpar \geq s_\TO + \len_{\varkappa} p_\TO$, and so there are only finitely many non-equivalent phantoms of $\q$. For $\kpar < 0$, the phantoms are constructed in a symmetric way.

The cases of other temporal operators are left to the reader. For $\varkappa = \varkappa_1 \land \varkappa_2$, the rewriting of $\varkappa$ is the conjunction of the rewritings for $\varkappa_1$ and $\varkappa_2$, and each phantom for $\varkappa$ is the conjunction of the corresponding phantoms for $\varkappa_1$ and~$\varkappa_2$. It is analogous for $\varkappa = \varkappa_1 \lor \varkappa_2$.
\qed
\end{proof}


\section{OMQs with $\MFO(<)$-Queries}\label{FO-OMQs}

So far we have considered temporal queries given entirely in the language of \LTL{} and having one implicit `answer variable'\!.
It follows that, in this language, one cannot formulate Boolean queries (queries without any answer variables) nor queries with multiple answer variables. For example, in the article submission scenario from the introduction, we might want to ask whether the article is accepted at all rather than when it is accepted, or we might want to retrieve all pairs of time points consisting of the last revision date of the article and its acceptance date. To overcome this deficiency, we now extend our query language to monadic $\FO(<)$-formulas with arbitrarily many answer variables and show that all of our rewritability results can be generalised to the extended language in a natural way using variants of Kamp's Theorem.

Denote by $\MFO(<)$ the set of first-order formulas that are built from atoms of the form $A(t)$, $(t < t')$ and abbreviations $(t=t')$ and~$(t=t'+1)$.
By an $\LTL^{\op}_\frag$ \emph{ontology-mediated query} (\OMQ{}, for short) we mean a pair $\q=(\TO, \psi(\avec{t}))$, where $\TO$ is an $\LTL^{\op}_\frag$ ontology and $\psi(\avec{t})$ an $\MFO(<)$-formula with
free variables $\avec{t}=(t_1,\dots,t_{m})$. The free variables $\avec{t}$ in $\psi(\avec{t})$ are called the \emph{answer variables} of~$\q$. If $\avec{t}$ is empty ($m=0$), then $\q$ is called a \emph{Boolean} \OMQ.

\begin{example}\em\label{ex:mfo-query}
We return to the article submission scenario from the introduction and formulate the informal
queries from above:  $\exists t\, \nm{Accept}(t)$ asks whether the article in question is accepted
and
\begin{equation*}
\psi(t,t') ~=~ \nm{Revise}(t) \wedge \nm{Accept}(t') \wedge \forall s \, \big( (t < s < t') \rightarrow \neg \nm{Revise}(s)\big)
\end{equation*}
retrieves all pairs $(t,t')$ such that $t$ is the last revision date before the acceptance date $t'$.
%
%
%
\end{example}

The $\MFO(<)$-formula $\psi(\avec{t})$ in an \OMQ{} $\q=(\TO, \psi(\avec{t}))$ is interpreted in models $\Mmf$ of $\TO$ and any data instance $\A$ as usual in first-order logic.
Now, assuming that $\avec{t}= (t_1,\dots, t_m)$ is nonempty, we call a tuple $\avec{\ell} = (\ell_1,\dots,\ell_m)$ of timestamps $\ell_i$ from $\tem(\A)$ a \emph{certain answer} to $\q$ over $\A$ if $\Mmf\models \psi(\avec{\ell})$, for every model $\Mmf$ of $(\TO,\A)$. If $\q$ is a Boolean \OMQ{}, then
we say that $\q$ is \emph{entailed over $\A$} if $\Mmf\models \psi$,
for every model $\Mmf$ of $(\TO,\A)$.

Let $\lang$ be one of the three classes of FO-formulas introduced above: $\FO(<)$, $\FOE$ or $\FO(\RPR)$. An \OMQ{} $\q=(\TO, \psi(\avec{t}))$ is $\lang$-\emph{rewritable} if there is an $\lang$-formula $\rq(\avec{t})$ such that, for any data instance $\A$, an $m$-tuple $\avec{\ell}$ in $\tem (\A)$ is a certain answer to $\q$ over $\A$ iff~$\SA \models \rq(\avec{\ell})$. In the case of Boolean $\q$, we need an $\lang$-sentence $\rq$ such that $\q$ is entailed over any given data instance $\A$  iff~$\SA \models \rq$.


To extend our rewritability results to the newly introduced \OMQ{}s, we employ~\cite[Proposition 4.3]{DBLP:journals/corr/Rabinovich14} from Rabinovich's `simple proof' of the celebrated Kamp's Theorem. For the reader's convenience, we formulate it below:

%
\begin{lemma}[\cite{DBLP:journals/corr/Rabinovich14}]\label{Rabino}
%
Any $\MFO(<)$-formula $\psi(\avec{t})$ with free variables $\avec{t} = (t_1,\dots,t_m)$, $m\geq 1$, is equivalent over $(\Z,<)$ to a disjunction $\varphi(\avec{t}) = \bigvee_{l=1}^k \varphi_l(\avec{t})$, in which the disjuncts $\varphi_l(\avec{t})$ take the form
\begin{multline}\label{kamp-gen}
\varphi_l(\avec{t}) ~=~ \exists x_1,\dots,x_n\, \Big[ \bigwedge_{i=1}^m (t_i = x_{j_i}) \ \ \land \ \ \bigwedge_{i = 1}^{n-1} (x_i < x_{i+1}) \ \ \land \ \ \bigwedge_{i=1}^n \alpha_i(x_i) \ \ \land{}\\
\forall y\, \big ( (y < x_1) \to \beta_0(y) \big ) \ \ \land \ \ \bigwedge_{i=1}^{n-1} \forall y\, \big ( (x_i < y < x_{i+1}) \to \beta_i(y) \big ) \ \ \land \ \
\forall y\, \big ( (x_n < y) \to \beta_n(y) \big )  \Big],
\end{multline}
for some sets of variables $\{x_1,\dots,x_n\} \supseteq \{x_{j_1}, \dots, x_{j_m}\}$ and Boolean combinations $\alpha_i(x_i)$ and $\beta_i(y)$ of unary atoms with one free variable $x_i$ and $y$, respectively \textup{(}all of which depend on $l$\textup{)}.
\end{lemma}

We begin by extending Theorem~\ref{fullLTL} to \OMQ{}s.

\begin{theorem}\label{fullLTL-OMQ}
All $\LTL\Xallop_\bool$ \OMQ{}s are $\FO(\RPR)$-, and so $\MSO(<)$-rewritable.
\end{theorem}
\begin{proof}
Let $\q = (\TO,\psi(\avec{t}))$ be an $\LTL\Xallop_\bool$ \OMQ{}. If $\psi$ has one answer variable, then we use Kamp's Theorem, according to which $\psi$ is equivalent to an \LTL-formula, and apply Theorem~\ref{fullLTL} to the latter.
If $\psi$ is a sentence (that is, $\q$ is Boolean), then we set $\psi'(t)= \psi\wedge (t=t)$ and take any rewriting $\rq'(t)$ of the \OMQ{} $\q' = (\TO,\psi'(t))$ with one answer variable; it should be clear that $\rq= \forall t \, \rq'(t)$ is a rewriting of $\q$.
	
Suppose now that there are at least two variables in $\avec{t} = (t_1,\dots,t_m)$. It is not hard to see using Lemma~\ref{Rabino}  that $\psi(\avec{t})$ is equivalent over $(\Z,<)$  to a disjunction of formulas of the form
\begin{equation}\label{kamp-nn}
\varphi(\avec{t}) ~=~  \textit{ord}(\avec{t}) \land \bigvee_{j = 1}^l \Big( \bigwedge_{i=1}^m \alpha_i^j(t_i) \land \bigwedge_{t \prec t'} \delta^j(t,t') \Big) ,
\end{equation}
where
\begin{itemize}
\item[--] $\textit{ord}(\avec{t})$ is a conjunction of atoms of the form $(t=t')$ and $(t<t')$, for $t,t' \in \avec{t}$, defining a total order on $\avec{t}$;

\item[--] the total orders given by distinct disjuncts $\varphi(\avec{t})$, if any, are  inconsistent with each other;

\item[--] $t \prec t'$ means that $t$ is the immediate predecessor of $t'$ in the order $\textit{ord}(\avec{t})$, and $\delta^j(t,t')$ takes the form
\begin{equation}\label{eq:fullLTL-OMQ-delta}
\delta^j(t,t') ~=~ \exists x_1,\dots,x_n\, \Big[ (t = x_1 < \dots < x_n = t') \land \bigwedge_{i=2}^{n-1} \gamma_i(x_i) \land \bigwedge_{i=1}^{n-1} \forall y\, \big ( (x_{i} < y < x_{i+1}) \to \beta_i(y) \big) \Big];
\end{equation}

\item[--] the $\alpha_i^j$, $\gamma_i$ and $\beta_i$ are $\FO(<)$-formulas with one free variable ($\gamma_1$\nz{added} subsumes $\alpha_1$ and the conjunct with $\beta_0$ from \eqref{kamp-gen}).
\end{itemize}
In fact, by adding, if necessary, disjuncts of the form $\varphi(\avec{t}) = \textit{ord}(\avec{t}) \land \bot$ to $\psi(\avec{t})$, we can make sure that every total order $\textit{ord}(\avec{t})$ on $\avec{t}$ corresponds to exactly one disjunct $\varphi(\avec{t})$ of $\psi(\avec{t})$.
By Kamp's Theorem, each of the $\alpha_i^j$, $\beta_i$ and $\gamma_i$ is equivalent to an \LTL-formula, so we think of them as such when convenient.
Let $\subq$ comprise all subformulas of the \LTL-formulas occurring in $\TO$ and  all of the $\smash{\alpha_i^j}$, $\beta_i$ and $\gamma_i$, together with their negations. By a \emph{type} for $\q$ we mean any maximal consistent subset $\tau \subseteq \subq$. For any temporal interpretation $\Mmf$ and $\ell \in \Z$, we denote by $\tau_\Mmf(\ell)$ the type for $\q$ defined by $\Mmf$ at $\ell$.

Let $\A$ be a data instance and $\avec{\ell} = (\ell_1,\dots,\ell_m)$ an $m$-tuple from $\tem(\A)$. We say that $\textit{ord}(\avec{t})$ \emph{respects} $\avec{\ell}$ if $t_i \mapsto \ell_i$ is an isomorphism from the total order defined by $\textit{ord}(\avec{t})$ onto the natural total order on $\avec{\ell}$. Then $\avec{\ell}$ is \emph{not} a certain answer to $\q$ over $\A$ iff either $\textit{ord}(\avec{t})$ respects $\avec{\ell}$ and $\textit{ord}(\avec{t}) \land \bot$ is in $\psi(\avec{t})$ or there is a model $\Mmf$ of $\TO$ and $\A$ such that, for the disjunct $\varphi(\avec{t})$ whose $\textit{ord}(\avec{t})$ respects $\avec{\ell}$, we have the following, for every $j$, $1 \leq j\leq l$:
\begin{itemize}
\item[--] either $\alpha_i^j \not \in \tau_\Mmf(\ell_i)$, for some $i$, $1\leq i \leq m$,

\item[--] or $\Mmf \not \models \delta^j(\ell,\ell')$, for some $\ell \prec \ell'$ in the given total order of $\avec{\ell}$.
\end{itemize}
Let $\varphi(\avec{t})$ respect $\avec{\ell}$.
To avoid clutter and without loss of generality, we assume that $\textit{ord}(\avec{t})$ is $t_1 < t_{2} < \dots < t_{m-1} < t_m$. Let $T$ be the set of all distinct (up to variable renaming) formulas $\delta^j(t,t')$ occurring in $\varphi(\avec{t})$ and let $S \subseteq T$. Suppose that, for any such $S$ and any types $\tau$ and $\tau'$, we have an $\FO(\RPR)$-formula $\pi_{\tp S \tp'}(t,t')$ with the following property:
\begin{itemize}
\item[($\star$)] $\SA \models \pi_{\tp S \tp'}(\ell,\ell')$ just in case there is a model $\Mmf$ of $\TO$ and the intersection $\A|_{\ge\ell} \cap \A|_{\le\ell'}$ such that $\tp_\Mmf(\ell) = \tp$, $\tp_\Mmf(\ell') = \tp'$, and $\Mmf \models \delta^j(\ell,\ell')$ iff $\delta^j(t,t') \in S$.
\end{itemize}
Consider the $\FO(\RPR)$-formula
\begin{equation*}
\varphi'(\avec{t}) ~=~ \bigvee_{\tau_1 S_1 \tau_2 \dots \tau_{m-1} S_{m-1} \tau_m} \Big(\bigwedge_{i=1}^{m-1} \pi_{\tau_i S_i \tau_{i+1}}(t_i, t_{i+1}) \land \bigwedge_{j=1}^ l \big( \bigvee_{i=1}^m \alpha_i^j \notin \tau_i \lor  \bigvee_{i=1}^{m-1} \delta^j \notin S_i \big) \Big),
\end{equation*}
where $\alpha_i^j \notin \tau_i$ is $\top$ if the type $\tau_i$ does \emph{not} contain $\alpha_i^j$ and $\bot$ otherwise, and similarly for $\delta^j \notin S_i$. If $\varphi(\avec{t}) = \textit{ord}(\avec{t}) \land \bot$, we set $\varphi'(\avec{t}) = \top$. Then it is not hard to verify (as we did it in the proof of Theorem~\ref{fullLTL}) that $\SA \models \varphi'(\avec{\ell})$ iff there exists a model $\Mmf$ of $\TO$ and $\A$ with $\Mmf \not \models \varphi(\avec{\ell})$. It follows that a conjunction of $\textit{ord}(\avec{t}) \to \neg \varphi'(\avec{t})$, for all orders $\textit{ord}(\avec{t})$, is an $\FO(\RPR)$-rewriting of $\q$.

In the remainder of the proof, we define the required $\pi_{\tp S \tp'}(t,t')$ in $\FO(\RPR)$. First, we represent every formula $\delta \in T$ of the form~\eqref{eq:fullLTL-OMQ-delta}, by the \LTL-formula
\begin{equation*}
\tilde \delta ~=~ \beta_1 \U (\gamma_2 \land (\beta_{2}  \dots \beta_{n-2} \U (\gamma_{n-1} \land \beta_{n-1} \U M) \dots )),
\end{equation*}
for a fresh atomic concept $M$ for `marker' (which never occurs in data instances). Intuitively, for any $i < j \in \Z$ and any model $\Mmf$ of $\TO$ with $M^{\Mmf} = \emptyset$, we have $\Mmf \models \delta(i,j)$ iff $i \in \tilde \delta^{\mathfrak N}$, where $\mathfrak N$ is $\Mmf$ extended with $M^{\mathfrak N} = \{j\}$.
 Let $\textit{un}(\tilde \delta)$ be the set of subformulas of the form $\beta \U \beta'$ in $\tilde \delta$. An \emph{extended type}, $\sigma$, is any maximal consistent subset of the set comprising $\subq$, all subformulas of $\tilde \delta$, for $\delta \in T$, and negations thereof.
%
%
Let $\sigma_1, \dots, \sigma_k$ be the set of all extended types.

To explain the intuition behind our $\FO(\RPR)$-formulas, we assume, as in Section~\ref{sec:ltl-tomaq:stutter}, that a data instance~$\A$ is given as a sequence $\A_0,\dots,\A_n$ with $\A_i = \{\, A \mid A(i) \in \A\,\}$.
Consider an NFA $\mathfrak A$ (similar to the one in Section~\ref{sec:ltl-tomaq:stutter}) with states $\sigma_i$, $1 \le i \le k$, and the transition relation  $\sigma \to_{\A_i} \sigma'$ iff $\suc(\sigma, \sigma')$ and $\A_i \subseteq \sigma'$ (see the proof of Theorem~\ref{fullLTL}). Then there is a run $\sigma_{j_{-1}}, \sigma_{j_0}, \dots, \sigma_{j_n}$ of $\mathfrak A$ on $\A$ with $M \notin \sigma_{j_i}$, for all $i$, $0 \leq i < n$, and $\{M\} \cup \{ \neg \varkappa \mid \varkappa \in \textit{un}(\tilde \delta),\  \delta \in T\} \subseteq \sigma_{j_n}$ iff there exists a model $\Mmf$ of $\TO$ and $\A$ such that $\Mmf \models \delta(0,n)$ for all $\tilde \delta \in \sigma_{j_0}$ and $\Mmf \not \models \delta(0,n)$ for all $\tilde \delta \notin \sigma_{j_0}$. The formula $\pi_{\tp S \tp'}(t_0, t_1)$ below encodes the existence of such a run of $\mathfrak A$ on $\A$ similarly to the proof of Theorem~\ref{fullLTL}:
\begin{align*}
& \pi_{\tp S \tp'}(t_0, t_1) = \bigvee_{\sigma \ \supseteq \ \tp \ \cup \ \{\, \tilde \delta \,\mid\, \delta \in S\,\}\  \cup\  \{\, \neg \tilde \delta \,\mid\, \delta \in T \setminus S\,\}} \hspace*{-2em}\psi_{\sigma, \tp'} (t_0, t_1),\\
& \psi_{\sigma, \tp'}(t_0, t_1) \ \ = \ \ \left[ \begin{array}{l}
Q_{\sigma_1}(t_0,t) \equiv \eta^{\sigma}_{\sigma_1}\\
\dots\\
Q_{\sigma_k}(t_0,t) \equiv \eta^{\sigma}_{\sigma_k}
\end{array}\right] \ \ \bigvee_{\substack{\suc(\sigma_i, \sigma'),\\  \sigma' \ \supseteq\  \{M\}  \ \cup \ \tau' \ \cup \ \{\, \neg \varkappa \,\mid\, \varkappa \in \textit{un}(\tilde \delta),\,  \delta \in T\, \}}} \hspace*{-4em}\big(Q_{\sigma_i}(t_0,t_1 - 1)\land \typef_{\sigma'}(t_1)\big),
\end{align*}
where
\begin{equation*}
\eta^{\sigma}_{\sigma_i}(t_0, t, Q_{\sigma_1}(t_0,t-1),\dots,Q_{\sigma_k}(t_0,t-1)) ~=~
\begin{cases}
\bot, & \text{ if } M \in \sigma_i \cup \sigma;\\
\displaystyle
\typef_{\sigma_i}(t) \land \bigl((t=t_0) \lor \bigvee_{\suc(\sigma', \sigma_i)} Q_{\sigma'}(t_0,t-1)\bigr), & \text{ if } \sigma_i = \sigma;\\ \displaystyle
\typef_{\sigma_i}(t) \land  \bigvee_{\suc(\sigma', \sigma_i)} Q_{\sigma'}(t_0,t-1), & \text{ if } \sigma_i \ne \sigma.
\end{cases}
\end{equation*}
It is not hard to show that $\pi_{\tp S \tp'}(t,t')$ satisfies ($\star$).
\qed
\end{proof}

Our next aim is to extend our $\FO(<)$- and $\FOE$-rewritability results for OMPIQs with Horn ontologies to $\MFO(<)$-queries with multiple answer variables. To this end, we have to generalise the notion of positive temporal concept to first-order logic. As an obvious natural candidate, one could take the set of $\MFO(<)$-formulas constructed from the usual atoms with the help of the connectives $\land$, $\lor$ and quantifiers $\forall$, $\exists$. This, however, would not cover all positive temporal concepts as the binary operators $\Si$ and $\U$ are not expressible in this language. To include them in a natural way, we define our second query language.

A \emph{quasi-positive} $\MFO(<)$-\emph{formula} is built (inductively) from atoms of the form $A(t)$, $(t < t')$, $(t=t')$ and $(t = t'+1)$  using $\land$, $\lor$, $\forall$, $\exists$ as well as the guarded universal quantification of the form
\begin{equation*}
\forall y \, \big( (x < y < z) \to \varphi \big),\qquad \forall y \, \big( (x < y) \to \varphi \big), \qquad \forall y \, \big( (y < z) \to \varphi \big),
\end{equation*}
where $\varphi$ is a quasi-positive $\MFO(<)$-formula.
A \emph{quasi-positive $\LTL^{\op}_\frag$ \OMQ{}} is then a pair $\q = (\TO,\psi(\avec{t}))$ where $\TO$ is an $\LTL^{\op}_\frag$ ontology and $\psi(\avec{t})$ a quasi-positive $\MFO(<)$-formula.


We now prove a transparent semantic characterisation of quasi-positive $\MFO(<)$-formulas. Given temporal interpretations $\Mmf_1$ and $\Mmf_2$, we write $\Mmf_1 \preceq \Mmf_2$ if $A^{\Mmf_1} \subseteq A^{\Mmf_2}$, for every atomic concept $A$. An $\MFO(<)$-formula $\psi(\avec{t})$ is called \emph{monotone} if $\Mmf_1 \models \psi(\avec{n})$ and $\Mmf_1 \preceq \Mmf_2$ imply $\Mmf_2 \models \psi(\avec{n})$, for any tuple  $\avec{n}$ in $\Z$.

\begin{theorem}\label{positive Kamp}
An $\MFO(<)$-formula is monotone iff it is equivalent over $(\Z,<)$ to a quasi-positive $\MFO(<)$-formula.
\end{theorem}
\begin{proof}
$(\Rightarrow)$ Suppose a monotone $\MFO(<)$-formula $\psi(\avec{t})$ with
$\avec{t} = (t_1,\dots,t_m)$ and $m>0$ is given. By Lemma~\ref{Rabino}, we may assume that $\psi(\avec{t})$ is in fact  a disjunction $\varphi(\avec{t}) =  \bigvee_{l=1}^k \varphi_l(\avec{t})$, where each $\varphi_l(\avec{t})$ takes the form~\eqref{kamp-gen} and the $\alpha_i(x_i)$ and $\beta_i(y)$ are Boolean combinations of unary atoms with one free variable $x_i$ and $y$, respectively.
Let $\gamma(z)$ be any of these Boolean combinations all of whose atoms are among $A_1(z),\dots,A_s(z)$. Denote by $\gamma'(z)$ a formula in DNF that contains a disjunct $A_{i_1}(z) \land \dots \land A_{i_p}(z)$ iff $\gamma(z)$ is true as a propositional  Boolean formula under the valuation $\mathfrak v \colon \{A_1(z),\dots,A_s(z)\} \to \{0,1\}$ such that $\mathfrak v(A_i(z)) = 1$ for $i \in \{i_1,\dots,i_p\}$ and $\mathfrak v(A_i(z)) = 0$ for the remaining $i$. (We remind the reader that an empty disjunction is $\bot$, which is a quasi-positive $\MFO(<)$-formula.) Clearly, $\gamma(z) \to \gamma'(z)$ is a tautology.

Denote by $\varphi'(\avec{t})$ and $\varphi'_l(\avec{t})$ the results of replacing the $\alpha_i(x_i)$ and $\beta_i(y)$ in $\varphi(\avec{t})$ and, respectively, $\varphi_l(\avec{t})$ by their primed versions $\alpha'_i(x_i)$ and $\beta'_i(y)$. We claim that $\varphi(\avec{t})$ and $\varphi'(\avec{t})$ are equivalent. Clearly, we have $\Mmf \models^{\mathfrak a} \varphi(\avec{t}) \to \varphi'(\avec{t})$, for any temporal interpretation $\Mmf$ and assignment $\mathfrak a$.
Suppose $\Mmf \not\models^{\mathfrak a} \varphi'(\avec{t}) \to \varphi(\avec{t})$. Denote by $\hat{\varphi}_l$ and $\hat{\varphi}'_l$ the matrices (that is, the quantifier-free parts)  of $\varphi_l$ and, respectively, $\varphi'_l$. Then there is an assignment $\mathfrak b \colon \{x_1,\dots,x_n\} \to \Z$ such that $\Mmf \models^{\mathfrak b} \varphi'_r$, for some $r \in \{1,\dots,k\}$, and $\Mmf \not\models^{\mathfrak b} \varphi_l$, for all $l \in \{1,\dots,k\}$. Consider any $z \in \Z$. Let $\gamma(z)$ be $\beta_0(z)$ if $z < \mathfrak b(x_1)$, $\alpha_1(z)$ if $z = \mathfrak b(x_1)$, $\beta_1(z)$ if $\mathfrak b(x_1) < z < \mathfrak b(x_2)$, $\alpha_2(z)$ if $z = \mathfrak b(x_2)$, etc. Then we have $\Mmf \models^{\mathfrak b} \gamma'(z)$ for all $z \in \Z$. If $\Mmf \not\models^{\mathfrak b} \gamma(z)$, we do the following. The interpretation $\Mmf$ gives a valuation $\mathfrak v$ for the atoms  $A_1(z),\dots,A_s(z)$ in $\gamma$ such that $\mathfrak v(\gamma(z)) = 0$ while $\mathfrak v(\gamma'(z)) = 1$. Suppose a disjunct $A_{i_1}(z) \land \dots \land A_{i_p}(z)$ in $\gamma'(z)$ is true under $\mathfrak v$. Then we remove from $\Mmf$ all atoms in the set $\{A_1(z),\dots,A_s(z)\} \setminus \{A_{i_1}(z), \dots, A_{i_p}(z)\}$. We do the same for all $z \in \Z$ with $\Mmf \not\models^{\mathfrak b} \gamma(z)$ and denote the resulting interpretation by $\Mmf'$. Obviously, $\Mmf' \preceq \Mmf$. Since $\psi =\varphi$ is monotone, we must have $\Mmf' \not\models^{\mathfrak a} \varphi(\avec{t})$. On the other hand, by the definition of $\gamma'$, we must have $\Mmf \models^{\mathfrak b} \varphi_r$, which is impossible.

Thus, $\varphi$ is equivalent over $(\Z,<)$ to $\varphi'$, which is, by construction, a quasi-positive $\MFO(<)$-formula.
The case $m=0$ is treated in the same way as in the proof of Theorem~\ref{fullLTL-OMQ}.

The implication $(\Leftarrow)$ is readily shown by induction on the construction of quasi-positive formulas.
\qed
\end{proof}

The following is an immediate consequence of the proof of Theorem~\ref{positive Kamp}:

\begin{corollary}\label{cor-mono}
Every OMQ $\q = (\TO,\psi(\avec{t}))$ with a monotone $\MFO(<)$-query $\psi(\avec{t})$ is equivalent to a quasi-positive OMQ $\q' = (\TO,\psi'(\avec{t}))$, where $\psi'(\avec{t})$ is a disjunction of formulas of the form
\begin{equation}\label{kamp-mono}
\varphi_l(\avec{t}) ~=~ \exists x_1,\dots,x_n\, \Big[ \bigwedge_{i=1}^m (t_i = x_{j_i}) \land \bigwedge_{i = 1}^{n-1} (x_i < x_{i+1}) \land \bigwedge_{i=1}^n \alpha_i(x_i) \land \bigwedge_{i=1}^{n-1} \forall y\, \big ( (x_i < y < x_{i+1}) \to \beta_i(y) \big ) \Big],
\end{equation}
where the first conjunction contains $(t_i = x_1)$ and $(t_j = x_n)$ with free variables  $t_i,t_j \in \avec{t}$, and the $\alpha_i$ and $\beta_i$ are $\MFO(<)$-formulas with one free variable.
\end{corollary}

\begin{corollary}[Monotone Kamp Theorem]
Every monotone $\MFO(<)$-formula with one free variable is equivalent over $(\Z,<)$ either to a positive temporal concept or to $\bot$.
\end{corollary}
\begin{proof}
In this case, the quasi-positive formula $\varphi'_l$ of the form~\eqref{kamp-gen}  defined in the proof of Theorem~\ref{positive Kamp} has one free variable. If one of the $\alpha'_i(x_i)$ and $\beta'_i(y)$ is $\bot$, then $\varphi'_l$ is equivalent to $\bot$. Otherwise, since the $\alpha'_i(x_i)$ and $\beta'_i(y)$ in $\varphi'_l$ are positive, we can transform $\varphi'_l$ into an equivalent positive \LTL{} concept in exactly the same way as in the proof of~\cite[Proposition 3.5]{DBLP:journals/corr/Rabinovich14}.
\qed
\end{proof}

We are now in a position to prove a generalisation of Theorem~\ref{thm:LTLrewritability} to quasi-positive OMQs with multiple answer variables.

\begin{theorem}\label{thm:rewritability-monFO}
All quasi-positive $\LTL\Xallop_\core$ \OMQ{}s are $\FOE$-rewritable, and all quasi-positive $\LTL_{\smash{\horn}}\Xbox$ \OMQ{}s are \mbox{$\FO(<)$}-rewritable.
\end{theorem}
\begin{proof}
Let $\q=(\TO, \psi(\avec{t}))$ be a quasi-positive \OMQ{} in $\LTL\Xallop_\core$ or $\LTL_\horn\Xbox$ with $\avec{t} = (t_1,\dots,t_m)$. If $m=0$, then we set $\psi'(t)= \psi\wedge (t=t)$ and take any rewriting $\rq'(t)$ of $\q' = (\TO,\psi'(t))$, which gives us the rewriting $\rq(t)= \forall t \, \rq'(t)$ of $\q$.

Suppose $m \geq 1$ and $\psi(\avec{t})$ takes the form specified in  Corollary~\ref{cor-mono}.
Since $\psi(\avec{t})$ is monotone, for every $\A$ consistent with $\TO$ and every tuple $\avec{\ell} = (\ell_1,\dots,\ell_m)$ in $\tem(\A)$, we have:
\begin{equation}\label{eq:can-monFO}
\avec{\ell}\in\ans(\q,\A) \quad  \text{ iff } \quad\C_{\TO,\A} \models \psi(\avec{\ell}).
\end{equation}
The implication $(\Rightarrow)$ follows from Theorem~\ref{canonical} (\emph{iii}). To prove $(\Leftarrow)$, observe that, by Theorem~\ref{canonical} (\emph{ii}), $\C_{\TO,\A} \preceq \Mmf$ for every model $\Mmf$ of $(\TO,\A)$, and so, by monotonicity, $\C_{\TO,\A} \models \psi(\avec{\ell})$ implies $\Mmf \models \psi(\avec{\ell})$. Next, observe that, by Theorem~\ref{canonical}~{\rm (\emph{iii})}, there exists an FO-formula $\rq_\bot(t)$ such that $\SA \models \exists t \, \rq_\bot(t)$ iff $(\TO,\A)$ is inconsistent. Indeed, it is sufficient to take as $\rq_\bot(t)$ the FO-rewriting provided by Theorem~\ref{thm:LTLrewritability}  for the \OMPIQ{} $\q_\bot ~=~ \big(\TO, \varkappa_\bot)$, where $\varkappa_\bot$ is a disjunction of all
\begin{equation*}
\Ldiamond \Rdiamond(C_1 \land \dots \land C_n), \qquad \text{ for } C_1 \land \dots \land C_n \to \bot \text{ in } \TO.
\end{equation*}

Recall that $\psi(\avec{t})$ is equivalent to $\bigvee_{l=1}^k \varphi_l(\avec{t})$, where each $\varphi_l(\avec{t})$ takes the form~\eqref{kamp-mono}.
Let $\varkappa_i$ and $\lambda_i$ be positive temporal concepts equivalent to $\alpha_i$ and, respectively, $\beta_i$ (see~\eqref{kamp-mono}) that are different from $\bot$. Let $\rq_i'(t)$ and $\rq_i''(t)$ be the FO-rewritings of the \OMPIQ{}s $(\TO, \varkappa_i)$ and, respectively, $(\TO, \lambda_i)$ provided by Theorem~\ref{thm:LTLrewritability}; if $\alpha_i$ or $\beta_i$ is $\bot$, then we set $\rq_i'(t) = \bot$ or, respectively, $\rq_i''(t)= \bot$. In each $\varphi_l$, we substitute every $\alpha_i$ by $\rq_i'$ and every $\beta_i$ by $\rq_i''$. Let $\varphi_l'$ be the result of the substitution and let $\rq(\avec{t})= \bigvee_{l=1}^k \varphi_l'(\avec{t})$. By~\eqref{eq:can-monFO}, $\rq(\avec{t}) \lor \exists t\, \rq_\bot(t)$ is an FO-rewriting of $\q$. Moreover, by Theorem~\ref{thm:LTLrewritability}, this rewriting is an $\FOE$-formula if $\q$ is an $\LTL\Xallop_\core$ \OMQ{}, and an \mbox{$\FO(<)$}-formula if $\q$ is an $\LTL_\horn\Xbox$ \OMQ{}.
\qed
\end{proof}


\section{OMQs with $\FO(<)$-Queries under the Epistemic Semantics}\label{epistecic}

Instead of generalising \LTL{}-queries in \OMQ{}s to $\MFO(<)$-queries using the standard open-world semantics for ontology-mediated query answering, one can increase the expressive power of \OMQ{}s by using \LTL{}-queries as building blocks for first-order queries under a closed-world semantics. We sketch an implementation of this idea following the approach of~\citeauthor{CalvaneseGLLR07}~\cite{CalvaneseGLLR07} and inspired by the SPARQL~1.1 entailment regimes~\cite{GlimmOgbuji13}. An \emph{epistemic temporal} \emph{ontology-mediated query} (EOMQ, for short) is a pair $\q=(\TO, \psi(\avec{t}))$, in which $\TO$ is an ontology and $\psi(\avec{t})$ a first-order formula built from atoms of the form $\varkappa(t)$, $(t<t')$ and abbreviations $(t=t')$ and $(t=t'+1)$, where $\varkappa$ is a positive temporal concept and the tuple $\avec{t}$ of  the free variables in $\psi(\avec{t})$ comprises the \emph{answer variables} of~$\q$.
Given a data instance $\A$, the FO-component of the EOMQ~$\q$ is evaluated over the structure $\mathfrak{G}_{\TO,\A}$ with domain $\tem(\A)$.
Let $\mathfrak{a}$ be an assignment that
maps variables to elements of $\tem(\A)$.
The relation $\GOA\models^\mathfrak{a} \psi$ (`$\psi$ is true in
$\GOA$ under $\mathfrak{a}$') is formally defined as follows:
\begin{align*}
\GOA\models^\mathfrak{a}\varkappa(t)\quad & \text{ iff }\quad \mathfrak{a}(t)\in \ans(\TO,\varkappa,\A),\\
\GOA\models^\mathfrak{a} t < t' \quad & \text{ iff } \quad \mathfrak{a}(t) < \mathfrak{a}(t'),
\end{align*}
and the standard clauses for the Boolean connectives and first-order quantifiers over $\tem(\A)$.
Let $\avec{t} = (t_1,\dots,t_m)$ be the free variables of
$\psi$. We say that
$\avec{\ell} = (\ell_1,\dots,\ell_m)$ is an \emph{answer} to the \OMQ{} $\q=(\TO,\psi(\avec{t}))$ over $\A$ if $\GOA\models \psi(\avec{\ell})$.

Thus, similarly to the SPARQL~1.1 entailment regimes, we interpret the
(implicit quantifiers in the) \LTL{} formulas of EOMQs in arbitrary temporal models over $\Z$, while the temporal variables of EOMQs range over the active
domain only. The first-order constructs in $\psi$ are interpreted under the  epistemic semantics~\cite{CalvaneseGLLR07}.

\begin{example}\em
Consider the \OMQ{} $\q = (\TO,\psi(t))$, where $\TO$ is the ontology
defined in the introduction and
\begin{equation*}
\psi(t)= \nm{Submission}(t) \wedge \neg \Rdiamond \bigl(\nm{Accept} \vee \nm{Reject}\bigr)(t).
\end{equation*}
Then $\q$ retrieves the submission date of the article if no accept or reject decision has yet been made according to the database and for which no publication
date is in the database.
\end{example}

Let $\lang$ be one of the three classes $\FO(<)$, $\FOE$ or $\FO(\RPR)$. An EOMQ $\q=(\TO, \psi(\avec{t}))$ is called $\lang$-\emph{rewritable} if there is an $\lang$-formula $\rq(\avec{t})$ such that, for any data instance $\A$ and any tuple  $\avec{\ell}$ in $\tem (\A)$, $\avec{\ell}$ is an answer to $\q$ over $\A$ iff~$\SA \models \rq(\avec{\ell})$.
It is readily seen that to construct an FO-rewriting of~$\q$ one can replace all occurrences of OMPIQs in~$\psi$ with their FO-rewritings (if any). Thus we obtain the following variant of a result from~\cite{CalvaneseGLLR07}:
\begin{theorem}\label{th:constitute}
Let $\lang$ be one of $\FO(<)$, $\FOE$ or $\FO(\RPR)$.
If all constituent OMPIQs of an EOMQ $\q$ are $\lang$-rewritable,
then so is $\q$.
\end{theorem}


%

\section{Conclusions and Future Work}

The main contributions of this article are as follows. Aiming to extend the well-developed theory of ontology-based data access (OBDA) to temporal data, we introduced, motivated and systematically investigated the problem of FO-rewritability for temporal ontology-mediated queries based on linear temporal logic \LTL{}.
We classified the OMQs by the shape of their ontology axioms (core, horn, krom or bool)
and also by the temporal operators that can occur in the ontology axioms. We then identified the classes whose OMQs are $\FO(<)$-, $\FOE$- or $\FO(\RPR)$-rewritable by establishing a connection of OMQ answering to various types of finite automata and investigating the structure of temporal models. 

While we believe that the results obtained in this article are of interest in
themselves, our second main aim was to lay the foundations for studying ontology-mediated query answering over not only propositional but also relational timestamped data using two-dimensional OMQs, in which the temporal component is captured by \LTL{} and the domain component by a description logic from the \DL{} family. In fact, in a subsequent article we are going to present a novel technique that allows us to lift the rewritability results obtained above to two-dimensional (\DL{} + \LTL) OMQs with Horn role inclusions. In particular, we show $\FO(<)$-rewritability results for suitable combinations of $\DL_{\core}$ with $\LTL_{\core}^{\Box}$ and $\FOE$-rewritability results for suitable combinations of $\DL_{\core}$ with  $\LTL_\core\Xallop$, in both cases by using results presented in this article.

Many interesting and challenging open research questions, both theoretical and practical, are arising in the context of the results obtained in this article. Below, we briefly mention some of them.
The classical OBDA theory has recently investigated the fine-grained combined and parameterised complexities of OMQ answering and the succinctness problem for FO-rewritings~\cite{DBLP:journals/jacm/BienvenuKKPZ18,DBLP:conf/pods/BienvenuKKPRZ17,DBLP:conf/dlog/BienvenuKKRZ17}. These problems are of great importance for the temporal case, too (in particular, because the presented rewritings are far from optimal). Another development in the classical OBDA theory is the classification of single ontologies and even OMQs according to their data complexity and rewritability~\cite{DBLP:journals/lmcs/LutzW17,DBLP:conf/ijcai/LutzS17,DBLP:conf/pods/HernichLPW17}. Extending this approach to temporal OMQs will most probably require totally different methods because of the linearly ordered temporal domain. 

On the practical side, more real-world use cases are needed to understand which temporal constructs are required to specify relevant temporal events and evaluate the performance of OMQ rewritings; for some activities in this direction, we refer the reader to~\cite{DBLP:conf/aaai/BeckDEF15,DBLP:journals/ws/KharlamovMSXKR19,brandt2017ontology,TIME19}.

\section*{Acknowledgements}

This research was supported by the EPSRC U.K.\ grants EP/S032207 and EP/S032282 for the joint project  `\textsl{quant$^\textsl{MD}$}: Ontology-Based Management for Many-Dimensional Quantitative Data'\!. The authors are grateful to the anonymous referees for their careful reading, valuable comments and constructive suggestions.

%


\bibliographystyle{elsarticle-num-names}
\bibliography{LTL,time17-bib}

\begin{thebibliography}{91}
\expandafter\ifx\csname natexlab\endcsname\relax\def\natexlab#1{#1}\fi
\providecommand{\url}[1]{\texttt{#1}}
\providecommand{\href}[2]{#2}
\providecommand{\path}[1]{#1}
\providecommand{\DOIprefix}{doi:}
\providecommand{\ArXivprefix}{arXiv:}
\providecommand{\URLprefix}{URL: }
\providecommand{\Pubmedprefix}{pmid:}
\providecommand{\doi}[1]{\href{http://dx.doi.org/#1}{\path{#1}}}
\providecommand{\Pubmed}[1]{\href{pmid:#1}{\path{#1}}}
\providecommand{\bibinfo}[2]{#2}
\ifx\xfnm\relax \def\xfnm[#1]{\unskip,\space#1}\fi
\bibitem[{Calvanese et~al.(2007)Calvanese, {De Giacomo}, Lembo, Lenzerini, and
  Rosati}]{DBLP:journals/jar/CalvaneseGLLR07}
\bibinfo{author}{D.~Calvanese}, \bibinfo{author}{G.~{De Giacomo}},
  \bibinfo{author}{D.~Lembo}, \bibinfo{author}{M.~Lenzerini},
  \bibinfo{author}{R.~Rosati},
\newblock \bibinfo{title}{Tractable reasoning and efficient query answering in
  description logics: The \emph{DL-Lite} family},
\newblock \bibinfo{journal}{J. Autom. Reasoning} \bibinfo{volume}{39}
  (\bibinfo{year}{2007}) \bibinfo{pages}{385--429}.
\bibitem[{Artale et~al.(2009)Artale, Calvanese, Kontchakov, and
  Zakharyaschev}]{DBLP:journals/jair/ArtaleCKZ09}
\bibinfo{author}{A.~Artale}, \bibinfo{author}{D.~Calvanese},
  \bibinfo{author}{R.~Kontchakov}, \bibinfo{author}{M.~Zakharyaschev},
\newblock \bibinfo{title}{The {DL-Lite} family and relations},
\newblock \bibinfo{journal}{J. Artif. Intell. Res. {(JAIR)}}
  \bibinfo{volume}{36} (\bibinfo{year}{2009}) \bibinfo{pages}{1--69}.
\bibitem[{Antonioli et~al.(2014)Antonioli, Castan{\`{o}}, Coletta, Grossi,
  Lembo, Lenzerini, Poggi, Virardi, and
  Castracane}]{DBLP:conf/fois/AntonioliCCGLLPVC14}
\bibinfo{author}{N.~Antonioli}, \bibinfo{author}{F.~Castan{\`{o}}},
  \bibinfo{author}{S.~Coletta}, \bibinfo{author}{S.~Grossi},
  \bibinfo{author}{D.~Lembo}, \bibinfo{author}{M.~Lenzerini},
  \bibinfo{author}{A.~Poggi}, \bibinfo{author}{E.~Virardi},
  \bibinfo{author}{P.~Castracane},
\newblock \bibinfo{title}{Ontology-based data management for the {I}talian
  public debt},
\newblock in: \bibinfo{booktitle}{Proc.\ of the 8th Int.\ Conf.\ on Formal
  Ontology in Information Systems, FOIS 2014}, \bibinfo{publisher}{IOS Press},
  \bibinfo{year}{2014}, pp. \bibinfo{pages}{372--385}.
\bibitem[{Bail et~al.(2012)Bail, Alkiviadous, Parsia, Workman, Van~Harmelen,
  Goncalves, and Garilao}]{FishMarkPaper}
\bibinfo{author}{S.~Bail}, \bibinfo{author}{S.~Alkiviadous},
  \bibinfo{author}{B.~Parsia}, \bibinfo{author}{D.~Workman},
  \bibinfo{author}{M.~Van~Harmelen}, \bibinfo{author}{R.~S. Goncalves},
  \bibinfo{author}{C.~Garilao},
\newblock \bibinfo{title}{Fishmark: A linked data application benchmark},
\newblock in: \bibinfo{booktitle}{Proc.\ of SSWS+HPCSW 2012},
  \bibinfo{publisher}{CEUR-WS}, \bibinfo{year}{2012}, pp.
  \bibinfo{pages}{1--15}.
\bibitem[{Giese et~al.(2015)Giese, Soylu, Vega{-}Gorgojo, Waaler, Haase,
  Jim{\'{e}}nez{-}Ruiz, Lanti, Rezk, Xiao, {\"{O}}z{\c{c}}ep, and
  Rosati}]{DBLP:journals/computer/GieseSVWHJLRXOR15}
\bibinfo{author}{M.~Giese}, \bibinfo{author}{A.~Soylu},
  \bibinfo{author}{G.~Vega{-}Gorgojo}, \bibinfo{author}{A.~Waaler},
  \bibinfo{author}{P.~Haase}, \bibinfo{author}{E.~Jim{\'{e}}nez{-}Ruiz},
  \bibinfo{author}{D.~Lanti}, \bibinfo{author}{M.~Rezk},
  \bibinfo{author}{G.~Xiao}, \bibinfo{author}{{\"{O}}.~L. {\"{O}}z{\c{c}}ep},
  \bibinfo{author}{R.~Rosati},
\newblock \bibinfo{title}{Optique: Zooming in on big data},
\newblock \bibinfo{journal}{{IEEE} Computer} \bibinfo{volume}{48}
  (\bibinfo{year}{2015}) \bibinfo{pages}{60--67}.
\bibitem[{Calvanese et~al.(2011)Calvanese, {De Giacomo}, Lembo, Lenzerini,
  Poggi, Rodriguez{-}Muro, Rosati, Ruzzi, and
  Savo}]{DBLP:journals/semweb/CalvaneseGLLPRRRS11}
\bibinfo{author}{D.~Calvanese}, \bibinfo{author}{G.~{De Giacomo}},
  \bibinfo{author}{D.~Lembo}, \bibinfo{author}{M.~Lenzerini},
  \bibinfo{author}{A.~Poggi}, \bibinfo{author}{M.~Rodriguez{-}Muro},
  \bibinfo{author}{R.~Rosati}, \bibinfo{author}{M.~Ruzzi},
  \bibinfo{author}{D.~F. Savo},
\newblock \bibinfo{title}{The {MASTRO} system for ontology-based data access},
\newblock \bibinfo{journal}{Semantic Web Journal} \bibinfo{volume}{2}
  (\bibinfo{year}{2011}) \bibinfo{pages}{43--53}.
\bibitem[{Calvanese et~al.(2016)Calvanese, Liuzzo, Mosca, Remesal, Rezk, and
  Rull}]{DBLP:journals/eaai/CalvaneseLMRRR16}
\bibinfo{author}{D.~Calvanese}, \bibinfo{author}{P.~Liuzzo},
  \bibinfo{author}{A.~Mosca}, \bibinfo{author}{J.~Remesal},
  \bibinfo{author}{M.~Rezk}, \bibinfo{author}{G.~Rull},
\newblock \bibinfo{title}{Ontology-based data integration in {EPN}et:
  Production and distribution of food during the {R}oman {E}mpire},
\newblock \bibinfo{journal}{Eng. Appl. of {AI}} \bibinfo{volume}{51}
  (\bibinfo{year}{2016}) \bibinfo{pages}{212--229}.
\bibitem[{Rodriguez{-}Muro et~al.(2013)Rodriguez{-}Muro, Kontchakov, and
  Zakharyaschev}]{DBLP:conf/semweb/Rodriguez-MuroKZ13}
\bibinfo{author}{M.~Rodriguez{-}Muro}, \bibinfo{author}{R.~Kontchakov},
  \bibinfo{author}{M.~Zakharyaschev},
\newblock \bibinfo{title}{Ontology-based data access: Ontop of databases},
\newblock in: \bibinfo{booktitle}{Proc.\ of the 12th Int.\ Semantic Web Conf.,
  ISWC'13, Part {I}}, volume \bibinfo{volume}{8218} of
  \textit{\bibinfo{series}{Lecture Notes in Computer Science}},
  \bibinfo{publisher}{Springer}, \bibinfo{year}{2013}, pp.
  \bibinfo{pages}{558--573}.
\bibitem[{Sequeda and Miranker(2017)}]{DBLP:journals/internet/SequedaM17}
\bibinfo{author}{J.~F. Sequeda}, \bibinfo{author}{D.~P. Miranker},
\newblock \bibinfo{title}{A pay-as-you-go methodology for ontology-based data
  access},
\newblock \bibinfo{journal}{{IEEE} Internet Computing} \bibinfo{volume}{21}
  (\bibinfo{year}{2017}) \bibinfo{pages}{92--96}.
\bibitem[{Hovland et~al.(2017)Hovland, Kontchakov, Skj{\ae}veland, Waaler, and
  Zakharyaschev}]{DBLP:conf/semweb/HovlandKSWZ17}
\bibinfo{author}{D.~Hovland}, \bibinfo{author}{R.~Kontchakov},
  \bibinfo{author}{M.~G. Skj{\ae}veland}, \bibinfo{author}{A.~Waaler},
  \bibinfo{author}{M.~Zakharyaschev},
\newblock \bibinfo{title}{Ontology-based data access to {S}legge},
\newblock in: \bibinfo{booktitle}{Proc.\ of the 16th Int.\ Semantic Web Conf.,
  ISWC 2017, Part {II}}, volume \bibinfo{volume}{10588} of
  \textit{\bibinfo{series}{Lecture Notes in Computer Science}},
  \bibinfo{publisher}{Springer}, \bibinfo{year}{2017}, pp.
  \bibinfo{pages}{120--129}.
\bibitem[{Xiao et~al.(2018)Xiao, Calvanese, Kontchakov, Lembo, Poggi, Rosati,
  and Zakharyaschev}]{IJCAI-18}
\bibinfo{author}{G.~Xiao}, \bibinfo{author}{D.~Calvanese},
  \bibinfo{author}{R.~Kontchakov}, \bibinfo{author}{D.~Lembo},
  \bibinfo{author}{A.~Poggi}, \bibinfo{author}{R.~Rosati},
  \bibinfo{author}{M.~Zakharyaschev},
\newblock \bibinfo{title}{Ontology-based data access: A survey},
\newblock in: \bibinfo{booktitle}{Proc.\ of the 28th Int.\ Joint Conf.\ on
  Artificial Intelligence (IJCAI)}, \bibinfo{publisher}{IJCAI/AAAI},
  \bibinfo{year}{2018}, pp. \bibinfo{pages}{5511--5519}.
\bibitem[{Gabbay et~al.(1994)Gabbay, Hodkinson, and Reynolds}]{Gabbayetal94}
\bibinfo{author}{D.~Gabbay}, \bibinfo{author}{I.~Hodkinson},
  \bibinfo{author}{M.~Reynolds}, \bibinfo{title}{Temporal Logic: Mathematical
  Foundations and Computational Aspects}, volume~\bibinfo{volume}{1},
  \bibinfo{publisher}{Oxford University Press}, \bibinfo{year}{1994}.
\bibitem[{Gabbay et~al.(2003)Gabbay, Kurucz, Wolter, and Zakharyaschev}]{gkwz}
\bibinfo{author}{D.~Gabbay}, \bibinfo{author}{A.~Kurucz},
  \bibinfo{author}{F.~Wolter}, \bibinfo{author}{M.~Zakharyaschev},
  \bibinfo{title}{Many-Dimensional Modal Logics: Theory and Applications},
  volume \bibinfo{volume}{148} of \textit{\bibinfo{series}{Studies in Logic}},
  \bibinfo{publisher}{Elsevier}, \bibinfo{year}{2003}.
\bibitem[{Demri et~al.(2016)Demri, Goranko, and
  Lange}]{DBLP:books/cu/Demri2016}
\bibinfo{author}{S.~Demri}, \bibinfo{author}{V.~Goranko},
  \bibinfo{author}{M.~Lange}, \bibinfo{title}{Temporal Logics in Computer
  Science}, Cambridge Tracts in Theoretical Computer Science,
  \bibinfo{publisher}{Cambridge University Press}, \bibinfo{year}{2016}.
\bibitem[{Schmiedel(1990)}]{DBLP:conf/aaai/Schmiedel90}
\bibinfo{author}{A.~Schmiedel},
\newblock \bibinfo{title}{Temporal terminological logic},
\newblock in: \bibinfo{booktitle}{Proc.\ of the 8th National Conf. on
  Artificial Intelligence, AAAI'90}, \bibinfo{publisher}{{AAAI} Press / The
  {MIT} Press}, \bibinfo{year}{1990}, pp. \bibinfo{pages}{640--645}.
\bibitem[{Schild(1993)}]{DBLP:conf/epia/Schild93}
\bibinfo{author}{K.~Schild},
\newblock \bibinfo{title}{Combining terminological logics with tense logic},
\newblock in: \bibinfo{booktitle}{Proc.\ of the 6th Portuguese Conf.\ on
  Progress in Artificial Intelligence, {EPIA'93}}, volume \bibinfo{volume}{727}
  of \textit{\bibinfo{series}{Lecture Notes in Computer Science}},
  \bibinfo{publisher}{Springer}, \bibinfo{year}{1993}, pp.
  \bibinfo{pages}{105--120}.
\bibitem[{Baader et~al.(2003)Baader, K\"{u}sters, and
  Wolter}]{Baader:2003:EDL:885746.885753}
\bibinfo{author}{F.~Baader}, \bibinfo{author}{R.~K\"{u}sters},
  \bibinfo{author}{F.~Wolter},
\newblock \bibinfo{title}{Extensions to description logics},
\newblock in: \bibinfo{booktitle}{The Description Logic Handbook},
  \bibinfo{publisher}{Cambridge University Press}, \bibinfo{year}{2003}, pp.
  \bibinfo{pages}{219--261}.
\bibitem[{Artale and Franconi(2005)}]{DBLP:reference/fai/ArtaleF05}
\bibinfo{author}{A.~Artale}, \bibinfo{author}{E.~Franconi},
\newblock \bibinfo{title}{Temporal description logics},
\newblock in: \bibinfo{booktitle}{Handbook of Temporal Reasoning in Artificial
  Intelligence}, volume~\bibinfo{volume}{1} of
  \textit{\bibinfo{series}{Foundations of Artificial Intelligence}},
  \bibinfo{publisher}{Elsevier}, \bibinfo{year}{2005}, pp.
  \bibinfo{pages}{375--388}.
\bibitem[{Lutz et~al.(2008)Lutz, Wolter, and
  Zakharyaschev}]{DBLP:conf/time/LutzWZ08}
\bibinfo{author}{C.~Lutz}, \bibinfo{author}{F.~Wolter},
  \bibinfo{author}{M.~Zakharyaschev},
\newblock \bibinfo{title}{Temporal description logics: {A} survey},
\newblock in: \bibinfo{booktitle}{Proc.\ of the 15th Int. Symposium on Temporal
  Representation and Reasoning, TIME'08}, \bibinfo{publisher}{{IEEE} Computer
  Society}, \bibinfo{year}{2008}, pp. \bibinfo{pages}{3--14}.
\bibitem[{Pagliarecci et~al.(2013)Pagliarecci, Spalazzi, and
  Taccari}]{DBLP:conf/dlog/PagliarecciST13}
\bibinfo{author}{F.~Pagliarecci}, \bibinfo{author}{L.~Spalazzi},
  \bibinfo{author}{G.~Taccari},
\newblock \bibinfo{title}{Reasoning with temporal {ABoxes}: Combining
  $\textsl{DL-Lite}_{\textit{core}}$ with {CTL}},
\newblock in: \bibinfo{booktitle}{Proc.\ of the 26th Int.\ Workshop on
  Description Logics, {DL'13}}, \bibinfo{publisher}{CEUR-WS},
  \bibinfo{year}{2013}, pp. \bibinfo{pages}{885--897}.
\bibitem[{Artale et~al.(2014)Artale, Kontchakov, Ryzhikov, and
  Zakharyaschev}]{DBLP:journals/tocl/ArtaleKRZ14}
\bibinfo{author}{A.~Artale}, \bibinfo{author}{R.~Kontchakov},
  \bibinfo{author}{V.~Ryzhikov}, \bibinfo{author}{M.~Zakharyaschev},
\newblock \bibinfo{title}{A cookbook for temporal conceptual data modelling
  with description logics},
\newblock \bibinfo{journal}{{ACM} Trans. Comput. Log.} \bibinfo{volume}{15}
  (\bibinfo{year}{2014}) \bibinfo{pages}{25:1--25:50}.
\bibitem[{Guti{\'{e}}rrez{-}Basulto et~al.(2014)Guti{\'{e}}rrez{-}Basulto,
  Jung, and Schneider}]{DBLP:conf/kr/Gutierrez-BasultoJ014}
\bibinfo{author}{V.~Guti{\'{e}}rrez{-}Basulto}, \bibinfo{author}{J.~C. Jung},
  \bibinfo{author}{T.~Schneider},
\newblock \bibinfo{title}{Lightweight description logics and branching time: A
  troublesome marriage},
\newblock in: \bibinfo{booktitle}{Proc. of the 14th Int.\ Conf.\ on Principles
  of Knowledge Representation and Reasoning, {KR}'14},
  \bibinfo{publisher}{{AAAI} Press}, \bibinfo{year}{2014}, pp.
  \bibinfo{pages}{278--287}.
\bibitem[{Guti{\'{e}}rrez{-}Basulto et~al.(2015)Guti{\'{e}}rrez{-}Basulto,
  Jung, and Schneider}]{DBLP:conf/ijcai/Gutierrez-Basulto15}
\bibinfo{author}{V.~Guti{\'{e}}rrez{-}Basulto}, \bibinfo{author}{J.~C. Jung},
  \bibinfo{author}{T.~Schneider},
\newblock \bibinfo{title}{Lightweight temporal description logics with rigid
  roles and restricted {TBoxes}},
\newblock in: \bibinfo{booktitle}{Proc.\ of the 24th Int.\ Joint Conf.\ on
  Artificial Intelligence, {IJCAI} 2015}, \bibinfo{publisher}{IJCAI/AAAI},
  \bibinfo{year}{2015}, pp. \bibinfo{pages}{3015--3021}.
\bibitem[{Guti{\'{e}}rrez{-}Basulto et~al.(2016)Guti{\'{e}}rrez{-}Basulto,
  Jung, and Ozaki}]{DBLP:conf/ecai/Gutierrez-Basulto16}
\bibinfo{author}{V.~Guti{\'{e}}rrez{-}Basulto}, \bibinfo{author}{J.~C. Jung},
  \bibinfo{author}{A.~Ozaki},
\newblock \bibinfo{title}{On metric temporal description logics},
\newblock in: \bibinfo{booktitle}{Proc.\ of the 22nd European Conf.\ on
  Artificial Intelligence, ECAI 2016}, volume \bibinfo{volume}{285} of
  \textit{\bibinfo{series}{FAIA}}, \bibinfo{publisher}{IOS Press},
  \bibinfo{year}{2016}, pp. \bibinfo{pages}{837--845}.
\bibitem[{Baader et~al.(2017)Baader, Borgwardt, Koopmann, Ozaki, and
  Thost}]{DBLP:conf/dlog/BaaderBKOT17}
\bibinfo{author}{F.~Baader}, \bibinfo{author}{S.~Borgwardt},
  \bibinfo{author}{P.~Koopmann}, \bibinfo{author}{A.~Ozaki},
  \bibinfo{author}{V.~Thost},
\newblock \bibinfo{title}{Metric temporal description logics with
  interval-rigid names (extended abstract)},
\newblock in: \bibinfo{booktitle}{Proc.\ of the 30th Int.\ Workshop on
  Description Logics, DL'17}, volume \bibinfo{volume}{1879},
  \bibinfo{publisher}{CEUR-WS}, \bibinfo{year}{2017}.
\bibitem[{Manna and Pnueli(1992)}]{DBLP:books/daglib/0077033}
\bibinfo{author}{Z.~Manna}, \bibinfo{author}{A.~Pnueli}, \bibinfo{title}{The
  temporal logic of reactive and concurrent systems - specification},
  \bibinfo{publisher}{Springer}, \bibinfo{year}{1992}.
\bibitem[{Kamp(1968)}]{phd-kamp}
\bibinfo{author}{H.~W. Kamp}, \bibinfo{title}{Tense Logic and the Theory of
  Linear Order}, \bibinfo{type}{{PhD} thesis}, Computer Science Department,
  University of California at Los~Angeles, USA, \bibinfo{year}{1968}.
\bibitem[{Rabinovich(2014)}]{DBLP:journals/corr/Rabinovich14}
\bibinfo{author}{A.~Rabinovich},
\newblock \bibinfo{title}{A proof of {K}amp's theorem},
\newblock \bibinfo{journal}{Logical Methods in Computer Science}
  \bibinfo{volume}{10} (\bibinfo{year}{2014}).
\bibitem[{Prior(1956)}]{prior:1956b}
\bibinfo{author}{A.~Prior}, \bibinfo{title}{Time and Modality},
  \bibinfo{publisher}{Oxford University Press}, \bibinfo{year}{1956}.
\bibitem[{Ono and Nakamura(1980)}]{ono1980on}
\bibinfo{author}{H.~Ono}, \bibinfo{author}{A.~Nakamura},
\newblock \bibinfo{title}{On the size of refutation {K}ripke models for some
  linear modal and tense logics},
\newblock \bibinfo{journal}{Studia Logica}  (\bibinfo{year}{1980})
  \bibinfo{pages}{325--333}.
\bibitem[{Burgess(1984)}]{Burgess84}
\bibinfo{author}{J.~P. Burgess},
\newblock \bibinfo{title}{Basic tense logic},
\newblock in: \bibinfo{booktitle}{Handbook of Philosophical Logic: Volume II:
  Extensions of Classical Logic}, \bibinfo{publisher}{Reidel},
  \bibinfo{address}{Dordrecht}, \bibinfo{year}{1984}, pp.
  \bibinfo{pages}{89--133}.
\bibitem[{Vardi(2008)}]{DBLP:conf/spin/Vardi08}
\bibinfo{author}{M.~Y. Vardi},
\newblock \bibinfo{title}{From {C}hurch and {P}rior to {PSL}},
\newblock in: \bibinfo{booktitle}{25 Years of Model Checking - History,
  Achievements, Perspectives}, volume \bibinfo{volume}{5000} of
  \textit{\bibinfo{series}{Lecture Notes in Computer Science}},
  \bibinfo{publisher}{Springer}, \bibinfo{year}{2008}, pp.
  \bibinfo{pages}{150--171}.
\bibitem[{Straubing(1994)}]{Straubing94}
\bibinfo{author}{H.~Straubing}, \bibinfo{title}{Finite Automata, Formal Logic,
  and Circuit Complexity}, \bibinfo{publisher}{Birkhauser Verlag},
  \bibinfo{year}{1994}.
\bibitem[{Libkin(2004)}]{Libkin}
\bibinfo{author}{L.~Libkin}, \bibinfo{title}{Elements Of Finite Model Theory},
  \bibinfo{publisher}{Springer}, \bibinfo{year}{2004}.
\bibitem[{Furst et~al.(1984)Furst, Saxe, and
  Sipser}]{DBLP:journals/mst/FurstSS84}
\bibinfo{author}{M.~L. Furst}, \bibinfo{author}{J.~B. Saxe},
  \bibinfo{author}{M.~Sipser},
\newblock \bibinfo{title}{Parity, circuits, and the polynomial-time hierarchy},
\newblock \bibinfo{journal}{Mathematical Systems Theory} \bibinfo{volume}{17}
  (\bibinfo{year}{1984}) \bibinfo{pages}{13--27}.
\bibitem[{Fisher(1997)}]{DBLP:journals/logcom/Fisher97}
\bibinfo{author}{M.~Fisher},
\newblock \bibinfo{title}{A normal form for temporal logics and its
  applications in theorem-proving and execution},
\newblock \bibinfo{journal}{J. Log. Comput.} \bibinfo{volume}{7}
  (\bibinfo{year}{1997}) \bibinfo{pages}{429--456}.
\bibitem[{Compton and Laflamme(1990)}]{DBLP:journals/iandc/ComptonL90}
\bibinfo{author}{K.~J. Compton}, \bibinfo{author}{C.~Laflamme},
\newblock \bibinfo{title}{An algebra and a logic for {NC}{\({^1}\)}},
\newblock \bibinfo{journal}{Inf. Comput.} \bibinfo{volume}{87}
  (\bibinfo{year}{1990}) \bibinfo{pages}{240--262}.
\bibitem[{Arora and Barak(2009)}]{Arora&Barak09}
\bibinfo{author}{S.~Arora}, \bibinfo{author}{B.~Barak},
  \bibinfo{title}{Computational Complexity: A Modern Approach},
  \bibinfo{edition}{1st} ed., \bibinfo{publisher}{Cambridge University Press},
  \bibinfo{address}{New York, USA}, \bibinfo{year}{2009}.
\bibitem[{Immerman(1999)}]{Immerman99}
\bibinfo{author}{N.~Immerman}, \bibinfo{title}{Descriptive Complexity},
  \bibinfo{publisher}{Springer}, \bibinfo{year}{1999}.
\bibitem[{Calvanese et~al.(2007)Calvanese, {De Giacomo}, Lembo, Lenzerini, and
  Rosati}]{CalvaneseGLLR07}
\bibinfo{author}{D.~Calvanese}, \bibinfo{author}{G.~{De Giacomo}},
  \bibinfo{author}{D.~Lembo}, \bibinfo{author}{M.~Lenzerini},
  \bibinfo{author}{R.~Rosati},
\newblock \bibinfo{title}{{EQL-Lite}: Effective first-order query processing in
  description logics},
\newblock in: \bibinfo{booktitle}{Proc.\ of the 20th Int.\ Joint Conf.\ on
  Artificial Intelligence (IJCAI 2007)}, \bibinfo{year}{2007}, pp.
  \bibinfo{pages}{274--279}.
\bibitem[{Glimm and Ogbuji(2013)}]{GlimmOgbuji13}
\bibinfo{author}{B.~Glimm}, \bibinfo{author}{C.~Ogbuji},
  \bibinfo{title}{{SPARQL 1.1} entailment regimes}, \bibinfo{howpublished}{W3C
  Recommendation}, \bibinfo{year}{2013}. \URLprefix
  \url{http://www.w3.org/TR/sparql11-entailment}.
\bibitem[{Schwentick et~al.(2002)Schwentick, Th{\'e}rien, and
  Vollmer}]{DBLP:conf/dlt/SchwentickTV01}
\bibinfo{author}{T.~Schwentick}, \bibinfo{author}{D.~Th{\'e}rien},
  \bibinfo{author}{H.~Vollmer},
\newblock \bibinfo{title}{Partially-ordered two-way automata: A new
  characterization of {DA}},
\newblock in: \bibinfo{booktitle}{Revised Papers of the 5th Int.\ Conf.\ in
  Developments in Language Theory, DLT 2001}, volume \bibinfo{volume}{2295} of
  \textit{\bibinfo{series}{Lecture Notes in Computer Science}},
  \bibinfo{publisher}{Springer}, \bibinfo{year}{2002}, pp.
  \bibinfo{pages}{239--250}.
\bibitem[{Chrobak(1986)}]{chrobak-ufa}
\bibinfo{author}{M.~Chrobak},
\newblock \bibinfo{title}{Finite automata and unary languages},
\newblock \bibinfo{journal}{Theor. Comp. Sci.} \bibinfo{volume}{47}
  (\bibinfo{year}{1986}) \bibinfo{pages}{149--158}.
\bibitem[{Cal\`{\i} et~al.(2012)Cal\`{\i}, Gottlob, and Lukasiewicz}]{CaliGL12}
\bibinfo{author}{A.~Cal\`{\i}}, \bibinfo{author}{G.~Gottlob},
  \bibinfo{author}{T.~Lukasiewicz},
\newblock \bibinfo{title}{A general datalog-based framework for tractable query
  answering over ontologies},
\newblock \bibinfo{journal}{J. Web Semant.} \bibinfo{volume}{14}
  (\bibinfo{year}{2012}) \bibinfo{pages}{57--83}.
\bibitem[{Baget et~al.(2011)Baget, Lecl{\`{e}}re, Mugnier, and
  Salvat}]{DBLP:journals/ai/BagetLMS11}
\bibinfo{author}{J.~Baget}, \bibinfo{author}{M.~Lecl{\`{e}}re},
  \bibinfo{author}{M.~Mugnier}, \bibinfo{author}{E.~Salvat},
\newblock \bibinfo{title}{On rules with existential variables: Walking the
  decidability line},
\newblock \bibinfo{journal}{Artif. Intell.} \bibinfo{volume}{175}
  (\bibinfo{year}{2011}) \bibinfo{pages}{1620--1654}.
\bibitem[{Bienvenu and Ortiz(2015)}]{DBLP:conf/rweb/BienvenuO15}
\bibinfo{author}{M.~Bienvenu}, \bibinfo{author}{M.~Ortiz},
\newblock \bibinfo{title}{Ontology-mediated query answering with data-tractable
  description logics},
\newblock in: \bibinfo{booktitle}{Web Logic Rules: Tutorial Lectures at the
  11th Int. Summer School on Reasoning Web, RW 2015}, volume
  \bibinfo{volume}{9203} of \textit{\bibinfo{series}{Lecture Notes in Computer
  Science}}, \bibinfo{publisher}{Springer}, \bibinfo{year}{2015}, pp.
  \bibinfo{pages}{218--307}.
\bibitem[{Bienvenu et~al.(2014)Bienvenu, ten Cate, Lutz, and
  Wolter}]{DBLP:journals/tods/BienvenuCLW14}
\bibinfo{author}{M.~Bienvenu}, \bibinfo{author}{B.~ten Cate},
  \bibinfo{author}{C.~Lutz}, \bibinfo{author}{F.~Wolter},
\newblock \bibinfo{title}{Ontology-based data access: {A} study through
  disjunctive datalog, {CSP}, and {MMSNP}},
\newblock \bibinfo{journal}{{ACM} Trans.\ on Database Systems}
  \bibinfo{volume}{39} (\bibinfo{year}{2014}) \bibinfo{pages}{33:1--33:44}.
\bibitem[{Bienvenu et~al.(2017{\natexlab{a}})Bienvenu, Kikot, Kontchakov,
  Ryzhikov, and Zakharyaschev}]{DBLP:conf/dlog/BienvenuKKRZ17}
\bibinfo{author}{M.~Bienvenu}, \bibinfo{author}{S.~Kikot},
  \bibinfo{author}{R.~Kontchakov}, \bibinfo{author}{V.~Ryzhikov},
  \bibinfo{author}{M.~Zakharyaschev},
\newblock \bibinfo{title}{On the parametrised complexity of tree-shaped
  ontology-mediated queries in {OWL} 2 {QL}},
\newblock in: \bibinfo{booktitle}{Proc.\ of the 30th Int.\ Workshop on
  Description Logics, DL 2017}, volume \bibinfo{volume}{1879},
  \bibinfo{publisher}{CEUR-WS}, \bibinfo{year}{2017}{\natexlab{a}}.
\bibitem[{Bienvenu et~al.(2017{\natexlab{b}})Bienvenu, Kikot, Kontchakov,
  Podolskii, Ryzhikov, and Zakharyaschev}]{DBLP:conf/pods/BienvenuKKPRZ17}
\bibinfo{author}{M.~Bienvenu}, \bibinfo{author}{S.~Kikot},
  \bibinfo{author}{R.~Kontchakov}, \bibinfo{author}{V.~V. Podolskii},
  \bibinfo{author}{V.~Ryzhikov}, \bibinfo{author}{M.~Zakharyaschev},
\newblock \bibinfo{title}{The complexity of ontology-based data access with
  {OWL} 2 {QL} and bounded treewidth queries},
\newblock in: \bibinfo{booktitle}{Proc.\ of the 36th {ACM}
  {SIGMOD-SIGACT-SIGAI} Symposium on Principles of Database Systems, PODS
  2017}, \bibinfo{publisher}{ACM}, \bibinfo{year}{2017}{\natexlab{b}}, pp.
  \bibinfo{pages}{201--216}.
\bibitem[{Bienvenu et~al.(2018)Bienvenu, Kikot, Kontchakov, Podolskii, and
  Zakharyaschev}]{DBLP:journals/jacm/BienvenuKKPZ18}
\bibinfo{author}{M.~Bienvenu}, \bibinfo{author}{S.~Kikot},
  \bibinfo{author}{R.~Kontchakov}, \bibinfo{author}{V.~V. Podolskii},
  \bibinfo{author}{M.~Zakharyaschev},
\newblock \bibinfo{title}{Ontology-mediated queries: Combined complexity and
  succinctness of rewritings via circuit complexity},
\newblock \bibinfo{journal}{J. {ACM}} \bibinfo{volume}{65}
  (\bibinfo{year}{2018}) \bibinfo{pages}{28:1--28:51}.
\bibitem[{Lutz and Sabellek(2017)}]{DBLP:conf/ijcai/LutzS17}
\bibinfo{author}{C.~Lutz}, \bibinfo{author}{L.~Sabellek},
\newblock \bibinfo{title}{Ontology-mediated querying with the description logic
  {EL:} trichotomy and linear datalog rewritability},
\newblock in: \bibinfo{booktitle}{Proc.\ of the 26th Int.\ Joint Conf.\ on
  Artificial Intelligence (IJCAI 2017)}, \bibinfo{year}{2017}, pp.
  \bibinfo{pages}{1181--1187}.
\bibitem[{Hernich et~al.(2017)Hernich, Lutz, Papacchini, and
  Wolter}]{DBLP:conf/pods/HernichLPW17}
\bibinfo{author}{A.~Hernich}, \bibinfo{author}{C.~Lutz},
  \bibinfo{author}{F.~Papacchini}, \bibinfo{author}{F.~Wolter},
\newblock \bibinfo{title}{Dichotomies in ontology-mediated querying with the
  guarded fragment},
\newblock in: \bibinfo{booktitle}{Proc.\ of the 36th {ACM}
  {SIGMOD-SIGACT-SIGAI} Symposium on Principles of Database Systems, PODS
  2017}, \bibinfo{publisher}{ACM}, \bibinfo{year}{2017}, pp.
  \bibinfo{pages}{185--199}.
\bibitem[{Rosati(2007)}]{DBLP:conf/icdt/Rosati07}
\bibinfo{author}{R.~Rosati},
\newblock \bibinfo{title}{The limits of querying ontologies},
\newblock in: \bibinfo{booktitle}{Proc.\ of the 11th Int.\ Conf.\ Database
  Theory, ICDT 2007}, volume \bibinfo{volume}{4353} of
  \textit{\bibinfo{series}{Lecture Notes in Computer Science}},
  \bibinfo{publisher}{Springer}, \bibinfo{year}{2007}, pp.
  \bibinfo{pages}{164--178}.
\bibitem[{Guti{\'{e}}rrez{-}Basulto et~al.(2015)Guti{\'{e}}rrez{-}Basulto,
  Ib{\'{a}}{\~{n}}ez{-}Garc{\'{\i}}a, Kontchakov, and
  Kostylev}]{DBLP:journals/ws/Gutierrez-Basulto15}
\bibinfo{author}{V.~Guti{\'{e}}rrez{-}Basulto}, \bibinfo{author}{Y.~A.
  Ib{\'{a}}{\~{n}}ez{-}Garc{\'{\i}}a}, \bibinfo{author}{R.~Kontchakov},
  \bibinfo{author}{E.~V. Kostylev},
\newblock \bibinfo{title}{Queries with negation and inequalities over
  lightweight ontologies},
\newblock \bibinfo{journal}{J. Web Semant.} \bibinfo{volume}{35}
  (\bibinfo{year}{2015}) \bibinfo{pages}{184--202}.
\bibitem[{B\"uchi(1960)}]{Buchi60}
\bibinfo{author}{J.~B\"uchi},
\newblock \bibinfo{title}{Weak second-order arithmetic and finite automata},
\newblock \bibinfo{journal}{Zeitschrift f\"ur Mathematische Logik und
  Grundlagen der Mathematik} \bibinfo{volume}{6} (\bibinfo{year}{1960})
  \bibinfo{pages}{66--92}.
\bibitem[{Elgot(1961)}]{Elgot61}
\bibinfo{author}{C.~Elgot},
\newblock \bibinfo{title}{Decision problems of finite automata design and
  related arithmetics},
\newblock \bibinfo{journal}{Transactions of the American Mathematical Society}
  \bibinfo{volume}{98} (\bibinfo{year}{1961}) \bibinfo{pages}{21--51}.
\bibitem[{Trakhtenbrot(1962)}]{Trakh62}
\bibinfo{author}{B.~Trakhtenbrot},
\newblock \bibinfo{title}{Finite automata and the logic of one-place
  predicates},
\newblock \bibinfo{journal}{Siberian Mathematical Journal} \bibinfo{volume}{3}
  (\bibinfo{year}{1962}) \bibinfo{pages}{103--131}. \bibinfo{note}{English
  translation in: AMS Transl. 59 (1966) 23--55}.
\bibitem[{Barrington(1989)}]{DBLP:journals/jcss/Barrington89}
\bibinfo{author}{D.~A.~M. Barrington},
\newblock \bibinfo{title}{Bounded-width polynomial-size branching programs
  recognize exactly those languages in {NC\({^1}\)}},
\newblock \bibinfo{journal}{J. Comput. Syst. Sci.} \bibinfo{volume}{38}
  (\bibinfo{year}{1989}) \bibinfo{pages}{150--164}.
\bibitem[{To(2009)}]{to-ufa}
\bibinfo{author}{A.~W. To},
\newblock \bibinfo{title}{Unary finite automata vs. arithmetic progressions},
\newblock \bibinfo{journal}{Inf. Process. Lett.} \bibinfo{volume}{109}
  (\bibinfo{year}{2009}) \bibinfo{pages}{1010--1014}.
\bibitem[{Baader et~al.(2013)Baader, Borgwardt, and
  Lippmann}]{DBLP:conf/cade/BaaderBL13}
\bibinfo{author}{F.~Baader}, \bibinfo{author}{S.~Borgwardt},
  \bibinfo{author}{M.~Lippmann},
\newblock \bibinfo{title}{Temporalizing ontology-based data access},
\newblock in: \bibinfo{booktitle}{Proc.\ of the 24th Int.\ Conf.\ on Automated
  Deduction, {CADE-24}}, volume \bibinfo{volume}{7898} of
  \textit{\bibinfo{series}{Lecture Notes in Computer Science}},
  \bibinfo{publisher}{Springer}, \bibinfo{year}{2013}, pp.
  \bibinfo{pages}{330--344}.
\bibitem[{Baader et~al.(2015)Baader, Borgwardt, and
  Lippmann}]{DBLP:journals/ws/BaaderBL15}
\bibinfo{author}{F.~Baader}, \bibinfo{author}{S.~Borgwardt},
  \bibinfo{author}{M.~Lippmann},
\newblock \bibinfo{title}{Temporal query entailment in the description logic
  {SHQ}},
\newblock \bibinfo{journal}{J.\ Web Semant.} \bibinfo{volume}{33}
  (\bibinfo{year}{2015}) \bibinfo{pages}{71--93}.
\bibitem[{Borgwardt and
  Thost(2015{\natexlab{a}})}]{DBLP:conf/gcai/BorgwardtT15}
\bibinfo{author}{S.~Borgwardt}, \bibinfo{author}{V.~Thost},
\newblock \bibinfo{title}{Temporal query answering in \textit{DL-Lite} with
  negation},
\newblock in: \bibinfo{booktitle}{Proc.\ of the Global Conf.\ on Artificial
  Intelligence, {GCAI15}}, volume~\bibinfo{volume}{36} of
  \textit{\bibinfo{series}{EPiC Series in Computing}},
  \bibinfo{year}{2015}{\natexlab{a}}, pp. \bibinfo{pages}{51--65}.
\bibitem[{Borgwardt and
  Thost(2015{\natexlab{b}})}]{DBLP:conf/ijcai/BorgwardtT15}
\bibinfo{author}{S.~Borgwardt}, \bibinfo{author}{V.~Thost},
\newblock \bibinfo{title}{Temporal query answering in the description logic
  $\mathcal{EL}$},
\newblock in: \bibinfo{booktitle}{Proc.\ of the 24h Int.\ Joint Conf.\ on
  Artificial Intelligence, IJCAI'15}, \bibinfo{publisher}{{AAAI} Press},
  \bibinfo{year}{2015}{\natexlab{b}}, pp. \bibinfo{pages}{2819--2825}.
\bibitem[{Baader et~al.(2015)Baader, Borgwardt, and
  Lippmann}]{DBLP:conf/ausai/BaaderBL15}
\bibinfo{author}{F.~Baader}, \bibinfo{author}{S.~Borgwardt},
  \bibinfo{author}{M.~Lippmann},
\newblock \bibinfo{title}{Temporal conjunctive queries in expressive
  description logics with transitive roles},
\newblock in: \bibinfo{booktitle}{Proc.\ of the 28th Australasian Joint Conf.\
  on Advances in Artificial Intelligence, AI'15}, volume \bibinfo{volume}{9457}
  of \textit{\bibinfo{series}{Lecture Notes in Computer Science}},
  \bibinfo{publisher}{Springer}, \bibinfo{year}{2015}, pp.
  \bibinfo{pages}{21--33}.
\bibitem[{Borgwardt et~al.(2013)Borgwardt, Lippmann, and
  Thost}]{DBLP:conf/frocos/BorgwardtLT13}
\bibinfo{author}{S.~Borgwardt}, \bibinfo{author}{M.~Lippmann},
  \bibinfo{author}{V.~Thost},
\newblock \bibinfo{title}{Temporal query answering in the description logic
  {DL-Lite}},
\newblock in: \bibinfo{booktitle}{Proc.\ of the 9th Int.\ Symposium on
  Frontiers of Combining Systems, FroCoS'13}, volume \bibinfo{volume}{8152} of
  \textit{\bibinfo{series}{Lecture Notes in Computer Science}},
  \bibinfo{publisher}{Springer}, \bibinfo{year}{2013}, pp.
  \bibinfo{pages}{165--180}.
\bibitem[{Borgwardt et~al.(2015)Borgwardt, Lippmann, and
  Thost}]{DBLP:journals/ws/BorgwardtLT15}
\bibinfo{author}{S.~Borgwardt}, \bibinfo{author}{M.~Lippmann},
  \bibinfo{author}{V.~Thost},
\newblock \bibinfo{title}{Temporalizing rewritable query languages over
  knowledge bases},
\newblock \bibinfo{journal}{J.\ Web Semant.} \bibinfo{volume}{33}
  (\bibinfo{year}{2015}) \bibinfo{pages}{50--70}.
\bibitem[{Bourgaux et~al.(2019)Bourgaux, Koopmann, and
  Turhan}]{DBLP:journals/semweb/BourgauxKT19}
\bibinfo{author}{C.~Bourgaux}, \bibinfo{author}{P.~Koopmann},
  \bibinfo{author}{A.~Turhan},
\newblock \bibinfo{title}{Ontology-mediated query answering over temporal and
  inconsistent data},
\newblock \bibinfo{journal}{Semantic Web} \bibinfo{volume}{10}
  (\bibinfo{year}{2019}) \bibinfo{pages}{475--521}.
\bibitem[{Koopmann(2019)}]{DBLP:conf/aaai/Koopmann19}
\bibinfo{author}{P.~Koopmann},
\newblock \bibinfo{title}{Ontology-based query answering for probabilistic
  temporal data},
\newblock in: \bibinfo{booktitle}{Proc.\ of the 33rd {AAAI} Conference on
  Artificial Intelligence, AAAI 2019}, \bibinfo{year}{2019}, pp.
  \bibinfo{pages}{2903--2910}.
\bibitem[{Guti{\'{e}}rrez{-}Basulto et~al.(2016)Guti{\'{e}}rrez{-}Basulto,
  Jung, and Kontchakov}]{DBLP:conf/ijcai/Gutierrez-Basulto16}
\bibinfo{author}{V.~Guti{\'{e}}rrez{-}Basulto}, \bibinfo{author}{J.~C. Jung},
  \bibinfo{author}{R.~Kontchakov},
\newblock \bibinfo{title}{Temporalized {$\mathcal{EL}$} ontologies for
  accessing temporal data: Complexity of atomic queries},
\newblock in: \bibinfo{booktitle}{Proc.\ of the 25th Int. Joint Conf.\ on
  Artificial Intelligence, {IJCAI'16}}, \bibinfo{publisher}{IJCAI/AAAI},
  \bibinfo{year}{2016}, pp. \bibinfo{pages}{1102--1108}.
\bibitem[{Borgwardt et~al.(2019)Borgwardt, Forkel, and
  Kovtunova}]{DBLP:conf/ruleml/BorgwardtFK19}
\bibinfo{author}{S.~Borgwardt}, \bibinfo{author}{W.~Forkel},
  \bibinfo{author}{A.~Kovtunova},
\newblock \bibinfo{title}{Finding new diamonds: Temporal minimal-world query
  answering over sparse aboxes},
\newblock in: \bibinfo{booktitle}{Proc.\ of the 3rd Int.\ Joint Conf.\ on Rules
  and Reasoning, RuleML+RR 2019}, volume \bibinfo{volume}{11784} of
  \textit{\bibinfo{series}{Lecture Notes in Computer Science}},
  \bibinfo{publisher}{Springer}, \bibinfo{year}{2019}, pp.
  \bibinfo{pages}{3--18}.
\bibitem[{Revesz(1993)}]{DBLP:journals/tcs/Revesz93}
\bibinfo{author}{P.~Z. Revesz},
\newblock \bibinfo{title}{A closed-form evaluation for datalog queries with
  integer (gap)-order constraints},
\newblock \bibinfo{journal}{Theor. Comput. Sci.} \bibinfo{volume}{116}
  (\bibinfo{year}{1993}) \bibinfo{pages}{117--149}.
\bibitem[{Kanellakis et~al.(1995)Kanellakis, Kuper, and
  Revesz}]{DBLP:journals/jcss/KanellakisKR95}
\bibinfo{author}{P.~C. Kanellakis}, \bibinfo{author}{G.~M. Kuper},
  \bibinfo{author}{P.~Z. Revesz},
\newblock \bibinfo{title}{Constraint query languages},
\newblock \bibinfo{journal}{J. Comput. Syst. Sci.} \bibinfo{volume}{51}
  (\bibinfo{year}{1995}) \bibinfo{pages}{26--52}.
\bibitem[{Toman and Chomicki(1998)}]{DBLP:journals/jlp/TomanC98}
\bibinfo{author}{D.~Toman}, \bibinfo{author}{J.~Chomicki},
\newblock \bibinfo{title}{Datalog with integer periodicity constraints},
\newblock \bibinfo{journal}{J. Log. Program.} \bibinfo{volume}{35}
  (\bibinfo{year}{1998}) \bibinfo{pages}{263--290}.
\bibitem[{Wolper(1983)}]{DBLP:journals/iandc/Wolper83}
\bibinfo{author}{P.~Wolper},
\newblock \bibinfo{title}{Temporal logic can be more expressive},
\newblock \bibinfo{journal}{Inf. Control.} \bibinfo{volume}{56}
  (\bibinfo{year}{1983}) \bibinfo{pages}{72--99}.
\bibitem[{Banieqbal and Barringer(1987)}]{DBLP:conf/tls/BanieqbalB87}
\bibinfo{author}{B.~Banieqbal}, \bibinfo{author}{H.~Barringer},
\newblock \bibinfo{title}{Temporal logic with fixed points},
\newblock in: \bibinfo{booktitle}{Proc.\ of Temporal Logic in Specification,
  1987}, volume \bibinfo{volume}{398} of \textit{\bibinfo{series}{Lecture Notes
  in Computer Science}}, \bibinfo{publisher}{Springer}, \bibinfo{year}{1987},
  pp. \bibinfo{pages}{62--74}.
\bibitem[{Vardi(1988)}]{DBLP:conf/popl/Vardi88}
\bibinfo{author}{M.~Y. Vardi},
\newblock \bibinfo{title}{A temporal fixpoint calculus},
\newblock in: \bibinfo{booktitle}{Record of the Fifteenth Annual {ACM}
  Symposium on Principles of Programming Languages, 1988},
  \bibinfo{publisher}{{ACM} Press}, \bibinfo{year}{1988}, pp.
  \bibinfo{pages}{250--259}.
\bibitem[{Artale et~al.(2013)Artale, Kontchakov, Ryzhikov, and
  Zakharyaschev}]{AKRZ:LPAR13}
\bibinfo{author}{A.~Artale}, \bibinfo{author}{R.~Kontchakov},
  \bibinfo{author}{V.~Ryzhikov}, \bibinfo{author}{M.~Zakharyaschev},
\newblock \bibinfo{title}{The complexity of clausal fragments of {LTL}},
\newblock in: \bibinfo{booktitle}{Proc.\ of the 19th Int.\ Conf.\ on Logic for
  Programming, Artificial Intelligence and Reasoning, LPAR 2013}, volume
  \bibinfo{volume}{8312} of \textit{\bibinfo{series}{Lecture Notes in Computer
  Science}}, \bibinfo{publisher}{Springer}, \bibinfo{year}{2013}, pp.
  \bibinfo{pages}{35--52}.
\bibitem[{Sistla and Clarke(1985)}]{DBLP:journals/jacm/SistlaC85}
\bibinfo{author}{A.~P. Sistla}, \bibinfo{author}{E.~M. Clarke},
\newblock \bibinfo{title}{The complexity of propositional linear temporal
  logics},
\newblock \bibinfo{journal}{J. {ACM}} \bibinfo{volume}{32}
  (\bibinfo{year}{1985}) \bibinfo{pages}{733--749}.
\bibitem[{Chen and Lin(1994)}]{Chen199495}
\bibinfo{author}{C.-C. Chen}, \bibinfo{author}{I.-P. Lin},
\newblock \bibinfo{title}{The computational complexity of the satisfiability of
  modal {H}orn clauses for modal propositional logics},
\newblock \bibinfo{journal}{Theor. Comp. Sci.} \bibinfo{volume}{129}
  (\bibinfo{year}{1994}) \bibinfo{pages}{95--121}.
\bibitem[{Fisher et~al.(2001)Fisher, Dixon, and Peim}]{FisherDP01}
\bibinfo{author}{M.~Fisher}, \bibinfo{author}{C.~Dixon},
  \bibinfo{author}{M.~Peim},
\newblock \bibinfo{title}{Clausal temporal resolution},
\newblock \bibinfo{journal}{{ACM} Trans.\ Comput. Logic} \bibinfo{volume}{2}
  (\bibinfo{year}{2001}) \bibinfo{pages}{12--56}.
\bibitem[{Straubing and Weil(2010)}]{abs-1011-6491}
\bibinfo{author}{H.~Straubing}, \bibinfo{author}{P.~Weil},
\newblock \bibinfo{title}{An introduction to finite automata and their
  connection to logic},
\newblock \bibinfo{journal}{CoRR} \bibinfo{volume}{abs/1011.6491}
  (\bibinfo{year}{2010}).
\bibitem[{Compton and Straubing(2001)}]{DBLP:books/ws/phaunRS01/ComptonS01}
\bibinfo{author}{K.~J. Compton}, \bibinfo{author}{H.~Straubing},
\newblock \bibinfo{title}{Characterizations of regular languages in low level
  complexity classes},
\newblock in: \bibinfo{editor}{G.~Paun}, \bibinfo{editor}{G.~Rozenberg},
  \bibinfo{editor}{A.~Salomaa} (Eds.), \bibinfo{booktitle}{Current Trends in
  Theoretical Computer Science, Entering the 21th Century},
  \bibinfo{publisher}{World Scientific}, \bibinfo{year}{2001}, pp.
  \bibinfo{pages}{235--246}.
\bibitem[{Vardi and Wolper(1986)}]{VardiW86}
\bibinfo{author}{M.~Y. Vardi}, \bibinfo{author}{P.~Wolper},
\newblock \bibinfo{title}{An automata-theoretic approach to automatic program
  verification (preliminary report)},
\newblock in: \bibinfo{booktitle}{Proc.\ of the Symposium on Logic in Computer
  Science (LICS'86)}, \bibinfo{year}{1986}, pp. \bibinfo{pages}{332--344}.
\bibitem[{B{\"o}rger et~al.(1997)B{\"o}rger, Gr{\"a}del, and
  Gurevich}]{Borgeretal97}
\bibinfo{author}{E.~B{\"o}rger}, \bibinfo{author}{E.~Gr{\"a}del},
  \bibinfo{author}{Y.~Gurevich}, \bibinfo{title}{The Classical Decision
  Problem}, Perspectives in Mathematical Logic, \bibinfo{publisher}{Springer},
  \bibinfo{year}{1997}.
\bibitem[{Jech(1997)}]{DBLP:books/daglib/0090259}
\bibinfo{author}{T.~Jech}, \bibinfo{title}{Set theory, Second Edition},
  Perspectives in Mathematical Logic, \bibinfo{publisher}{Springer},
  \bibinfo{year}{1997}.
\bibitem[{Hajnal and Hamburger(1999)}]{Hajnal99}
\bibinfo{author}{A.~Hajnal}, \bibinfo{author}{P.~Hamburger},
  \bibinfo{title}{Set Theory (London Mathematical Society Student Texts)},
  \bibinfo{publisher}{Cambridge: Cambridge University Press},
  \bibinfo{year}{1999}.
\bibitem[{Lutz and Wolter(2017)}]{DBLP:journals/lmcs/LutzW17}
\bibinfo{author}{C.~Lutz}, \bibinfo{author}{F.~Wolter},
\newblock \bibinfo{title}{The data complexity of description logic ontologies},
\newblock \bibinfo{journal}{Logical Methods in Computer Science}
  \bibinfo{volume}{13} (\bibinfo{year}{2017}).
\bibitem[{Beck et~al.(2015)Beck, Dao{-}Tran, Eiter, and
  Fink}]{DBLP:conf/aaai/BeckDEF15}
\bibinfo{author}{H.~Beck}, \bibinfo{author}{M.~Dao{-}Tran},
  \bibinfo{author}{T.~Eiter}, \bibinfo{author}{M.~Fink},
\newblock \bibinfo{title}{{LARS:} a logic-based framework for analyzing
  reasoning over streams},
\newblock in: \bibinfo{booktitle}{Proc.\ of the 29th {AAAI} Conf.\ on
  Artificial Intelligence, AAAI 2015}, \bibinfo{year}{2015}, pp.
  \bibinfo{pages}{1431--1438}.
\bibitem[{Kharlamov et~al.(2019)Kharlamov, Mehdi, Savkovic, Xiao, Kalayci, and
  Roshchin}]{DBLP:journals/ws/KharlamovMSXKR19}
\bibinfo{author}{E.~Kharlamov}, \bibinfo{author}{G.~Mehdi},
  \bibinfo{author}{O.~Savkovic}, \bibinfo{author}{G.~Xiao},
  \bibinfo{author}{E.~G. Kalayci}, \bibinfo{author}{M.~Roshchin},
\newblock \bibinfo{title}{Semantically-enhanced rule-based diagnostics for
  industrial {I}nternet of {T}hings: The {SDRL} language and case study for
  {S}iemens trains and turbines},
\newblock \bibinfo{journal}{J. Web Semant.} \bibinfo{volume}{56}
  (\bibinfo{year}{2019}) \bibinfo{pages}{11--29}.
\bibitem[{Brandt et~al.(2018)Brandt, Kalayci, Ryzhikov, Xiao, and
  Zakharyaschev}]{brandt2017ontology}
\bibinfo{author}{S.~Brandt}, \bibinfo{author}{E.~G. Kalayci},
  \bibinfo{author}{V.~Ryzhikov}, \bibinfo{author}{G.~Xiao},
  \bibinfo{author}{M.~Zakharyaschev},
\newblock \bibinfo{title}{Querying log data with metric temporal logic},
\newblock \bibinfo{journal}{J. Artif. Intell. Res.} \bibinfo{volume}{62}
  (\bibinfo{year}{2018}) \bibinfo{pages}{829--877}.
\bibitem[{Brandt et~al.(2019)Brandt, Calvanese, Kalayc{\i}, Kontchakov,
  M{\"o}rzinger, Ryzhikov, Xiao, and Zakharyaschev}]{TIME19}
\bibinfo{author}{S.~Brandt}, \bibinfo{author}{D.~Calvanese},
  \bibinfo{author}{E.~G. Kalayc{\i}}, \bibinfo{author}{R.~Kontchakov},
  \bibinfo{author}{B.~M{\"o}rzinger}, \bibinfo{author}{V.~Ryzhikov},
  \bibinfo{author}{G.~Xiao}, \bibinfo{author}{M.~Zakharyaschev},
\newblock \bibinfo{title}{Two-dimensional rule language for querying sensor log
  data: a framework and use cases},
\newblock in: \bibinfo{booktitle}{Proc.\ of the 26th Int.\ Symposium on
  Temporal Representation and Reasoning, TIME 19}, volume
  \bibinfo{volume}{147}, \bibinfo{publisher}{Dagstuhl Publishing},
  \bibinfo{year}{2019}, pp. \bibinfo{pages}{7:1--7:15}.

\end{thebibliography}

\end{document}